\documentclass[11pt]{article}
\usepackage{inputenc}
\usepackage[authoryear]{natbib}
\usepackage{geometry}
\usepackage{setspace}
\usepackage{amsmath}
\usepackage{amsfonts}
\usepackage{amssymb}
\usepackage[normalem]{ulem}
\usepackage{multirow}
\usepackage{array}
\usepackage{rotating}
\usepackage{version}
\usepackage{graphicx}
\usepackage{url}
\usepackage{color}
\usepackage{float}
\usepackage{caption}
\usepackage{subcaption}

\numberwithin{equation}{section}
\newtheorem{theorem}{Theorem}
\newtheorem{remark}{Remark}
\newtheorem{assumption}{Assumption}

\newenvironment{proof}[1][Proof]{\noindent\textbf{#1.} }{\ \rule{0.5em}{0.5em}}
\input{tcilatex}
\begin{document}

\title{Identification and Estimation of Categorical Random Coefficient Models\thanks{We would like to thank Timothy Armstrong, Hidehiko
Ichimura, Esfandiar Maasoumi, Geert Ridder, Ron Smith and Hayun Song for helpful
comments, and the editor and two anonymous referees for constructive comments and
suggestions.}
}
\author{ Zhan Gao\thanks{
Ph.D. student, Department of Economics, University of Southern California,
3620 South Vermont Avenue, Los Angeles, CA 90089, USA. Email: \texttt{%
zhangao@usc.edu}.} \and M. Hashem Pesaran\thanks{%
1. Department of Economics, University of Southern California, 3620 South
Vermont Avenue, Los Angeles, CA 90089, USA. Email: \texttt{pesaran@usc.edu}.
2. Trinity College, Cambridge, United Kingdom} }
\maketitle

\begin{abstract}
This paper proposes a linear categorical random coefficient model, in which
the random coefficients follow parametric categorical distributions. The
distributional parameters are identified based on a linear recurrence
structure of moments of the random coefficients. A Generalized Method of
Moments estimation procedure is proposed also employed by Peter Schmidt and
his coauthors to address heterogeneity in time effects in panel data models.
Using Monte Carlo simulations, we find that moments of the random
coefficients can be estimated reasonably accurately, but large samples are
required for estimation of the parameters of the underlying categorical
distribution. The utility of the proposed estimator is illustrated by
estimating the distribution of returns to education in the U.S. by gender
and educational levels. We find that rising heterogeneity between
educational groups is mainly due to the increasing returns to education for
those with postsecondary education, whereas within group heterogeneity has
been rising mostly in the case of individuals with high school or less
education.
\end{abstract}

\bigskip{}



\thispagestyle{empty}

\newpage \setcounter{page}{1}

\section{Introduction}

Random coefficient models have been used extensively in time series,
cross-section and panel regressions. \citet{nicholls198516} consider the
estimation of first and second moments of the random coefficient $\beta _{i}$
and the error term $u_{i}$, in a linear regression model. In the seminal
work, \citet{beran1992estimating} establish the conditions of identifying
and estimating the distribution of $\beta _{i}$ and $u_{i}$
non-parametrically. The baseline linear univariate regression in %
\citet{beran1992estimating} has been extended in non-parametric framework by %
\citet*{beran1993semiparametric, beran1994minimum, beran1996nonparametric,
hoderlein2010analyzing, hoderlein2017triangular} and %
\citet{breunig2018specification}, to just name a few. \citet{hsiao2008random}
survey random coefficient models in linear panel data models.

In some econometric applications, %
\citet{hausman1981exact,hausman1995nonparametric,foster_hahn2000consistent}
for examples, the main interest is to estimate the consumer surplus
distribution based on a linear demand system where the coefficient
associated with the price is random. In such settings, the distribution of
the random coefficients is needed when computing the consumer surplus
function, and the non-parametric estimation is more general, flexible and
suitable for the purpose. On the other hand, parametric models may be
favored in applications in which the implied economic meaning of the
distribution of the random coefficients is of interests. Examples include
estimation of the return to education %
\citep{lemieux2006postnber,lemieux2006postsecondary} and the labor supply
equation \citep*{bick2020hours}.

In this paper, we consider a linear regression model with a random
coefficient $\beta _{i}$ that is assumed to follow a categorical
distribution, i.e. $\beta _{i}$ has a discrete support $\left\{
b_{1},b_{2},\cdots ,b_{K}\right\} $, and $\beta _{i}=b_{k}$ with probability 
$\pi _{k}$. The discretization of the support of the random coefficient $%
\beta _{i}$ naturally corresponds to the interpretation that each individual
belongs to a certain category, or group, $k$ with probability $\pi _{k}$.
Compared to a non-parametric distribution with continuous support, assuming
a categorical distribution allows us not only to model the heterogeneous
responses across individuals but also to interpret the results with sharper
economic meaning. As we will illustrate in the empirical application in
Section \ref{sec:Empirical-Application}, it is hard to clearly interpret the
distribution of returns to education without imposing some form of
parametric restrictions.

In addition, with the categorical distribution imposed, the identification
and estimation of the distribution of $\beta _{i}$ do not rely on
identically distributed error terms $u_{i}$ and regressors $\mathbf{w}_i$,
as shown in Section \ref{sec:Identification} and \ref{sec:Estimation}.
Heterogeneously generated errors can be allowed, which is important in many
empirical applications. To the best of our knowledge, this is the first
identification result in linear random coefficient model without a strict
IID setting.

The identification of the distribution of $\beta _{i}$ is established in
this paper based on the identification of the moments of $\beta _{i}$, which
coincides with the identification condition in \citet{beran1992estimating}
that the distribution of $\beta _{i}$ is uniquely determined by its moments,
which is assumed to exist up to an arbitrary order. Since under our setup
the distribution of $\beta _{i}$ is parametrically specified, the moments of 
$\beta _{i}$ exist and can be derived explicitly. The parameters of the
assumed categorical distribution can then be uniquely determined by a system
of equations in terms of the moments, as in Theorem \ref%
{prop:identification_of_beta_L_H}. The parameters of the categorical
distribution are then estimated consistently by the generalized method of
moments (GMM). The estimation procedure based on moment conditions shares
similar spirits as in \citet*{ahn2001gmm,ahn2013panel} in which Peter
Schmidt and coauthors study panel data models with interactive effects where
they allow for the time effects to vary across individual units. Comparing
to alternative non-parametric random coefficient models, the standard GMM
estimation is easy to implement, and the identified categorical structure
has a clear economic interpretation.

Using Monte Carlo (MC) simulations, we find that moments of the random
coefficients can be estimated reasonably accurately, but large samples are
required for estimation of the parameters of the underlying categorical
distributions. Our theoretical and MC results also suggest that our method
is suitable when the number heterogeneous coefficients and the number of
categories are small (2 or 3). With the number of categories rising the
burden on identification from the moments to the parameters of the
categorical distribution also rises rapidly. The quality of identification
also deteriorates as we need to rely on higher and higher moments to
identify a larger number of categories, since the information content of the
moments tend to decline with their order.

The proposed method is also illustrated by providing estimates of the
distribution of returns to education in the U.S. by gender and educational
levels, using the May and Outgoing Rotation Group (ORG) supplements of the
Current Population Survey (CPS) data. Comparing the estimates obtained over
the sub-periods 1973-75 and 2001-03, we find that rising between group
heterogeneity is largely due to rising returns to education in the case of
individuals with postsecondary education, whilst within group heterogeneity
has been rising in the case of individuals with high school or less
education.

\textbf{Related Literature:} This paper draws mainly upon the literature of
random coefficient models. As already mentioned, the main body of the recent
literature is focused on non-parametric identification and estimation.
Following \citet{beran1992estimating}, \citet{beran1993semiparametric} and %
\citet{beran1994minimum} extend the model to a linear semi-parametric model
with a multivariate setup and propose a minimum distance estimator for the
unknown distribution. \citet{foster_hahn2000consistent} extend the
identification results in \citet{beran1992estimating} and apply the minimum
distance estimator to a gasoline consumption data to estimate the consumer
surplus function. \citet*{beran1996nonparametric} and \citet*{%
hoderlein2010analyzing} propose kernel density estimators based on the Radon
inverse transformation in linear models.

In addition to linear models, \citet{ichimura1998maximum} and %
\citet{gautier2013nonparametric} incorporate the random coefficients in
binary choice models. \citet{gautier2011triangular} and \citet*{%
hoderlein2017triangular} consider triangular models with random coefficients
allowing for causal inference. \citet{matzkin2012identification} and %
\citet{masten2018random} discuss the identification of random coefficients
in simultaneous equation models. \citet{breunig2018specification} propose a
general specification test in a variety of random coefficient models. Random
coefficients are also widely studied in panel data models, for example %
\citet{hsiao2008random} and \citet{arellano2012identifying}

\medskip

The rest of the paper is organized as follows: Section \ref%
{sec:Identification} establishes the main identification results. The GMM
estimation procedure is proposed and discussed in Section \ref%
{sec:Estimation}. An extension to a multivariate setting is considered in
Section \ref{sec:Extensions}. Small sample properties of the proposed
estimator are investigated in Section \ref{sec:Monte-Carlo-Simulation},
using Monte Carlo techniques under different regressor and error
distributions. Section \ref{sec:Empirical-Application} presents and
discusses our empirical application to the return to education. Section \ref%
{sec:Conclusion} provides some concluding remarks and suggestions for future
work. Technical proofs are given in Appendix \ref{sec:Proofs}. 

\medskip

\noindent \textbf{Notations:} Largest and smallest eigenvalues of the $%
p\times p$ matrix $\mathbf{A}=\left( a_{ij}\right) $ are denoted by $\lambda
_{\max }\left( \mathbf{A}\right) $ and $\lambda _{\min }\left( \mathbf{A}%
\right) ,$ respectively, its spectral norm by $\left\Vert \mathbf{A}%
\right\Vert =\lambda _{\max }^{1/2}\left( \mathbf{A}^{\prime }\mathbf{A}%
\right) $, $\mathbf{A}\succ 0$ means that $\mathbf{A}$ is positive definite, 
$\mathrm{vech\left( \mathbf{A}\right) }$ denotes the vectorization of
distinct elements of $\mathbf{A}$, $\mathbf{0}$ denotes zero matrix (or
vector). For $\mathbf{a}\in \mathbb{R}^{p}$, $\mathrm{diag}\left( \mathbf{a}%
\right) $ represents the diagonal matrix with elements of $%
a_{1},a_{2},\cdots ,a_{p}$. For random variables (or vectors) $u$ and $v$, $%
u\perp v$ represents $u$ is independent of $v$. We use $c$ ($C$) to denote
some small (large) positive constants. For a differentiable real-valued
function $f\left( \mathbf{\theta }\right) $, $\nabla _{\mathbf{\theta }%
}f\left( \mathbf{\theta }\right) $ denotes the gradient vector. Operator $%
\rightarrow _{p}$ denotes convergence in probability, and $\rightarrow _{d}$
convergence in distribution. The symbols $O(1)$, and $O_{p}(1)$ denote
asymptotically bounded deterministic and random sequences, respectively.

\section{Categorical random coefficient model\label{sec:Identification}}

We suppose the single cross-section observations, $\left\{ y_{i},x_{i},%
\mathbf{z}_{i}\right\} _{i=1}^{n}$, follow the categorical random
coefficient model 
\begin{equation}
y_{i}=x_{i}\beta _{i}+\mathbf{z}_{i}^{\prime }\mathbf{\gamma }+u_{i},
\label{eq:basic_model}
\end{equation}%
where $y_{i},x_{i}\in \mathbb{R}$, $\mathbf{z}_{i}\in \mathbb{R}^{p_z},$ and 
$\beta _{i}\in \left\{ b_{1},b_{2},\cdots ,b_{K}\right\} $ admits the
following $K$-categorical distribution, 
\begin{equation}
\beta _{i}=%
\begin{cases}
b_{1}, & \text{w.p. }\pi _{1}, \\ 
b_{2}, & \text{w.p. }\pi _{2}, \\ 
\vdots & \vdots \\ 
b_{K}, & \text{w.p. }\pi _{K},%
\end{cases}
\label{eq:category_dist}
\end{equation}%
w.p. denotes ``with probability'', $\pi_{k}\in \left( 0,1\right) $, $%
\sum_{k=1}^{K}\pi _{k}=1$, $b_{1}<b_{2}<\cdots <b_{K}$, $\mathbf{\gamma }\in 
\mathbb{R}^{p_z}$ is homogeneous and $\mathbf{z}_{i}$ could include an
intercept term as its first element. It is assumed that $\beta _{i}\perp 
\mathbf{w}_i = \left( x_{i},\mathbf{z}_{i}^{\prime }\right) ^{\prime }$, and
the idiosyncratic errors $u_{i}$ are independently distributed with mean $0$.

\begin{remark}
The model can be extended to allow $\mathbf{x}_{i},\mathbf{\beta }_{i}\in 
\mathbb{R}^{p}$, with $\mathbf{\beta }_{i}$ following a multivariate
categorical distribution, though with more complicated notations. We will
consider possible extensions in Section \ref{sec:Extensions}.
\end{remark}

\begin{remark}
Since we consider a pure cross-sectional setting, the key assumption that $%
\beta _{i}$ and $x_{i}$ are independently distributed cannot be relaxed.
Allowing $\beta _{i}$ to vary with $w_{i}$, without any further
restrictions, is tantamount to assuming $y_{i}$ is a general function of $%
w_{i}$, in effect rendering a nonparametric specification.
\end{remark}

\begin{remark}
The number of categories $K$ is assumed to be fixed and known. Conditions $%
\sum_{k=1}^{K}\pi _{k}=1$, $b_{1}<b_{2}<\cdots <b_{K},$ and $\pi _{k}\in
\left( 0,1\right) $ together are sufficient for the existence of $K$
categories. For example, if $b_{k}=b_{k^{\prime }}$, then we can merge
categories $k$ and $k^{\prime }$, and the number of categories reduces to $%
K-1$. Similarly, if $\pi _{k}=0$ for some $k$, then category $k$ can be
deleted, and the number of categories is again reduced to $K-1$. Information
criteria can be used to determine $K$, but this will not be pursued in this
paper. Model specification tests could also be considered. See, for
examples, \citet{andrews2001testing} and \citet{breunig2018specification}.
\end{remark}

In the rest of this section, we focus on the model (\ref{eq:basic_model})
and establish the conditions under which the distribution of $\beta _{i}$ is
identified.

\subsection{Identifying the moments of $\protect\beta _{i}$\label%
{subsec:identify_moments_beta}}

\begin{assumption}
\label{assu:identification_regularity_condition}

\begin{enumerate}
\item[(a)] (i) $u_{i}$ is distributed independently of $\mathbf{w}%
_{i}=\left( x_{i},\mathbf{z}_{i}^{\prime }\right) ^{\prime }$ and $\beta
_{i} $. (ii) $\sup_{i}\mathrm{E}\left( \left\vert u_{i}^{r}\right\vert
\right) <C$, $r=1,2,\cdots ,2K-1$. (iii) $n^{-1}\sum_{i=1}^{n} u_i^4 =
O_p(1) $.

\item[(b)] (i) Let $\mathbf{Q}_{n,ww}=n^{-1}\sum_{i=1}^{n}\mathbf{w}_{i}%
\mathbf{w}_{i}^{\prime }$, and $\mathbf{q}_{n,wy}=n^{-1}\sum_{i=1}^{n}%
\mathbf{w}_{i}y_{i}$. Then $\left\Vert \mathrm{E}\left( \mathbf{Q}%
_{n,ww}\right) \right\Vert <C<\infty $, and $\left\Vert \mathrm{E}\left( 
\mathbf{q}_{n,wy}\right) \right\Vert <C<\infty $, and there exists $n_{0}\in 
\mathbb{N}$ such that for all $n\geq n_{0}$, 
\begin{equation*}
0<c<\lambda _{\min }\left( \mathbf{Q}_{n,ww}\right) <\lambda _{\max }\left( 
\mathbf{Q}_{n,ww}\right) <C<\infty .
\end{equation*}%
(ii) $\sup_{i} \mathrm{E}\left( \left\Vert \mathbf{w}_{i}\right\Vert
^{r}\right) <C<\infty $, $r=1,2,\cdots ,4K-2$.\newline
(iii) $n^{-1} \sum_{i=1}^{n} \left\Vert \mathbf{w}_{i} \right\Vert^{4} =
O_p(1)$.

\item[(c)] $\left\Vert \mathbf{Q}_{n,ww}-\mathrm{E} \left( \mathbf{Q}%
_{n,ww}\right) \right\Vert =O_p\left( n^{-1/2}\right)$, $\left\Vert \mathbf{q%
}_{n,wy}- \mathrm{E} \left( \mathbf{q}_{n,wy}\right) \right\Vert =O_p\left(
n^{-1/2}\right)$, and 
\begin{equation*}
\mathrm{E} \left( \mathbf{Q}_{n,ww}\right) =n^{-1}\sum_{i=1}^{n}\mathrm{E}
\left( \mathbf{w}_{i}\mathbf{w}_{i}^{\prime }\right) \succ 0.
\end{equation*}

\item[(d)] $\left\Vert \mathrm{E} \left( \mathbf{Q}_{n,ww}\right) -\mathbf{Q}%
_{ww}\right\Vert =O\left( n^{-1/2}\right) $, $\left\Vert \mathrm{E} \left( 
\mathbf{q}_{n,wy}\right) -\mathbf{q}_{wy}\right\Vert =O\left(
n^{-1/2}\right) $, where $\mathbf{q}_{wy} = \lim\limits_{n \to \infty} 
\mathrm{E} \left( \mathbf{q}_{n, wy} \right)$, $\mathbf{Q}_{ww} =
\lim\limits_{n \to \infty} \mathrm{E} \left( \mathbf{Q}_{n, ww} \right) $
and $\mathbf{Q}_{ww}\succ 0$.
\end{enumerate}
\end{assumption}

\begin{remark}
Part (a) of Assumption \ref{assu:identification_regularity_condition}
relaxes the assumption that $u_{i}$ is identically distributed, and allows
for heterogeneously generated errors. For identification of the distribution
of $\beta _{i}$, we require $u_{i}$ to be distributed independently of $%
\mathbf{w}_{i}$ and $\beta _{i}$, which rules out conditional
heteroskedasticity. However, estimation and inference involving $\mathrm{E}%
\left( \beta _{i}\right) $ and $\mathbf{\gamma }$ can be carried out in
presence of conditionally error heteroskedastic, as shown in Theorem \ref%
{lem:gamma_est_consistency}. Parts (c) and (d) of Assumption \ref%
{assu:identification_regularity_condition} relax the condition that $\mathbf{%
w}_{i}$ is identically distributed across $i$. As we proceed, only $\beta
_{i}$, whose distribution is of interest, is assumed to be IID across $i$,
and it is not required for $\mathbf{w}_{i}$ and $u_{i}$ to be identically
distributed over $i$.
\end{remark}

\begin{remark}
\label{rem:weak_crosssectional_depen_condition} The high level conditions in
Assumption \ref{assu:identification_regularity_condition}, concerning the
convergence in probability of averages such as $Q_{n,ww}=n^{-1}\sum_{i=1}^{n}%
\mathbf{w}_{i}\mathbf{w}_{i}^{\prime }$, can be verified under weak
cross-sectional dependence. Let $f_{i}=f\left( \mathbf{w}_{i},\beta
_{i},u_{i}\right) $ be a generic function of $\mathbf{w}_{i}$, $\beta _{i}$
and $u_{i}$.\footnote{$f_{i}$ is assumed to be a scalar, and we can apply
the analysis element-by-element to a matrix, for example $\mathbf{w}_{i}%
\mathbf{w}_{i}^{\prime }$.} Assume that $\sup_{i}\mathrm{E}\left(
f_{i}^{2}\right) <C$, and $\sup_{j}\sum_{i=1}^{n}\left\vert \mathrm{cov}%
\left( f_{i},f_{j}\right) \right\vert <C$, for some fixed $C<\infty $. Then, 
\begin{equation*}
\mathrm{var}\left( \frac{1}{n}\sum_{i=1}^{n}f_{i}\right) \leq \frac{1}{n^{2}}%
\sum_{i=1}^{n}\sum_{j=1}^{n}\left\vert \mathrm{cov}\left( f_{i},f_{j}\right)
\right\vert \leq \frac{1}{n}\sup_{j}\sum_{i=1}^{n}\left\vert \mathrm{cov}%
\left( f_{i},f_{j}\right) \right\vert \leq \frac{C}{n}.
\end{equation*}%
By Chebyshev's inequality, for any $\varepsilon >0$, we have $M_{\varepsilon
}>\sqrt{C/\varepsilon }$ such that 
\begin{equation*}
\Pr \left( \sqrt{n}\left\vert \frac{1}{n}\sum_{i=1}^{n}\left[ f_{i}-\mathrm{E%
}\left( f_{i}\right) \right] \right\vert >M_{\varepsilon }\right) \leq \frac{%
n\mathrm{var}\left( n^{-1}\sum_{i=1}^{n}f_{i}\right) }{C}\varepsilon \leq
\varepsilon ,
\end{equation*}%
i.e. $n^{-1}\sum_{i=1}^{n}\left[ f_{i}-\mathrm{E}\left( f_{i}\right) \right]
=O_{p}\left( n^{-1/2}\right) $.
\end{remark}

Denote $\mathbf{\phi }_{i}=\left( \beta _{i},\mathbf{\gamma }^{\prime
}\right) ^{\prime }$ and $\mathbf{\phi }=\mathrm{E}\left( \mathbf{\phi }%
_{i}\right) =\left( \mathrm{E}\left( \beta _{i}\right) ,\mathbf{\gamma }%
^{\prime }\right) ^{\prime }$. Consider the moment condition, 
\begin{equation}
\mathrm{E}\left( \mathbf{w}_{i}y_{i}\right) =\mathrm{E}\left( \mathbf{w}_{i}%
\mathbf{w}_{i}^{\prime }\right) \mathbf{\phi },  \label{mc_y1_w1}
\end{equation}%
and sum \eqref{mc_y1_w1} over $i$ 
\begin{equation}
\frac{1}{n}\sum_{i=1}^{n}\mathrm{E}\left( \mathbf{w}_{i}y_{i}\right) =\left[ 
\frac{1}{n}\sum_{i=1}^{n}\mathrm{E}\left( \mathbf{w}_{i}\mathbf{w}%
_{i}^{\prime }\right) \right] \mathbf{\phi }.  \label{mc_y1w1_n}
\end{equation}%
Let $n\to \infty$, then $\mathbf{\phi }$ is identified by 
\begin{equation}
\mathbf{\phi } = \mathbf{Q}_{ww} ^{-1} \mathbf{q}_{wy},  \label{phi}
\end{equation}
under Assumption \ref{assu:identification_regularity_condition}.

\begin{assumption}
\label{assu:conv_moment} Let $\tilde{y}_{i}=y_{i}-\mathbf{z}_{i}^{\prime }%
\mathbf{\gamma }$.

\begin{enumerate}
\item[(a)] {$\left\vert n^{-1}\sum_{i=1}^{n}\mathrm{E}\left( \tilde{y}%
_{i}^{r}x_{i}^{s}\right) -\rho _{r,s}\right\vert =O\left( n^{-1/2}\right) ,$
and $\left\vert \rho _{r,s}\right\vert <\infty ,$ for $r,s=0,1,\cdots ,2K-1$%
. }

\item[(b)] {$\left\vert n^{-1}\sum_{i=1}^{n}\mathrm{E}\left(
u_{i}^{r}\right) -\sigma _{r}\right\vert =O\left( n^{-1/2}\right) ,$ and $%
\left\vert \sigma _{r}\right\vert <\infty $, for $r=2,3,\cdots ,2K-1$. }

\item[(c)] $n^{-1}\sum_{i=1}^{n}\left[ \mathrm{var}(x_{i}^{r})-\left( \rho
_{0,2r}-\rho _{0,r}^{2}\right) \right] =O\left( n^{-1/2}\right) $ where $%
\rho _{0,2r}-\rho _{0,r}^{2}>0,$ for $r=2,3,\cdots ,2K-1$.
\end{enumerate}
\end{assumption}

\begin{remark}
\label{heterogeniety of moments} The above assumption allows for a limited
degree of heterogeneity of the moments. As an example, let $\mathrm{E}\left(
u_{i}^{r}\right) =\sigma _{ir}$ and denote the heterogeneity of the $r^{th}$
moment of $u_{i}$ by $e_{ir}=\sigma _{ir}-\sigma _{r}$. Then 
\begin{equation*}
\left\vert n^{-1}\sum_{i=1}^{n}\mathrm{E}\left( u_{i}^{r}\right) -\sigma
_{r}\right\vert \leq n^{-1}\sum_{i=1}^{n}\left\vert e_{ir}\right\vert ,
\end{equation*}%
and condition (b)\ of Assumption \ref{assu:conv_moment} is met if $\
\sum_{i=1}^{n}\left\vert e_{ir}\right\vert =O(n^{\alpha _{r}})$ with $\alpha
_{r}<1/2$. $\alpha _{r}$ measures the degree of heterogeneity with $\alpha
_{r}=1$ representing the highest degree of heterogeneity. A similar idea is
used by \citet{pesaran2018pool} in their analysis of poolability in panel
data models.
\end{remark}

\begin{theorem}
\label{lem: identification_moments} Under Assumptions \ref%
{assu:identification_regularity_condition} and \ref{assu:conv_moment}, $%
\mathrm{E}\left( \beta _{i}^{r}\right) $ and $\sigma _{r}$, $r=2,3,\cdots
,2K-1$ are identified.
\end{theorem}

\begin{proof}
For $r=2,\cdots ,2K-1$, 
\begin{align}
\mathrm{E}\left( \tilde{y}_{i}^{r}\right) & =\mathrm{E}\left(
x_{i}^{r}\right) \mathrm{E}\left( \beta _{i}^{r}\right) +\mathrm{E}\left(
u_{i}^{r}\right) +\sum_{q=2}^{r-1}\binom{r}{q}\mathrm{E}\left(
x_{i}^{r-q}\right) \mathrm{E}\left( u_{i}^{q}\right) \mathrm{E}\left( \beta
_{i}^{r-q}\right) ,  \label{eq:mc_n_r} \\
\mathrm{E}\left( \tilde{y}_{i}^{r}x_{i}^{r}\right) & =\mathrm{E}\left(
x_{i}^{2r}\right) \mathrm{E}\left( \beta _{i}^{r}\right) +\mathrm{E}\left(
x_{i}^{r}\right) \mathrm{E}\left( u_{i}^{r}\right) +\sum_{q=2}^{r-1}\binom{r%
}{q}\mathrm{E}\left( x_{i}^{2r-q}\right) \mathrm{E}\left( u_{i}^{q}\right) 
\mathrm{E}\left( \beta _{i}^{r-q}\right) .  \label{eq:mc_n_2r}
\end{align}%
where $\binom{r}{q}=\frac{r!}{q!(r-q)!}$ are binomial coefficients, for
non-negative integers $q\leq r$.

Sum over $i$, then by parts (a) and (b) of Assumption \ref{assu:conv_moment}%
, 
\begin{align}
\rho _{0,r}\mathrm{E}\left( \beta _{i}^{r}\right) +\sigma _{r}& =\rho
_{r,0}-\sum_{q=2}^{r-1}\binom{r}{q}\rho _{0,r-q}\sigma _{q}\mathrm{E}\left(
\beta _{i}^{r-q}\right) ,  \label{eq:mc_limit_r} \\
\rho _{0,2r}\mathrm{E}\left( \beta _{i}^{r}\right) +\rho _{0,r}\sigma _{r}&
=\rho _{r,r}-\sum_{q=2}^{r-1}\binom{r}{q}\rho _{0,2r-q}\sigma _{q}\mathrm{E}%
\left( \beta _{i}^{r-q}\right) .  \label{eq:mc_limit_2r}
\end{align}%
Derivation details are relegated to Appendix \ref{sec:Proofs}. By part (c)
of Assumption \ref{assu:conv_moment}, the matrix $%
\begin{pmatrix}
\rho _{0,r} & 1 \\ 
\rho _{0,2r} & \rho _{0,r}%
\end{pmatrix}%
$ is invertible for $r=2,3,\cdots ,2K-1$. As a result, we can sequentially
solve \eqref{eq:mc_limit_r} and \eqref{eq:mc_limit_2r} for $\mathrm{E}\left(
\beta _{i}^{r}\right) $ and $\sigma _{r}$, for $r=2,3,\cdots ,2K-1$.
\end{proof}

\subsection{Identifying the distribution of $\protect\beta _{i}$}

\citet[Theorem 2.1, pp. 1972]{beran1992estimating} prove the identification
of the distribution of the random coefficient, $\beta _{i}$, in a canonical
model without covariates, $z_{i}$, under the condition that the distribution
of $\beta _{i}$ is uniquely determined by its moments. We show the
identification of moments of $\beta_i$ holds more generally when $x_i$ and $%
u_i$ are not identically distributed and the distribution of $\beta_i$ is
identified if it follows a categorical distribution. Note that under (\ref%
{eq:category_dist}), 
\begin{equation}
\mathrm{E}\left( \beta _{i}^{r}\right) =\sum_{k=1}^{K}\pi
_{k}b_{k}^{r},\;r=0,1,2,\cdots ,2K-1,  \label{mbeta}
\end{equation}%
with $\mathrm{E}\left( \beta _{i}^{r}\right) $ identified under Assumption %
\ref{assu:identification_regularity_condition}. To identify $\mathbf{\pi }%
=\left( \pi _{1},\pi _{2},...,\pi _{K}\right) ^{\prime }$ and $\mathbf{b}%
=\left( b_{1},b_{2},...,b_{K}\right) ^{\prime }$, we need to verify that the
system of $2K$ equations in \eqref{mbeta} has a unique solution if $%
b_{1}<b_{2}<\cdots <b_{K}$, and $\pi _{k}\in \left( 0,1\right) $. In the
proof, we construct a linear recurrence relation and make use of the
corresponding characteristic polynomial.

\begin{theorem}
\label{prop:identification_of_beta_L_H} Consider the random coefficient
regression model \eqref{eq:basic_model}, suppose that Assumptions \ref%
{assu:identification_regularity_condition} and \ref{assu:conv_moment} hold.
Then $\mathbf{\theta }=\left( \mathbf{\pi }^{\prime },\mathbf{b}^{\prime
}\right) ^{\prime }$ is identified subject to $b_{1}<b_{2}<\cdots <b_{K}$
and $\pi _{k}\in \left( 0,1\right) $, for all $k=1,2,\cdots ,K$.
\end{theorem}

\begin{proof}
We motivate the key idea of the proof in the special case where $K=2,$ and
relegate the proof of the general case to the Appendix \ref{sec:Proofs}. Let 
$b_{1}=\beta _{L}$, $b_{2}=\beta _{H}$, $\pi _{1}=\pi $ and $\pi _{2}=1-\pi $%
. Note that 
\begin{align}
\mathrm{E}\left( \beta _{i}\right) & =\pi \beta _{L}+\left( 1-\pi \right)
\beta _{H},  \label{id1} \\
\mathrm{E}\left( \beta _{i}^{2}\right) & =\pi \beta _{L}^{2}+\left( 1-\pi
\right) \beta _{H}^{2},  \label{id2} \\
\mathrm{E}\left( \beta _{i}^{3}\right) & =\pi \beta _{L}^{3}+\left( 1-\pi
\right) \beta _{H}^{3},  \label{id3}
\end{align}%
and $\mathrm{E}\left( \beta _{i}^{k}\right) $, $k=1,2,3$ are identified. $%
\left( \pi ,\beta _{L},\beta _{H}\right) $ can be identified if the system
of equations \eqref{id1} to \eqref{id3}, has a unique solution. By %
\eqref{id1}, 
\begin{equation}
\pi =\frac{\beta _{H}-\mathrm{E}\left( \beta _{i}\right) }{\beta _{H}-\beta
_{L}},\; \text{and}\; 1-\pi =\frac{\mathrm{E}\left( \beta _{i}\right) -\beta
_{L}}{\beta _{H}-\beta _{L}}.  \label{pi_expression}
\end{equation}%
Plug \eqref{pi_expression} into \eqref{id2} and \eqref{id3}, 
\begin{align}
\mathrm{E}\left( \beta _{i}\right) \left( \beta _{L}+\beta _{H}\right)
-\beta _{L}\beta _{H}& =\mathrm{E}\left( \beta _{i}^{2}\right) ,
\label{sys_1} \\
\mathrm{E}\left( \beta _{i}^2\right) \left( \beta_L + \beta_H\right) -%
\mathrm{E}\left(\beta_i\right) \beta _{L}\beta _{H} & =\mathrm{E}\left(
\beta _{i}^{3}\right) .  \label{sys_2}
\end{align}
Denote $\beta_{L+H} = \beta_{L} + \beta_{H}$ and $\beta_{LH} =
\beta_{L}\beta_{H}$, and write \eqref{sys_1} and \eqref{sys_2} in matrix
form, 
\begin{equation}
\mathbf{M} \mathbf{D} \mathbf{b}^{\ast} = \mathbf{m},
\label{eq:sys_mat_form}
\end{equation}
where 
\begin{equation*}
\mathbf{M} = 
\begin{pmatrix}
1 & \mathrm{E}\left(\beta_i\right) \\ 
\mathrm{E}\left(\beta_i\right) & \mathrm{E}\left(\beta_i^2 \right)%
\end{pmatrix}%
, \, \mathbf{D} = 
\begin{pmatrix}
-1 & 0 \\ 
0 & 1%
\end{pmatrix}%
,\, \mathbf{b}^{\ast} = 
\begin{pmatrix}
\beta_{LH} \\ 
\beta_{L+H}%
\end{pmatrix}%
, \;\text{and}\; \mathbf{m} = 
\begin{pmatrix}
\mathrm{E}\left(\beta_i^2\right) \\ 
\mathrm{E}\left(\beta_i^3\right)%
\end{pmatrix}%
.
\end{equation*}
Under the conditions $0 < \pi < 1$ and $\beta_H > \beta_L$, 
\begin{equation*}
\det\left( \mathbf{M} \right) = \mathrm{var}\left( \beta _{i}\right) = 
\mathrm{E}\left( \beta _{i}^{2}\right) -\mathrm{E}\left( \beta _{i}\right)
^{2} = \pi\left( 1-\pi \right)\left( \beta_H - \beta_L \right)^2 > 0.
\end{equation*}
As a result, we can solve \eqref{eq:sys_mat_form} for $\beta_{L+H}$ and $%
\beta_{LH}$ as 
\begin{align}
\beta _{L+H}& =\frac{\mathrm{E}\left( \beta _{i}^{3}\right) -\mathrm{E}%
\left( \beta _{i}\right) \mathrm{E}\left( \beta _{i}^{2}\right) }{\mathrm{var%
}\left( \beta_i \right)},  \label{sys_3} \\
\beta _{LH}& =\frac{\mathrm{E}\left( \beta _{i}\right) \mathrm{E}\left(
\beta _{i}^{3}\right) -\mathrm{E}\left( \beta _{i}^{2}\right) ^{2}}{\mathrm{%
var}\left( \beta_i \right)}.  \label{sys_4}
\end{align}%
$\beta _{L}$ and $\beta _{H}$ are solutions to the quadratic equation, 
\begin{equation}
\beta ^{2}-\beta _{L+H}\beta +\beta _{LH}=0.
\end{equation}%
We can verify that $\Delta =\beta _{L+H}^{2}-4\beta _{LH}>0$ by direct
calculation using \eqref{sys_3} and \eqref{sys_4}. Simplifying $\Delta $ in
terms of $\mathrm{E}\left( \beta _{i}^{k}\right) $ and then plugging in %
\eqref{id1}, \eqref{id2} and \eqref{id3}, 
\begin{align*}
\Delta =& \frac{ \left[ \mathrm{E}\left( \beta _{i}^{3}\right) - \mathrm{E}%
\left( \beta _{i}\right) \mathrm{E}\left( \beta _{i}^{2}\right) \right] ^{2}
-4 \mathrm{var}\left( \beta_i \right) \left[ \mathrm{E}\left( \beta
_{i}\right) \mathrm{E}\left(\beta _{i}^{3}\right) - \mathrm{E}\left(
\beta_{i}^{2}\right) ^{2} \right] }{ \left[ \mathrm{var}\left( \beta_i
\right) \right] ^{2} } \\
& = \left( \beta_H - \beta_L \right)^2 > 0.
\end{align*}%
Then, we obtain the unique solutions, 
\begin{align}
\beta _{L}& =\frac{1}{2}\left( \beta _{L+H}-\sqrt{\beta _{L+H}^{2}-4\beta
_{LH}}\right) ,  \label{solution_beta_L} \\
\beta _{H}& =\frac{1}{2}\left( \beta _{L+H}+\sqrt{\beta _{L+H}^{2}-4\beta
_{LH}}\right) ,  \label{solution_beta_H}
\end{align}%
and $\pi $ can be determined by (\ref{pi_expression}) correspondingly.
\end{proof}

\begin{remark}
The key identifying assumption in (\ref{prop:identification_of_beta_L_H}) is
the assumed existence of the strict ordinal relation $b_{1}<b_{2}<\cdots
<b_{K}$ so that $b_{k}$ and $b_{k^{\prime }}$ are not symmetric for $k\neq
k^{\prime }$, and $0<\pi _{k}<1$ so that the distribution of $\beta _{i}$
does not degenerate. When $K=2$, the conditions $b_{1}<b_{2}<\cdots <b_{K}$,
and $\pi _{k}\in \left( 0,1\right) $, are equivalent to $\mathrm{var}\left(
\beta _{i}\right) =\pi _{1}\left( 1-\pi _{1}\right) \left(
b_{2}-b_{1}\right) ^{2}>0$. In other words, not surprisingly, the
categorical distribution of $\beta _{i}$ are identified only if $\mathrm{var}%
\left( \beta _{i}\right) >0$.

In practice, a test for $\mathbb{H}_{0}:\mathrm{var}\left( \beta _{i}\right)
=0$ is possible, by noting that $\mathrm{var}\left( \beta _{i}\right) =0$ is
equivalent to 
\begin{equation*}
\kappa ^{2}=\frac{\mathrm{E}\left( \beta _{i}\right) ^{2}}{\mathrm{E}\left(
\beta _{i}^{2}\right) }=1,
\end{equation*}%
where $\kappa ^{2}$ is well-defined as long as $\beta _{i}\not\equiv 0$. One
important advantage of basing the test of slope homogeneity on $\kappa ^{2}$
rather than on $var(\beta _{i})=0$, is that $\kappa ^{2}\,$is
scale-invariant. $\mathrm{E}\left( \beta _{i}\right) $ and $\mathrm{E}\left(
\beta _{i}^{2}\right) $ are identified as in Section \ref%
{subsec:identify_moments_beta}, whose consistent estimation does not require 
$\mathrm{var}\left( \beta _{i}\right) >0$. Consequently, in principle it is
possible to test slope homogeneity by testing $\mathbb{H}_{0}:\kappa ^{2}=1$%
. However, the problem becomes much more complicated when there are more
than two categories and/or there are more than one regressor under
consideration. A full treatment of testing slope homogeneity in such general
settings is beyond the scope of the present paper.
\end{remark}


\begin{remark}
Note that in the special case of the proof of Theorem \ref%
{prop:identification_of_beta_L_H} where $K=2$, $\beta_{L+H}=\beta_{L}+%
\beta_{H}$ and $\beta_{LH}=\beta_{L}\beta_{H}$ corresponds to the $%
b_{1}^{\ast}$ and $b_{2}^{\ast}$ and \eqref{eq:sys_mat_form} is the same as %
\eqref{eq: linear_recurrence_matrix} when $K = 2$. The special case
illustrates the procedure of identification: identify $\left(b_{k}^{\ast}%
\right)_{k=1}^{K}$ by the moments of $\beta_{i}$, then solve for $%
\left(b_{k}\right)_{k=1}^{K}$ and finally identify $\left(\pi_{k}%
\right)_{k=1}^{K}$.
\end{remark}


\section{Estimation\label{sec:Estimation}}

In this section, we propose a generalized method of moments estimator for
the distributional parameters of $\beta_i$. To reduce the complexity of the
moment equations, we first obtain a $\sqrt{n}$-consistent estimator of $%
\gamma$ and consider the estimation of the distribution of $\beta_i$ by
replacing $\gamma$ by $\hat{\gamma}$.

\subsection{Estimation of $\protect\gamma$\label{subsec:Estimation-of-gamma}}

Let $\mathbf{\phi }=\left( \mathrm{E}\left( \beta _{i}\right) ,\mathbf{%
\gamma }^{\prime }\right) ^{\prime }$, $v_{i}=\beta _{i}-\mathrm{E}\left(
\beta _{i}\right) $ and using the notation in Assumption \ref%
{assu:identification_regularity_condition}, \eqref{eq:basic_model} can be
written as 
\begin{equation}
y_{i}=\mathbf{w}_{i}^{\prime }\mathbf{\phi }+\xi _{i},
\label{eq: model_w_phi}
\end{equation}%
where $\xi _{i}=u_{i}+x_{i}v_{i}$. Then $\mathbf{\phi }$ can be estimated
consistently by $\hat{\mathbf{\phi }}=\mathbf{Q}_{n,ww}^{-1}\mathbf{q}%
_{n,wy} $ where $\mathbf{Q}_{n,ww}$ and $\mathbf{q}_{n,wy}$ are defined in
Assumption \ref{assu:identification_regularity_condition}.

\begin{assumption}
\label{assu:gamma_est_normality} $\left\Vert n^{-1}\sum_{i=1}^{n}\mathrm{E}%
\left( \mathbf{w}_i\mathbf{w}_i^\prime \xi_i^2\right) - \mathbf{V}_{w\xi}
\right\Vert = O\left( n^{-1/2} \right) $, $\mathbf{V}_{w\xi}\succ 0, $ and 
\begin{equation}
\left\Vert \frac{1}{n}\sum_{i=1}^{n} \mathbf{w}_i\mathbf{w}_i^\prime \xi_i^2
- \frac{1}{n}\sum_{i=1}^{n}\mathrm{E}\left( \mathbf{w}_i\mathbf{w}_i^\prime
\xi_i^2\right) \right\Vert = O_p\left( n^{-1/2} \right).
\label{eq:conv_of_wwxi}
\end{equation}
\end{assumption}

\begin{remark}
As in the case of Assumption \ref{assu:identification_regularity_condition},
the high level condition \eqref{eq:conv_of_wwxi} can be shown to hold under
weak cross-sectional dependence, assuming that elements of $\mathbf{w}_{i}%
\mathbf{w}_{i}^{\prime }\xi _{i}^{2}$ are cross-sectionally weakly
correlated over $i$. See Remark \ref{rem:weak_crosssectional_depen_condition}%
.
\end{remark}

\begin{theorem}
\label{lem:gamma_est_consistency} Under Assumption \ref%
{assu:identification_regularity_condition}, $\hat{\mathbf{\phi}}$ is a
consistent estimator for $\mathbf{\phi}$. In addition, under Assumptions \ref%
{assu:identification_regularity_condition} and \ref{assu:gamma_est_normality}%
, as $n\to \infty$, 
\begin{equation}
\sqrt{n}\left( \hat{\mathbf{\phi}} - \mathbf{\phi} \right) \rightarrow_{d}
N\left( \mathbf{0}, \mathbf{V}_\phi \right),
\end{equation}
where $\mathbf{V}_{\phi} = \mathbf{Q}_{w w}^{-1} \mathbf{V}_{w\xi} \mathbf{Q}%
_{w w}^{-1}. $ $\mathbf{V}_{\phi}$ is consistently estimated by 
\begin{equation*}
\hat{\mathbf{V}}_{\phi} = \mathbf{Q}_{n,ww}^{-1}\hat{\mathbf{V}}_{w\xi}%
\mathbf{Q}_{n,ww}^{-1}\to_p \mathbf{V}_{\phi},
\end{equation*}
as $n\to\infty$, where $\hat{\mathbf{V}}_{w\xi} = n^{-1}\sum_{i=1}^{n} 
\mathbf{w}_i\mathbf{w}_i^\prime \hat{\xi}_{i}^2$, and $\hat{\xi}_i = y_i - 
\mathbf{w}_i^\prime \hat{\mathbf{\phi}}$.
\end{theorem}

The proof of Theorem \ref{lem:gamma_est_consistency} is provided in Section %
\ref{suppsec:Proofs} in the online supplement.

\subsection{Estimation of the distribution of $\protect\beta_i$\label%
{subsec:estimation_beta}}

Denote the moments of $\beta _{i}$ on the right-hand side of (\ref{mbeta})
by 
\begin{equation*}
\mathbf{m}_{\beta }=(m_{1},m_{2},...,m_{2K-1})^{\prime }=\left[ \mathrm{E}%
\left( \beta _{i}^{r}\right) \right] _{r=1}^{2K-1}\in \Theta _{m}\subset
\left\{ \mathbf{m}_{\beta }\in \mathbb{R}^{2K-1}:m_{r}\geq 0,\text{ }r\text{
is even}\right\} ,
\end{equation*}%
and note that 
\begin{equation}
\mathbf{m}_{\beta }=\left( 
\begin{array}{c}
m_{1} \\ 
m_{2} \\ 
\vdots \\ 
m_{2K-1}%
\end{array}%
\right) =\left( 
\begin{array}{cccc}
b_{1} & b_{2} & \cdots & b_{K} \\ 
b_{1}^{2} & b_{2}^{2} & \cdots & b_{K}^{2} \\ 
\vdots & \vdots & \vdots & \vdots \\ 
b_{1}^{2K-1} & b_{2}^{2K-1} & \cdots & b_{K}^{2K-1}%
\end{array}%
\right) \left( 
\begin{array}{c}
\pi _{1} \\ 
\pi _{2} \\ 
\vdots \\ 
\pi _{K}%
\end{array}%
\right) ,  \label{eq:mbeta_mat}
\end{equation}%
so in general we can write $\mathbf{m}_{\beta }\triangleq h\left( \mathbf{%
\theta }\right) ,$ where $\mathbf{\theta }=\left( \mathbf{\pi }^{\prime },%
\mathbf{b}^{\prime }\right) ^{\prime }\in \Theta $, and $\mathbf{\theta }$
can be uniquely determined in terms of $\mathbf{m}_{\beta }$ by Theorem \ref%
{prop:identification_of_beta_L_H}. To estimate $\mathbf{\theta }$, we
consider moment conditions following a similar procedure as in Section \ref%
{sec:Identification}, and propose a generalized method of moments (GMM)
estimator.

We consider the following moment conditions 
\begin{equation*}
\mathrm{E}\left( \tilde{y}_{i}^{r}\right) =\sum_{q=0}^{r}\binom{r}{q}\mathrm{%
E}\left( x_{i}^{r-q}\right) \mathrm{E}\left( u_{i}^{q}\right) m_{r-q},\text{ 
}
\end{equation*}%
and 
\begin{equation}
\mathrm{E}\left( \tilde{y}_{i}^{r}x_{i}^{s_{r}}\right) =\sum_{q=0}^{r}\binom{%
r}{q}\mathrm{E}\left( x_{i}^{r-q+s_{r}}\right) \mathrm{E}\left(
u_{i}^{q}\right) m_{r-q},  \label{eq:mc_yrxs}
\end{equation}%
where $\mathrm{E}\left( u_{i}\right) =0$, $\tilde{y}_{i}=y_{i}-\mathbf{z}%
_{i}^{\prime }\mathbf{\gamma }$, $r=1,2,...,2K-1$, and $s_{r}=0,1,\cdots
,S-r $, where $S$ is a user-specific tuning parameter, chosen such that the
highest order moments of $x_{i}$ included is at most $S$, where $S>2K-1$. 
\footnote{%
For identification, we require the moments of $x_{i}$ to exist up to order $%
4K-2$. $S$ can take values between $2K$ to $4K-2$. In practice, the choice
of $S$ affects the trade-off between bias and efficiency.}

Let $\sigma _{0}=1$ and $\sigma _{1}=0$ such that $\sigma _{r}$ is
well-defined for $r=0,1,\cdots ,2K-1$.\textbf{\ }Sum \eqref{eq:mc_yrxs} over 
$i$ and rearrange terms, 
\begin{align}
0& =\sum_{q=0}^{r}\binom{r}{q}\left[ \frac{1}{n}\sum_{i=1}^{n}\mathrm{E}%
\left( x_{i}^{r-q+s_{r}}\right) \mathrm{E}\left( u_{i}^{q}\right) \right]
m_{r-q}-\frac{1}{n}\sum_{i=1}^{n}\mathrm{E}\left( \tilde{y}%
_{i}^{r}x_{i}^{s_{r}}\right)  \notag \\
& =\sum_{q=0}^{r}\binom{r}{q} \left[ \frac{1}{n}\sum_{i=1}^{n}\mathrm{E}%
\left( x_{i}^{r-q+s_{r}}\right) \right] \sigma_{q}m_{r-q} - \frac{1}{n}%
\sum_{i=1}^{n}\mathrm{E}\left( \tilde{y}_{i}^{r}x_{i}^{s_{r}}\right)+\delta
_{n}^{(r,s_{r})},  \label{eq:mc_n}
\end{align}%
where 
\begin{equation*}
\delta _{n}^{(r,s_{r})} =\sum_{q=0}^{r}\binom{r}{q}\left[ \frac{1}{n}%
\sum_{i=1}^{n}\mathrm{E}\left( x_{i}^{r-q+s_{r}}\right) \left[ \mathrm{E}%
\left( u_{i}^{q}\right) -\sigma _{q}\right] \right] m_{r-q} = O\left(
n^{-1/2} \right),
\end{equation*}
as shown in the proof of Theorem \ref{lem: identification_moments}.

Taking $n\to \infty$ in \eqref{eq:mc_n}, 
\begin{equation}
\sum_{q=0}^{r}\binom{r}{q}\rho _{0,r-q+s_{r}}\sigma _{q}m_{r-q}-\rho
_{r,s_{r}} = 0,  \label{eq:mc_limit}
\end{equation}
by Assumption \ref{assu:conv_moment}. We stack the left-hand side of %
\eqref{eq:mc_limit} over $r=1,2,...,2K-1$, and $s_{r}=0,1,\cdots ,S-r$ and
transform $\mathbf{m}_\beta = h\left( \mathbf{\theta} \right)$ to get $%
\mathbf{g}_0\left(\mathbf{\theta }, \mathbf{\sigma}, \mathbf{\gamma }
\right) $.

To implement the GMM estimation we replace $\tilde{y}_{i}$, by $\hat{\tilde{y%
}}_{i}=y_{i}-\mathbf{z}_{i}^{\prime }\hat{\mathbf{\gamma }}$, and $\rho
_{r,s_{r}}$ by $n^{-1}\sum_{i=1}^{n}\hat{\tilde{y}}_{i}^{r}x_{i}^{s_{r}}$.
Noting that $\mathbf{m}_{\beta }=h\left( \mathbf{\theta }\right) $, denote
the sample version of the left-hand side of \eqref{eq:mc_limit} by 
\begin{equation}
\hat{g}_{n}^{(r,s_{r})}\left( \mathbf{\theta },\mathbf{\sigma },\hat{\mathbf{%
\gamma }}\right) =\frac{1}{n}\sum_{i=1}^{n}\hat{g}_{i}^{(r,s_{r})}\left( 
\mathbf{\theta },\mathbf{\sigma },\hat{\mathbf{\gamma }}\right) ,
\label{gbarr}
\end{equation}%
where%
\begin{equation*}
\hat{g}_{i}^{\left( r,s_{r}\right) }\left( \mathbf{\theta },\mathbf{\sigma },%
\hat{\mathbf{\gamma }}\right) =\sum_{q=0}^{r}\binom{r}{q}x_{i}^{r-q+s_{r}}%
\sigma _{q}\left[ h\left( \mathbf{\theta }\right) \right] _{r-q}-\hat{\tilde{%
y}}_{i}^{r}x_{i}^{s_{r}},
\end{equation*}%
and $\mathbf{\sigma }=\left( \sigma _{2},\sigma _{3},\cdots ,\sigma
_{2K-1}\right) ^{\prime }$. Stack the equations in (\ref{gbarr}), over $%
r=0,1,...,2K-1$ and $s_{r}=0,1,\cdots ,S-r$ ($S>2K-1$), in vector notations
we have 
\begin{equation}
\mathbf{\hat{g}}_{n}\left( \mathbf{\theta },\mathbf{\sigma },\hat{\mathbf{%
\gamma }}\right) =\frac{1}{n}\sum_{i=1}^{n}\mathbf{\hat{g}}_{i}\left( 
\mathbf{\theta },\mathbf{\sigma },\hat{\mathbf{\gamma }}\right) .  \label{gn}
\end{equation}%
Given $\hat{\mathbf{\gamma }}$, the GMM estimator of $\left( \mathbf{\theta }%
^{\prime },\mathbf{\sigma }^{\prime }\right) ^{\prime }$ is now computed as 
\begin{equation*}
\left( \hat{\mathbf{\theta }}^{\prime },\hat{\mathbf{\sigma }}^{\prime
}\right) ^{\prime }=\arg \min_{\mathbf{\theta }\in \Theta ,\mathbf{\sigma }%
\in \mathcal{S}}\hat{\Phi}_{n}\left( \mathbf{\theta },\mathbf{\sigma },\hat{%
\mathbf{\gamma }}\right) ,
\end{equation*}%
where $\hat{\Phi}_{n}=\mathbf{\hat{g}}_{n}\left( \mathbf{\theta },\mathbf{%
\sigma },\hat{\mathbf{\gamma }}\right) ^{\prime }\mathbf{A}_{n}\mathbf{\hat{g%
}}_{n}\left( \mathbf{\theta },\mathbf{\sigma },\hat{\mathbf{\gamma }}\right) 
$, and $\mathbf{A}_{n}$ is a positive definite matrix. We follow the GMM
literature using the following choice of $\mathbf{A}_{n}$, 
\begin{equation}
\mathbf{\hat{A}}_{n}=\left[ \frac{1}{n}\sum_{i=1}^{n}\mathbf{\hat{g}}%
_{i}\left( \tilde{\mathbf{\theta }},\tilde{\mathbf{\sigma }},\hat{\mathbf{%
\gamma }}\right) \mathbf{\hat{g}}_{i}\left( \tilde{\mathbf{\theta }},\tilde{%
\mathbf{\sigma }},\hat{\mathbf{\gamma }}\right) ^{\prime }-\mathbf{\bar{g}}%
_{n}\mathbf{\bar{g}}_{n}^{\prime }\right] ^{-1},
\label{eq:efficient_weight_mat}
\end{equation}%
where $\mathbf{\bar{g}}_{n}=\frac{1}{n}\sum_{i=1}^{n}\mathbf{\hat{g}}%
_{i}\left( \tilde{\mathbf{\theta }},\tilde{\mathbf{\sigma }},\hat{\mathbf{%
\gamma }}\right) $, and $\tilde{\mathbf{\theta }}$ and $\tilde{\mathbf{%
\sigma }}$ are preliminary estimators.

\begin{assumption}
\label{assu:Consistency Assumption} Denote the true values of $\mathbf{\theta%
}$, $\mathbf{\sigma}$ and $\mathbf{\gamma}$ by $\mathbf{\theta }_{0}$, $%
\mathbf{\sigma}_0$ and $\mathbf{\gamma}_0 $.

\begin{enumerate}
\item[(a)] $\Theta $ and $\mathcal{S}$ are compact. $\mathbf{\theta }_{0}\in 
\mathrm{int}\left( \Theta \right)$ and $\mathbf{\sigma}_0 \in \mathrm{int}%
\left( \mathcal{S} \right)$.

\item[(b)] $\mathbf{A}_{n}\rightarrow _{p}\mathbf{A}$ as $n\rightarrow
\infty $, where $\mathbf{A}$ is some positive definite matrix.

\item[(c)] $\,$ 
\begin{equation*}
\frac{1}{n}\sum_{i=1}^{n}\left[ \hat{\tilde{y}}_{i}^{r}x_{i}^{s_{r}}-\mathrm{%
E}\left( \tilde{y}_{i}^{r}x_{i}^{s_{r}}\right) \right] =O_{p}\left(
n^{-1/2}\right) ,
\end{equation*}%
for $r=0,1,2,\cdots ,2K-1$, $s_{r}=0,1,\cdots ,S-r,$ and $S>2K-1.$
\end{enumerate}
\end{assumption}

\begin{remark}
Parts (a) and (b) of Assumption \ref{assu:Consistency Assumption} are
standard regularity conditions in the GMM literature. Part (c) together with
Assumption \ref{assu:conv_moment} are high-level regularity conditions which
allow us to generalize the usual IID assumption and nest the IID data
generation process as a special case. The sample analogue terms in (c)
include $\hat{\tilde{y}}_{i}=y_{i}-\mathbf{z}_{i}^{\prime }\hat{\mathbf{%
\gamma }}$, instead of the infeasible $\tilde{y}_{i}=y_{i}-\mathbf{z}%
_{i}^{\prime }\mathbf{\gamma }$. The $\sqrt{n}$-consistency of $\hat{\mathbf{%
\gamma }}$ shown in Theorem \ref{lem:gamma_est_consistency} ensures that
replacing $\tilde{y}_{i}$ by $\hat{\tilde{y}}_{i}$ does not alter the
convergence rate.
\end{remark}

\begin{theorem}
\label{thm:consisteny} Let $\mathbf{\eta }=\left( \mathbf{\theta }^{\prime },%
\mathbf{\sigma }^{\prime }\right) ^{\prime }$ and $\mathbf{\eta }_{0}=\left( 
\mathbf{\theta }_{0}^{\prime },\mathbf{\sigma }_{0}^{\prime }\right)
^{\prime }$. Under Assumptions \ref{assu:identification_regularity_condition}%
, \ref{assu:conv_moment}, and \ref{assu:Consistency Assumption}, $\hat{%
\mathbf{\eta }}\rightarrow _{p}\mathbf{\eta }_{0}$ as $n\rightarrow \infty $.
\end{theorem}

The proof of Theorem \ref{thm:consisteny} is provided in Appendix \ref%
{sec:Proofs}.

\begin{assumption}
\label{assu:normality} Follow the notations as in Assumption \ref%
{assu:Consistency Assumption} and in addition denote $\mathbf{G}\left( 
\mathbf{\theta },\mathbf{\sigma },\mathbf{\gamma }\right) =\nabla _{\left(%
\mathbf{\theta }^{\prime }, \mathbf{\sigma }^{\prime }\right) ^{\prime }}%
\mathbf{g}_{0}\left( \mathbf{\theta },\mathbf{\sigma },\mathbf{\gamma }%
\right) $, $\mathbf{G}_{0}=\mathbf{G}\left( \mathbf{\theta }_{0},\mathbf{%
\sigma }_{0},\mathbf{\gamma }_{0}\right) $, $\mathbf{G}_{\gamma }\left( 
\mathbf{\theta },\mathbf{\sigma },\mathbf{\gamma }\right) =\nabla _{\mathbf{%
\gamma }}\mathbf{g}_{0}\left( \mathbf{\theta },\mathbf{\sigma },\mathbf{%
\gamma }\right) $, $\mathbf{G}_{0, \gamma}=\mathbf{G}_{\gamma} \left( 
\mathbf{\theta }_{0},\mathbf{\sigma }_{0},\mathbf{\gamma }_{0}\right) $.

\begin{enumerate}
\item[(a)] $\sqrt{n}\mathbf{\hat{g}}_{n}\left( \mathbf{\theta }_{0},\mathbf{%
\sigma }_{0},\mathbf{\gamma}_0\right) \rightarrow _{d}\mathbf{\zeta }\sim
N\left( 0,\mathbf{V}\right) $ as $n\to\infty$.

\item[(b)] $\mathbf{G}_{0}^{\prime }\mathbf{AG}_{0}\succ 0$.
\end{enumerate}
\end{assumption}

\begin{remark}
In Assumption \ref{assu:normality}, parts (a) is the high level condition
required to ensure the asymptotic normality of $\mathbf{\hat{g}}_{n}\left( 
\mathbf{\theta }_{0},\mathbf{\sigma }_{0},\mathbf{\gamma}_0\right) $, which
can be verified by Lindeberg central limit theorem under low-level
regularity conditions. Part (c) of Assumption \ref{assu:normality}
represents the full-rank condition on $\mathbf{G}_{0}$, required for
identification of $\mathbf{\theta }_{0}$ and $\mathbf{\sigma }_{0}$.
\end{remark}

By Theorem \ref{lem:gamma_est_consistency}, we have $\sqrt{n}\left( \hat{%
\mathbf{\gamma}} - \mathbf{\gamma} \right) \rightarrow_d \zeta_\gamma \sim
N(0, V_\gamma)$. The following theorem shows the asymptotic normality of the
GMM estimator $\hat{\mathbf{\eta}}$.

\begin{theorem}
\label{thm:normality} Under Assumptions \ref%
{assu:identification_regularity_condition}, \ref{assu:gamma_est_normality}, %
\ref{assu:Consistency Assumption} and \ref{assu:normality}, 
\begin{equation*}
\sqrt{n}\left( \hat{\mathbf{\eta }}-\mathbf{\eta }_{0}\right) \rightarrow
_{d}\left( \mathbf{G}_{0}^{\prime }\mathbf{A}\mathbf{G}_{0}\right) ^{-1}%
\mathbf{G}_{0}^{\prime }\mathbf{A}\left( \mathbf{\zeta }+\mathbf{G}_{0,
\gamma}\mathbf{\zeta }_{\gamma }\right),
\end{equation*}
as $n\rightarrow \infty $.
\end{theorem}

The proof of Theorem \ref{thm:normality} is provided in Appendix \ref%
{sec:Proofs}.

\begin{remark}
In practice, we estimate the variance of the asymptotic distribution of $%
\hat{\eta}$ by 
\begin{equation}  \label{eq:est_variance_matrix_gmm}
\hat{\mathbf{V}}_{\eta} = \left( \hat{\mathbf{G}}^\prime \hat{\mathbf{A}}_n 
\hat{\mathbf{G}} \right)^{-1} \hat{\mathbf{G}}^\prime \hat{\mathbf{A}}_n 
\hat{\mathbf{V}}_{\zeta} \hat{\mathbf{A}}_n^\prime \hat{\mathbf{G}} \left( 
\hat{\mathbf{G}}^\prime \hat{\mathbf{A}}_n \hat{\mathbf{G}} \right)^{-1},
\end{equation}
where $\hat{\mathbf{G}} = \nabla_{\left( \mathbf{\sigma }^{\prime },\mathbf{%
\theta }^{\prime }\right) ^{\prime }} \hat{\mathbf{g}}_{n}\left( \hat{%
\mathbf{\theta}}, \hat{\mathbf{\sigma}}, \hat{\gamma} \right)$, $\hat{%
\mathbf{A}}_n$ is given by \eqref{eq:efficient_weight_mat}, and 
\begin{equation*}
\hat{\mathbf{V}}_\zeta = \frac{1}{n} \sum_{i=1}^{n} \mathbf{\psi}_{n,i}%
\mathbf{\psi}_{n,i}^\prime,
\end{equation*}
where 
\begin{equation*}
\mathbf{\psi}_{n,i} = \hat{\mathbf{g}}_i\left( \hat{\mathbf{\theta}}, \hat{%
\mathbf{\sigma}}, \hat{\gamma} \right) + \nabla _{\mathbf{\gamma }}\hat{%
\mathbf{g}}_{n}\left( \hat{\mathbf{\theta}}, \hat{\mathbf{\sigma}}, \hat{%
\gamma}\right) \mathbf{L} \mathbf{Q}_{n,ww}^{-1}\left( \mathbf{w}_i\hat{\xi}_i
\right),
\end{equation*}
and $\mathbf{L} = 
\begin{pmatrix}
\mathbf{0}_{p_z\times 1} & \mathbf{I}_{p_z}%
\end{pmatrix}%
$ is the loading matrix that selects $\mathbf{\gamma}$ out of $\mathbf{\phi}$%
.
\end{remark}


\section{Multiple regressors with random coefficients\label{sec:Extensions}}

One important extension of the regression model \eqref{eq:basic_model} is to
allow for multiple regressors with random coefficients having categorical
distribution. With this in mind consider 
\begin{equation}
y_{i}=\mathbf{x}_{i}^{\prime }\mathbf{\beta }_{i}+\mathbf{z}_{i}^{\prime }%
\mathbf{\gamma }+u_{i},  \label{eq:multiple_x_model}
\end{equation}%
where the $p\times 1$ vector of random coefficients, $\mathbf{\beta }_{i}\in 
\mathbb{R}^{p}$ follows the multivariate distribution\footnote{%
We assume the number of categories $K$ is homogeneous across $j=1,2,\cdots
,p $. This is for notational simplicity, and can be readily generalized to
allow for $K_{j}\neq K_{j^{\prime }}$ without affecting the main results.} 
\begin{equation}
\mathrm{Pr}\left( \beta _{i1}=b_{1k_{1}},\beta _{i2}=b_{2k_{2}},\cdots
,\beta _{ip}=b_{pk_{p}}\right) =\pi _{k_{1},k_{2},\cdots ,k_{p}},
\label{eq:prob_dist_ext}
\end{equation}%
with $k_{j}\in \left\{ 1,2,\cdots ,K\right\} $, $b_{j1}<b_{j2}<\cdots
<b_{jK} $, and 
\begin{equation*}
\sum_{k_{1},k_{2},\cdots ,k_{p}\in \left\{ 1,2,\cdots ,K\right\} }\pi
_{k_{1},k_{2},\cdots ,k_{p}}=1.
\end{equation*}%
As in Section \ref{sec:Identification}, $\mathbf{\gamma }\in \mathbb{R}%
^{p_{z}}$, $\mathbf{w}_{i}=\left( \mathbf{x}_{i}^{\prime },\mathbf{z}%
_{i}^{\prime }\right) ^{\prime }$, $\mathbf{\beta }_{i}\perp \mathbf{w}_{i}$%
, $u_{i}\perp \mathbf{w}_{i}$, and $u_{i}$ are independently distributed
over $i$ with mean $0.$

\medskip 
\noindent \textbf{Example 1} \textit{Consider the simple case with $p = 2$
and $K = 2$. For $j = 1, 2$, denote two categories as $\left\{ L, H \right\}$%
. The probabilities of four possible combinations of realized $\mathbf{\beta}%
_i$ is summarized in Table \ref{tab:dist_beta_p2_k2}, where $\pi_{LL} +
\pi_{LH} + \pi_{HL} + \pi_{HH} = 1$. } 
\begin{table}[htbp]
\caption{Distribution of $\mathbf{\protect\beta}_i$ with $p = 2$ and $K = 2$}
\label{tab:dist_beta_p2_k2}
\begin{center}
\begin{tabular}{c|c|c}
\hline
& $k_2 = L$ & $k_2 = H$ \\ \hline
$k_1 = L$ & $\pi_{LL} = \Pr \left( \beta_{i1} = b_{1L}, \beta_{i2} = b_{2L}
\right)$ & $\pi_{LH} = \Pr \left( \beta_{i1} = b_{1L}, \beta_{i2} = b_{2H}
\right)$ \\ \hline
$k_1 = H$ & $\pi_{HL} = \Pr \left( \beta_{i1} = b_{1H}, \beta_{i2} = b_{2L}
\right)$ & $\pi_{HH} = \Pr \left( \beta_{i1} = b_{1H}, \beta_{i2} = b_{2H}
\right)$ \\ \hline
\end{tabular}%
\end{center}
\end{table}
\medskip 

We first identify the moments of $\mathbf{\beta }_{i}$. As in Section \ref%
{sec:Identification}, $\mathbf{\phi }=\left( \mathrm{E}\left( \mathbf{\beta }%
_{i}\right) ^{\prime },\mathbf{\gamma }^{\prime }\right) ^{\prime }$ is
identified by 
\begin{equation}
\mathbf{\phi } = \mathbf{Q}_{ww} ^{-1} \mathbf{q}_{wy} ,
\label{eq:identify_phi_ext}
\end{equation}%
under Assumption \ref{assu:identification_regularity_condition}. We now
consider the identification of the higher order moments of $\mathbf{\beta }%
_{i}$ up to the finite order $2K-1$.

Since $\mathbf{\gamma }$ is identified as in \eqref{eq:identify_phi_ext}, we
treat it as known and let $\tilde{y}_{i}^{r}=y_{i}-\mathbf{z}_{i}^{\prime }%
\mathbf{\gamma }$. For $r=2,3,\cdots ,2K-1$, consider the moment conditions 
\begin{align}
\mathrm{E}\left( \tilde{y}_{i}^{r}\right) & =\mathrm{E}\left[ \left( \mathbf{%
x}_{i}^{\prime }\mathbf{\beta }_{i}+u_{i}\right) ^{r}\right]  \notag \\
& =\mathrm{E}\left[ \left( \mathbf{x}_{i}^{\prime }\mathbf{\beta }%
_{i}\right) ^{r}\right] +\mathrm{E}\left( u_{i}^{r}\right) +\sum_{s=2}^{r-1}%
\binom{r}{s}\mathrm{E}\left[ \left( \mathbf{x}_{i}^{\prime }\mathbf{\beta }%
_{i}\right) ^{r-s}\right] \mathrm{E}\left( u_{i}^{s}\right) .
\label{eq:moment_y_r_ext}
\end{align}%
Note that $\mathbf{x}_{i}^{\prime }\mathbf{\beta }_{i}=\sum_{j=1}^{p}\beta
_{ij}x_{ij}$, and 
\begin{equation*}
\mathrm{E}\left[ \left( \sum_{j=1}^{p}\beta _{ij}x_{ij}\right) ^{r}\right]
=\sum_{\sum_{j=1}^{p}q_{j}=r}\binom{r}{\mathbf{q}}\mathrm{E}\left(
\prod_{j=1}^{p}x_{ij}^{q_{j}}\right) \mathrm{E}\left( \prod_{j=1}^{p}\beta
_{ij}^{q_{j}}\right) ,
\end{equation*}%
where $\binom{r}{\mathbf{q}}=\frac{r!}{q_{1}!q_{2}!\cdots q_{p}!}$, for
non-negative integers $r$, $q_{1}$, $\cdots $, $q_{p}$ with $%
r=\sum_{j=1}^{p}q_{j}$, denotes the multinomial coefficients. We stack $%
\prod_{j=1}^{p}x_{ij}^{q_{j}}$ with $\mathbf{q}\in \left\{ \mathbf{q}\in
\left\{ 0,1,\cdots r\right\} ^{p}:\sum_{j=1}^{p}q_{j}=r\right\} $ in a
vector form by denoting \footnote{%
For $\mathbf{x}\in \mathbb{R}^{p}$, note that $\mathbf{\tau }_{0}\left( 
\mathbf{x}\right) =1$, $\mathbf{\tau }_{1}\left( \mathbf{x}\right) =\mathbf{x%
}$ and $\mathbf{\tau }_{2}\left( \mathbf{x}\right) =\mathrm{vech}\left( 
\mathbf{x}\mathbf{x}^{\prime }\right) $.} 
\begin{equation*}
\mathbf{\tau }_{r}\left( \mathbf{x}_{i}\right) =\left[ \varphi \left( 
\mathbf{x}_{i},\mathbf{q}_{1}\right) ,\varphi \left( \mathbf{x}_{i},\mathbf{q%
}_{2}\right) ,\cdots ,\varphi \left( \mathbf{x}_{i},\mathbf{q}_{\nu
_{r}}\right) \right] ^{\prime },
\end{equation*}%
where $\varphi \left( \mathbf{x}_{i},\mathbf{q}\right)
=\prod_{j=1}^{p}x_{ij}^{q_{j}}$ and $\nu _{r}=\binom{r+p-1}{p-1}$ is the
number of distinct monomials of degree $r$ on the variables $%
x_{i1},x_{i2},\cdots ,x_{ip}$. Similarly, 
\begin{equation*}
\mathbf{\tau }_{r}\left( \mathbf{\beta }_{i}\right) =\left[ \varphi \left( 
\mathbf{\beta }_{i},\mathbf{q}_{1}\right) ,\varphi \left( \mathbf{\beta }%
_{i},\mathbf{q}_{2}\right) ,\cdots ,\varphi \left( \mathbf{\beta }_{i},%
\mathbf{q}_{\nu _{r}}\right) \right] ^{\prime },
\end{equation*}%
where $\varphi \left( \mathbf{\beta }_{i},\mathbf{q}\right)
=\prod_{j=1}^{p}\beta _{ij}^{q_{j}}$.

\medskip 
\noindent \textbf{Example 2} \textit{Consider $p = 2$ and $r = 2$, we have 
\begin{align*}
\mathbf{\tau}_2\left( \mathbf{x}_i \right) & = \left( x_{i1}^2,
x_{i1}x_{i2}, x_{i2}^2 \right)^\prime, \\
\mathbf{\tau}_2\left( \mathbf{\beta}_i \right) & = \left( \beta_{i1}^2,
\beta_{i1}\beta_{i2}, \beta_{i2}^2 \right)^\prime,
\end{align*}
and 
\begin{align*}
\mathrm{E}\left[ \left( x_{i1}\beta_{i1} + x_{i2}\beta_{i2} \right)^2 \right]
& = \mathrm{E}\left(x_{i1}^2\right)\mathrm{E}\left(\beta_{i1}^2\right) + 2 
\mathrm{E}\left(x_{i1}x_{i2}\right)\mathrm{E}\left(\beta_{i1}\beta_{i2}%
\right) + \mathrm{E}\left(x_{i2}^2 \right)\mathrm{E}\left(\beta_{i2}^2\right)
\\
& = \left[ \mathrm{E}\left(x_{i1}^2\right), \mathrm{E}\left(x_{i1}x_{i2}%
\right), \mathrm{E}\left(x_{i2}^2 \right) \right] \mathrm{diag}\left[ \left(
1, 2, 1 \right)^\prime \right] \left[ \mathrm{E}\left(\beta_{i1}^2\right), 
\mathrm{E}\left(\beta_{i1}\beta_{i2}\right), \mathrm{E}\left(\beta_{i2}^2
\right) \right]^\prime \\
& = \mathrm{E}\left[\mathbf{\tau}_2\left( \mathbf{x}_i \right)\right]^\prime 
\mathbf{\Lambda}_2 \mathrm{E}\left[\mathbf{\tau}_2\left( \mathbf{\beta}_i
\right)\right] ,
\end{align*}
where $\mathbf{\Lambda}_2 = \mathrm{diag}\left[ \left( 1, 2, 1
\right)^\prime \right]$.} 
\medskip

Then the moment condition \eqref{eq:moment_y_r_ext} can be written as 
\begin{align}
\mathrm{E}\left( \tilde{y}_{i}^{r}\right) & =\mathrm{E}\left[ \mathbf{\tau }%
_{r}\left( \mathbf{x}_{i}\right) \right] ^{\prime }\mathbf{\Lambda }_{r}%
\mathrm{E}\left[ \mathbf{\tau }_{r}\left( \mathbf{\beta }_{i}\right) \right]
+\mathrm{E}\left( u_{i}^{r}\right)  \notag \\
& \quad \quad \quad +\sum_{s=2}^{r-1}\binom{r}{s}\mathrm{E}\left[ \mathbf{%
\tau }_{r-s}\left( \mathbf{x}_{i}\right) \right] ^{\prime }\mathbf{\Lambda }%
_{r-s}\mathrm{E}\left[ \mathbf{\tau }_{r-s}\left( \mathbf{\beta }_{i}\right) %
\right] \mathrm{E}\left( u_{i}^{s}\right) ,  \label{eq:moment_y_r_ext_vec}
\end{align}%
where $\mathbf{\Lambda }_{r}=\mathrm{diag}\left[ \left[ \binom{r}{\mathbf{q}}%
\right] _{\sum_{j=1}^{p}q_{j}=r}\right] $ is the $\nu _{r}\times \nu _{r}$
diagonal matrix of multinomial coefficients. We further consider the moment
conditions 
\begin{align}
\mathrm{E}\left( \tilde{y}_{i}^{r}\mathbf{\tau }_{r}\left( \mathbf{x}%
_{i}\right) \right) & =\mathrm{E}\left[ \mathbf{\tau }_{r}\left( \mathbf{x}%
_{i}\right) \mathbf{\tau }_{r}\left( \mathbf{x}_{i}\right) ^{\prime }\right] 
\mathbf{\Lambda }_{r}\mathrm{E}\left[ \mathbf{\tau }_{r}\left( \mathbf{\beta 
}_{i}\right) \right] +\mathrm{E}\left[ \mathbf{\tau }_{r}\left( \mathbf{x}%
_{i}\right) \right] \mathrm{E}\left( u_{i}^{r}\right)  \notag \\
& \quad \quad \quad +\sum_{s=2}^{r-1}\binom{r}{s}\mathrm{E}\left[ \mathbf{%
\tau }_{r}\left( \mathbf{x}_{i}\right) \mathbf{\tau }_{r-s}\left( \mathbf{x}%
_{i}\right) ^{\prime }\right] \mathbf{\Lambda }_{r-s}\mathrm{E}\left[ 
\mathbf{\tau }_{r-s}\left( \mathbf{\beta }_{i}\right) \right] \mathrm{E}%
\left( u_{i}^{s}\right) ,  \label{eq:moment_y_r_x_r_ext_vec}
\end{align}%
$r=2,3,\cdots ,2K-1$. \eqref{eq:moment_y_r_ext_vec} and %
\eqref{eq:moment_y_r_x_r_ext_vec} reduce to \eqref{eq:mc_n_r} and %
\eqref{eq:mc_n_2r} when $p=1$.

\begin{assumption}
\label{assu:idenfication_regularity_condition_ext} $\,$

\begin{enumerate}
\item[(a)] $\left\Vert n^{-1}\sum_{i=1}^{n}\mathrm{E}\left( \tilde{y}_{i}^{r}%
\mathbf{\tau }_{s}\left( \mathbf{x}_{i}\right) \right) -\mathbf{\rho }%
_{r,s}\right\Vert =O\left( n^{-1/2}\right) ,$ and $\left\Vert \mathbf{\rho }%
_{r,s}\right\Vert <\infty $, $r,s=0,1,\cdots ,2K-1$.

\item[(b)] $\left\Vert n^{-1}\sum_{i=1}^{n}\mathrm{E}\left[ \mathbf{\tau }%
_{r}\left( \mathbf{x}_{i}\right) \mathbf{\tau }_{s}\left( \mathbf{x}%
_{i}\right) ^{\prime }\right] -\mathbf{\Xi }_{r,s}\right\Vert =O\left(
n^{-1/2}\right) ,$ and $\left\Vert \mathbf{\Xi }_{r,s}\right\Vert <\infty $, 
$r,s=0,1,\cdots ,2K-1$.

\item[(c)] $\left\vert n^{-1}\sum_{i=1}^{n}\mathrm{E}\left( u_{i}^{r}\right)
-\sigma _{r}\right\vert =O\left( n^{-1/2}\right) ,$ and $\left\vert \sigma
_{r}\right\vert <\infty $ for $r=2,3,\cdots ,2K-1$.

\item[(d)] $\left\Vert n^{-1}\sum_{i=1}^{n} \left[ \mathrm{var} \left( 
\mathbf{\tau}_r \left(\mathbf{x_i} \right) \right) - \left( \mathbf{\Xi}%
_{r,r} - \mathbf{\rho}_{0, r}\mathbf{\rho}_{0, r}^\prime \right) \right]
\right\Vert = O(n^{-1/2})$, where $\mathbf{\Xi}_{r,r} - \mathbf{\rho}_{0, r}%
\mathbf{\rho}_{0, r}^\prime \succ 0$ for $r=2,3\cdots, 2K-1$.
\end{enumerate}
\end{assumption}

\begin{theorem}
\label{thm: idenfitication_moment_beta_ext} For any $\mathbf{q} \in \left\{ 
\mathbf{q}\in \left\{ 0, 1, \cdots r \right\} ^{p}: \sum_{j=1}^p q_j =
r\right\}$ and $r = 2, 3,\cdots, 2K-1$, $\mathrm{E} \left(\prod_{j=1}^{p}%
\beta_{ij}^{q_{j}}\right)$ and $\sigma_r$ are identified under Assumptions %
\ref{assu:identification_regularity_condition} and \ref%
{assu:idenfication_regularity_condition_ext}.
\end{theorem}

\begin{proof}
For $r=2,3,\cdots ,2K-1$, sum \eqref{eq:moment_y_r_ext_vec} and %
\eqref{eq:moment_y_r_x_r_ext_vec} over $i,$ go through the same steps as in
the proof of Theorem \ref{lem: identification_moments}, then by Assumptions %
\ref{assu:idenfication_regularity_condition_ext}(a) to (c), we have (for $%
n\rightarrow \infty $) 
\begin{align}
\mathbf{\rho }_{r,0}^{\prime }\mathbf{\Lambda }_{r}\mathrm{E}\left[ \mathbf{%
\tau }_{r}\left( \mathbf{\beta }_{i}\right) \right] +\sigma _{r}& =\mathbf{%
\rho }_{r,0}-\sum_{s=2}^{r-1}\binom{r}{s}\mathbf{\rho }_{0,r-s}\mathbf{%
\Lambda }_{r-s}\mathrm{E}\left[ \mathbf{\tau }_{r-s}\left( \mathbf{\beta }%
_{i}\right) \right] \sigma _{s},  \label{eq:mc_limit_r_ext} \\
\mathbf{\Xi }_{r,r}\mathbf{\Lambda }_{r}\mathrm{E}\left[ \mathbf{\tau }%
_{r}\left( \mathbf{\beta }_{i}\right) \right] +\mathbf{\rho }_{0,r}\sigma
_{r}& =\mathbf{\rho }_{r,r}-\sum_{s=2}^{r-1}\binom{r}{s}\mathbf{\Xi }_{r,r-s}%
\mathbf{\Lambda }_{r-s}\mathrm{E}\left[ \mathbf{\tau }_{r-s}\left( \mathbf{%
\beta }_{i}\right) \right] \sigma _{s}.  \label{eq:mc_limit_2r_ext}
\end{align}%
Note that 
\begin{equation*}
\mathbf{M}_{r}=%
\begin{pmatrix}
\mathbf{\Xi }_{r,r} & \mathbf{\rho }_{0,r} \\ 
\mathbf{\rho }_{0,r}^{\prime } & 1%
\end{pmatrix}%
\begin{pmatrix}
\mathbf{\Lambda }_{r} & \mathbf{0} \\ 
\mathbf{0} & 1%
\end{pmatrix}%
,
\end{equation*}%
is invertible since $\det \left( \mathbf{M}_{r}\right) =\det \left( \mathbf{%
\Xi }_{r,r}-\mathbf{\rho }_{0,r}\mathbf{\rho }_{0,r}^{\prime }\right) \det
\left( \mathbf{\Lambda }_{r}\right) >0,$ for $r=2,3,\cdots ,R$, by
Assumption \ref{assu:idenfication_regularity_condition_ext}(d). As a result,
we can sequentially solve \eqref{eq:mc_limit_r_ext} and %
\eqref{eq:mc_limit_2r_ext} for $\mathrm{E}\left[ \mathbf{\tau }_{r}\left( 
\mathbf{\beta }_{i}\right) \right] $ and $\sigma _{r}$, for $r=2,3,\cdots
,2K-1$.
\end{proof}


We now move from the moments of $\mathbf{\beta }_{i}$ to the distribution of 
$\mathbf{\beta }_{i}$. We first focus on the identification of the marginal
probabilities obtained from (\ref{eq:prob_dist_ext}) by averaging out the
effects of the other coefficients except for $\beta _{ij}$, namely we
initially focus on identification of $\lambda _{jk}=\Pr \left( \beta
_{ij}=b_{jk}\right) $, for $k=1,2,\cdots ,K,$ and $j=1,2,\cdots ,p$.

\begin{remark}
Focusing on the marginal distribution of $\beta _{i}$ is similar to focusing
on estimation of partial derivatives in the context of non-parametric
estimation, where the curse of dimensionality applies. Consider the
estimation of regressing $y_{i}$ on $\mathbf{x}_{i}=\left(
x_{i1},x_{i2},\cdots ,x_{ip}\right) ^{\prime }$, 
\begin{equation*}
y_{i}=F\left( x_{i1},x_{i2},\cdots .x_{ip}\right) +u_{i}.
\end{equation*}%
Then if $F\left( x_{1},x_{i2},\cdots ,x_{ip}\right) $ is a homogeneous
function (of degree $1/\mu $), then 
\begin{equation*}
y_{i}=\sum_{j=1}^{p}\left( \mu \frac{\partial F\left( \cdot \right) }{%
\partial x_{ij}}\right) x_{ij}+u_{i},
\end{equation*}%
and under certain conditions we can treat $\mu \frac{\partial F\left( \cdot
\right) }{\partial x_{ij}}\equiv \beta _{ij}$.
\end{remark}

By Theorem \ref{thm: idenfitication_moment_beta_ext}, $\mathrm{E}\left(
\beta _{ij}^{r}\right) $ is identified for $r=1,2,\cdots ,2K-1$ under
Assumptions \ref{assu:identification_regularity_condition} and \ref%
{assu:idenfication_regularity_condition_ext}. By \eqref{eq:prob_dist_ext},
we have equations 
\begin{equation}
\mathrm{E}\left( \beta _{ij}^{r}\right) =\sum_{k=1}^{K}\lambda
_{jk}b_{jk}^{r},  \label{eq:mbeta_ext}
\end{equation}%
$r=0,1,\cdots ,2K-1$, which is of the same form as \eqref{mbeta} and %
\eqref{eq:mbeta_mat}. To identify $\mathbf{\lambda }_{j}=\left( \lambda
_{j1},\lambda _{j2},\cdots ,\lambda _{jK}\right) ^{\prime }$ and $\mathbf{b}%
_{j}=\left( b_{j1},b_{j2},\cdots ,b_{jK}\right) ^{\prime }$, we can verify
the system of $2K$ equations in \eqref{eq:mbeta_ext} has a unique solution
if $b_{j1}<b_{j2}<\cdots <b_{jK}$ and $\lambda _{jk}\in \left( 0,1\right) $.
The following corollary is a direct application of Theorem \ref%
{prop:identification_of_beta_L_H}. 
\begin{corollary}
    \label{core:marginal_dist}
    Consider the model \eqref{eq:multiple_x_model} and suppose that Assumptions \ref{assu:identification_regularity_condition} and \ref{assu:idenfication_regularity_condition_ext} hold. Then the parameters $\mathbf{\theta}_j = \left( \mathbf{\lambda}_{j}^\prime, \mathbf{b}_{j}^\prime \right)^\prime$ of the marginal distribution of $\beta_i$ with respect to $\beta_{ij}$ is identified subject to $b_{j1} < b_{j2} < \cdots < b_{jK}$ and $\lambda_{jk} \in \left( 0, 1 \right)$ for $j = 1, 2, \cdots, p$.
\end{corollary}

The problem of identification and estimation of the joint distribution of $%
\mathbf{\beta }_{i}$ is subject to the curse of dimensionality. We have $%
K^{p}-1$ probability weights, $\pi _{k_{1},k_{2},\cdots ,k_{p}}$, to be
identified in addition to the $pK$ categorical coefficients $b_{ij}$ that
are identified by Corollary \ref{core:marginal_dist}. The number of
parameters increases rapidly with $p$. Even in the simplest case with $K=2$,
the total number of unknown parameters is $2p+2^{p}-1$, which grows
exponentially.

Note that the marginal probabilities $\lambda _{jk}$ are related to the
joint distribution by 
\begin{equation}
\lambda _{jk}=\sum_{k_{1},\cdots ,k_{j-1},k_{j+1},\cdots ,k_{p}\in \left\{
1,2,\cdots ,K\right\} }\pi _{k_{1},k_{2},\cdots ,k_{j-1},k,k_{j+1},\cdots
,k_{p}},  \label{eq:marginal_dist_prob}
\end{equation}%
$k=1,2,\cdots ,K$ and $j=1,2,\cdots ,p$. The number of linearly independent
equations in \eqref{eq:marginal_dist_prob} is $pK-(p-1)$.

\medskip 
\noindent\textbf{Example 3} \textit{\ Consider the same setup as in Example
1 with $p = 2$ and $K = 2$. The marginal probabilities are obtained by 
\begin{align}
\lambda_{1L} = \Pr\left( \beta_{i1} = b_{1L} \right) = \pi_{LL} + \pi_{LH},
& \quad \lambda_{1H} = \Pr\left( \beta_{i1} = b_{1H} \right) = 1 -
\lambda_{1L} = \pi_{HL} + \pi_{HH},  \notag \\
\lambda_{2L} = \Pr\left( \beta_{i2} = b_{2L} \right) = \pi_{LL} + \pi_{HL},
& \quad \lambda_{2H} = \Pr\left( \beta_{i2} = b_{2H} \right) = 1 -
\lambda_{2L} = \pi_{LH} + \pi_{HH} .  \label{eq:marginal_prob_p2k2}
\end{align}
Note that any equation in \eqref{eq:marginal_prob_p2k2} can be expressed as
a linear combination of other three equations, for example $\lambda_{2H} =
\lambda_{1L} + \lambda_{1H} - \lambda_{2L}$. } 
\medskip

The equations corresponding to the cross-moments, $\mathrm{E}\left(
\prod_{j=1}^{p}\beta _{ij}^{q_{j}}\right) $, are 
\begin{equation}
\mathrm{E}\left( \prod_{j=1}^{p}\beta _{ij}^{q_{j}}\right)
=\sum_{k_{1},k_{2},\cdots ,k_{p}\in \left\{ 1,2,\cdots ,K\right\} }\left(
\prod_{j=1}^{p}b_{jk_{j}}^{q_{j}}\right) \pi _{k_{1},k_{2},\cdots ,k_{p}},
\label{eq:cross_moments}
\end{equation}%
for $\mathbf{q}\in \left\{ \mathbf{q}\in \left\{ 0,1,\cdots r-1\right\}
^{p}:\sum_{j=1}^{p}q_{j}=r\right\} $, $r=2,\cdots ,2K-1$. The linear system %
\eqref{eq:cross_moments} has 
\begin{equation*}
\sum_{r=1}^{2K-1}\binom{r+p-1}{p-1}-p(2K-1)
\end{equation*}%
equations. Then the total number of equations in %
\eqref{eq:marginal_dist_prob} and \eqref{eq:cross_moments} that can be
utilized to identify joint probabilities is $C_{r}=\sum_{r=1}^{2K-1}\binom{%
r+p-1}{p-1}-pK$, which is smaller than the number of joint probabilities $%
K^{p}-1$ for large $p$. When $K=2$, $C_{r}<K^{p}-1$ for $p\geq 7$.

Identification and estimation of the joint distribution of $\mathbf{\beta }%
_{i}$ in the general setting will not be pursued in this paper due to the
curse of dimensionality. Instead, we consider special cases, that are
empirically relevant, in which identification of the joint distribution of $%
\mathbf{\beta }_{i}$ can be readily established. We first consider small $p$
and $K$, in particular $p=2$ and $K=2$ as in Example 1.

\medskip \noindent \textbf{Example 4} 
Consider the same setup as in Example 1 with $p=2$ and $K=2$. In addition to %
\eqref{eq:marginal_prob_p2k2}, consider the cross-moment, 
\begin{equation}
\mathrm{E}\left( \beta _{i1}\beta _{i2}\right) =b_{1L}b_{2L}\pi
_{LL}+b_{1L}b_{2H}\pi _{LH}+b_{1H}b_{2L}\pi _{HL}+b_{1H}b_{2H}\pi _{HH}.
\label{eq:cross_moments_p2k2}
\end{equation}%
Writing \eqref{eq:marginal_prob_p2k2} and \eqref{eq:cross_moments_p2k2} in
matrix form, we have 
\begin{equation*}
\mathbf{B}\mathbf{\pi }=\mathbf{\lambda },
\end{equation*}%
where 
\begin{equation*}
\mathbf{B}=%
\begin{pmatrix}
1 & 1 & 0 & 0 \\ 
0 & 0 & 1 & 1 \\ 
1 & 0 & 1 & 0 \\ 
b_{1L}b_{2L} & b_{1L}b_{2H} & b_{1H}b_{2L} & b_{1H}b_{2H}%
\end{pmatrix}%
,\,\mathbf{\pi }=%
\begin{pmatrix}
\pi _{LL} \\ 
\pi _{LH} \\ 
\pi _{HL} \\ 
\pi _{HH}%
\end{pmatrix}%
,\,\mathbf{\lambda }=%
\begin{pmatrix}
\lambda _{1L} \\ 
\lambda _{1H} \\ 
\lambda _{2L} \\ 
\mathrm{E}\left( \beta _{i1}\beta _{i2}\right)%
\end{pmatrix}%
.
\end{equation*}%
Note that $\mathrm{E}\left( \beta _{i1}\beta _{i2}\right) $ is identified by
Theorem \ref{thm: idenfitication_moment_beta_ext}, and $b_{jk_{j}}$ and $%
\lambda _{jk_{j}}$ are identified by Corollary \ref{core:marginal_dist}, and
matrix $\mathbf{B}$ is invertible given that $b_{1L}<b_{1H}$ and $%
b_{2L}<b_{2H}$. (See Appendix \ref{sec:Proofs}). As a result, the joint
probabilities, $\mathbf{\pi },$ are identified.

\medskip

\begin{remark}
The argument in Example 4 is applicable for identification of the joint
distribution of $\left( \beta_{ij}, \beta_{i,j^\prime} \right)^\prime$ for $%
j\neq j^\prime$ when $p > 2$ and $K = 2$.
\end{remark}


\section{Finite sample properties using Monte Carlo experiments\label%
{sec:Monte-Carlo-Simulation}}

We examine the finite sample performance of the categorical coefficient
estimator proposed in Section \ref{sec:Estimation} by Monte Carlo
experiments.

\subsection{Data generating processes\label{subsec:dgp}}

We generate $y_{i}$ as 
\begin{equation}
y_{i}=\alpha +x_{i}\beta _{i}+z_{i1}\gamma _{1}+z_{i2}\gamma _{2}+u_{i},%
\text{ for }i=1,2,...,n,  \label{eq:mc_dgp}
\end{equation}%
with $\beta _{i}$ distributed as in (\ref{eq:category_dist}) with $K=2,$ and
the parameters $\pi ,\beta _{L}$ and $\beta _{H}$.\footnote{%
A Monte Carlo experiment with $K=3$ is relegated to Section \ref%
{subsec:mc_k3} in the online supplement.}

We draw $\beta _{i}$ for each individual $i$ independently by setting $\beta
_{i}=\beta _{L}$ with probability $\pi $ and $\beta _{i}=\beta _{H}$ with
probability $1-\pi $, through a sequence of independent Bernoulli draws. \
We consider two sets of parameters in all DGPs, denoted as \textit{high
variance} and \textit{low variance} parametrization, respectively, 
\begin{equation}
\left( \pi ,\beta _{L},\beta _{H},\mathrm{E}\left( \beta _{i}\right) ,%
\mathrm{var}\left( \beta _{i}\right) \right) =%
\begin{cases}
\left( 0.5,1,2,1.5,0.25\right) & \left( high\,variance\right) \\ 
\left( 0.3,0.5,1.345,1.0915,0.15\right) & \left( low\,variance\right)%
\end{cases}%
.  \label{eq:mc_dgp_para}
\end{equation}%
$\beta _{H}/\beta _{L}=2$ for the \textit{high variance} parametrization,
and $\beta _{H}/\beta _{L} = 2.69$, for the \textit{low variance}
parametrization, which is motivated by the estimates in our empirical
illustration in Section \ref{sec:Empirical-Application}.\footnote{%
The estimates for $\beta_H / \beta_L$ in our empirical analysis range from
1.50 to 2.79.} The values of E$(\beta _{i})$ and $\mathrm{var}\left( \beta
_{i}\right) $ are obtained noting that E$(\beta _{i})=\pi \beta _{L}+(1-\pi
)\beta _{H}$, and $\mathrm{var}\left( \beta _{i}\right) =\pi (1-\pi )(\beta
_{H}-\beta _{L})^{2}$. The remaining parameters are set as $\alpha =0.25$,
and $\mathbf{\gamma }=\left( 1,1\right) ^{\prime },$ across DGPs.

We generate the regressors and the error terms as follows.

\medskip \textbf{DGP 1 (Baseline)} We first generate $\tilde{x}_{i}\sim 
\text{IID}\chi ^{2}(2)$, and then set $x_{i}=(\tilde{x}_{i}-2)/2$ so that $%
x_{i}$ has $0$ mean and unit variance. The additional regressors, $z_{ij}$,
for $j=1,2$ with homogeneous slopes are generated as%
\begin{equation*}
z_{i1}=x_{i}+v_{i1}\text{ and }z_{i2}=z_{i1}+v_{i2},
\end{equation*}%
with $v_{ij}\sim \text{IID }N\left( 0,1\right) $, for $j=1,2$. This ensures
that the regressors are sufficiently correlated. The error term, $u_{i}$, is
generated as $u_{i}=\sigma _{i}\varepsilon _{i}$, where $\sigma _{i}^{2}$
are generated as $0.5(1+\text{IID}\chi ^{2}(1))$, and $\varepsilon _{i}\sim 
\text{IID}N(0,1)$. Note that $\varepsilon _{i}$ and $\sigma _{i}^{2}$ are
generated independently, and $E(u_{i}^{2})=1$.

\medskip \textbf{DGP 2 (Categorical $x$)} This setup deviates from the
baseline DGP, and allows the distribution of $x_{i}$ to differ across $i$.
Accordingly, we generate $x_{i}=\left( \tilde{x}_{1i}-2\right) /2$ where $%
\tilde{x}_{1i}\sim \text{IID}\chi ^{2}\left( 2\right) $ for $i=1,2,\cdots
,\lfloor n/2\rfloor $, and $x_{i}=\left( \tilde{x}_{2i}-2\right) /4$ where $%
\tilde{x}_{2i}\sim \text{IID}\chi ^{2}\left( 4\right) $, for $i=\lfloor
n/2\rfloor +1,\cdots ,n$. The additional regressors, $z_{ij}$, for $j=1,2$
with homogeneous slopes are generated as 
\begin{equation*}
z_{i1}=x_{i}+v_{i1}\text{ and }z_{i2}=z_{i1}+v_{i2},
\end{equation*}%
with $v_{ij}\sim \text{IID }N\left( 0,1\right) $, for $j=1,2$. The error
term $u_{i}$ is generated the same as in DGP 1.

\medskip \textbf{DGP 3 (Categorical $u$)} We generate $x_{i}$ and $\mathbf{z}%
_{i}$ the same as in DGP 1, but allow the error term $u_{i}$ to have a
heterogeneous distribution over $i$. For $i=1,2,\cdots ,\lfloor n/2\rfloor $%
, we set $u_{i}=\sigma _{i}\varepsilon _{i},$ where $\sigma _{i}^{2}\sim 
\text{IID}\chi ^{2}\left( 2\right) $ and $\varepsilon _{i}\sim \text{IID}%
N(0,1)$, and for $i=\lfloor n/2\rfloor +1,\cdots ,n$, we set $u_{i}=\left( 
\tilde{u}_{i}-2\right) /2$, where $\tilde{u}_{i}\sim \text{IID}\chi
^{2}\left( 2\right) $.

\medskip We investigate the finite sample performance of the estimator
proposed in Section \ref{sec:Estimation} across DGP 1 to 3 with \textit{low
variance} and \textit{high variance} scenarios.\footnote{%
We can consider a DGP with conditional heteroskedasticity, in which we
follow the baseline DGP and generate the error term as $u_{i}=x_{i}%
\varepsilon _{i}$, where $\varepsilon _{i}\sim N(0,1)$. The least square
estimator for $\mathbf{\phi }$ is valid in this setup in terms of estimation
and inference, whereas the GMM estimator for the distributional parameters $%
\mathbf{\theta }$ breaks down, which is to be expected since we can only
identify the first moment of $\beta _{i}$ under conditional
heteroskedasticity. The results are available on request.} Details of the
computational algorithm used to carry out the Monte Carlo experiments (and
the empirical results that follow) are given in Section \ref{sec:computation}
of the online supplement. An accompanying R package is available at %
\url{https://github.com/zhan-gao/ccrm}.

\subsection{Summary of the MC results}

\begin{table}[tbp]
\caption{Bias, RMSE and size of the least square estimator $\hat{\mathbf{%
\protect\phi}}$}
\label{tab:mc_bias_rmse_gamma}
\begin{center}
{\small 
\begin{tabular}{rrrrrrrrrrr}
\hline
\multicolumn{2}{r|}{DGP} & \multicolumn{3}{c|}{Baseline} & 
\multicolumn{3}{c|}{Categorical $x$} & \multicolumn{3}{c}{Categorical $u$}
\\ \hline
\multicolumn{2}{r|}{Sample size $n$} & \multicolumn{1}{c}{Bias} & 
\multicolumn{1}{c}{RMSE} & \multicolumn{1}{c|}{Size} & \multicolumn{1}{c}{
Bias} & \multicolumn{1}{c}{RMSE} & \multicolumn{1}{c|}{Size} & 
\multicolumn{1}{c}{Bias} & \multicolumn{1}{c}{RMSE} & \multicolumn{1}{c}{Size
} \\ \hline
\multicolumn{11}{c}{\textit{high variance}: $\mathrm{var}\left( \beta_i
\right) = 0.25$} \\ \hline
\multirow{7}{*}{\begin{turn}{90} $\mathrm{E}\left(\beta_i\right) = 1.5$
\end{turn}} & \multicolumn{1}{r|}{100} & -0.0024 & 0.2035 & 
\multicolumn{1}{r|}{0.0966} & -0.0037 & 0.2035 & \multicolumn{1}{r|}{0.0858}
& -0.0042 & 0.2268 & 0.0920 \\ 
& \multicolumn{1}{r|}{1,000} & -0.0017 & 0.0669 & \multicolumn{1}{r|}{0.0568}
& -0.0002 & 0.0657 & \multicolumn{1}{r|}{0.0540} & -0.0019 & 0.0738 & 0.0540
\\ 
& \multicolumn{1}{r|}{2,000} & -0.0008 & 0.0463 & \multicolumn{1}{r|}{0.0512}
& -0.0015 & 0.0475 & \multicolumn{1}{r|}{0.0534} & -0.0010 & 0.0523 & 0.0522
\\ 
& \multicolumn{1}{r|}{5,000} & -0.0004 & 0.0301 & \multicolumn{1}{r|}{0.0540}
& -0.0008 & 0.0300 & \multicolumn{1}{r|}{0.0546} & -0.0007 & 0.0335 & 0.0560
\\ 
& \multicolumn{1}{r|}{10,000} & 0.0002 & 0.0214 & \multicolumn{1}{r|}{0.0508}
& 0.0000 & 0.0212 & \multicolumn{1}{r|}{0.0510} & 0.0000 & 0.0229 & 0.0456
\\ 
& \multicolumn{1}{r|}{100,000} & -0.0001 & 0.0066 & \multicolumn{1}{r|}{
0.0472} & 0.0000 & 0.0066 & \multicolumn{1}{r|}{0.0460} & 0.0000 & 0.0075 & 
0.0506 \\ \hline
\multirow{7}{*}{\begin{turn}{90} $\gamma_1 = 1$ \end{turn}} & 
\multicolumn{1}{r|}{100} & -0.0022 & 0.1571 & \multicolumn{1}{r|}{0.0604} & 
-0.0006 & 0.1598 & \multicolumn{1}{r|}{0.0666} & 0.0018 & 0.1912 & 0.0656 \\ 
& \multicolumn{1}{r|}{1,000} & 0.0004 & 0.0501 & \multicolumn{1}{r|}{0.0496}
& -0.0005 & 0.0496 & \multicolumn{1}{r|}{0.0508} & 0.0000 & 0.0600 & 0.0530
\\ 
& \multicolumn{1}{r|}{2,000} & 0.0003 & 0.0352 & \multicolumn{1}{r|}{0.0530}
& -0.0004 & 0.0350 & \multicolumn{1}{r|}{0.0544} & 0.0002 & 0.0432 & 0.0602
\\ 
& \multicolumn{1}{r|}{5,000} & -0.0001 & 0.0222 & \multicolumn{1}{r|}{0.0470}
& 0.0005 & 0.0225 & \multicolumn{1}{r|}{0.0548} & 0.0007 & 0.0267 & 0.0522
\\ 
& \multicolumn{1}{r|}{10,000} & -0.0004 & 0.0157 & \multicolumn{1}{r|}{0.0470
} & 0.0002 & 0.0157 & \multicolumn{1}{r|}{0.0512} & 0.0000 & 0.0188 & 0.0504
\\ 
& \multicolumn{1}{r|}{100,000} & -0.0001 & 0.0049 & \multicolumn{1}{r|}{
0.0494} & 0.0000 & 0.0049 & \multicolumn{1}{r|}{0.0468} & 0.0000 & 0.0059 & 
0.0500 \\ \hline
\multirow{7}{*}{\begin{turn}{90} $\gamma_2 = 1$ \end{turn}} & 
\multicolumn{1}{r|}{100} & 0.0011 & 0.1115 & \multicolumn{1}{r|}{0.0616} & 
0.0016 & 0.1121 & \multicolumn{1}{r|}{0.0654} & -0.0002 & 0.1364 & 0.0700 \\ 
& \multicolumn{1}{r|}{1,000} & -0.0003 & 0.0358 & \multicolumn{1}{r|}{0.0558}
& 0.0001 & 0.0354 & \multicolumn{1}{r|}{0.0550} & 0.0006 & 0.0421 & 0.0508
\\ 
& \multicolumn{1}{r|}{2,000} & -0.0001 & 0.0253 & \multicolumn{1}{r|}{0.0522}
& 0.0006 & 0.0246 & \multicolumn{1}{r|}{0.0502} & -0.0003 & 0.0302 & 0.0560
\\ 
& \multicolumn{1}{r|}{5,000} & 0.0000 & 0.0158 & \multicolumn{1}{r|}{0.0480}
& 0.0000 & 0.0159 & \multicolumn{1}{r|}{0.0570} & -0.0003 & 0.0185 & 0.0470
\\ 
& \multicolumn{1}{r|}{10,000} & 0.0002 & 0.0111 & \multicolumn{1}{r|}{0.0494}
& -0.0002 & 0.0111 & \multicolumn{1}{r|}{0.0530} & -0.0001 & 0.0134 & 0.0522
\\ 
& \multicolumn{1}{r|}{100,000} & 0.0001 & 0.0035 & \multicolumn{1}{r|}{0.0488
} & 0.0000 & 0.0034 & \multicolumn{1}{r|}{0.0446} & 0.0000 & 0.0042 & 0.0496
\\ \hline
\multicolumn{11}{c}{\textit{low variance}: $\mathrm{var}\left( \beta_i
\right) = 0.15$} \\ \hline
\multirow{7}{*}{\begin{turn}{90} $\mathrm{E}\left(\beta_i\right) = 1.0915$
\end{turn}} & \multicolumn{1}{r|}{100} & -0.0006 & 0.1829 & 
\multicolumn{1}{r|}{0.0810} & -0.0023 & 0.1855 & \multicolumn{1}{r|}{0.0766}
& -0.0025 & 0.2094 & 0.0828 \\ 
& \multicolumn{1}{r|}{1,000} & -0.0005 & 0.0597 & \multicolumn{1}{r|}{0.0610}
& 0.0005 & 0.0590 & \multicolumn{1}{r|}{0.0478} & -0.0006 & 0.0670 & 0.0542
\\ 
& \multicolumn{1}{r|}{2,000} & -0.0002 & 0.0408 & \multicolumn{1}{r|}{0.0516}
& -0.0007 & 0.0427 & \multicolumn{1}{r|}{0.0606} & -0.0004 & 0.0475 & 0.0544
\\ 
& \multicolumn{1}{r|}{5,000} & -0.0002 & 0.0264 & \multicolumn{1}{r|}{0.0530}
& -0.0006 & 0.0266 & \multicolumn{1}{r|}{0.0480} & -0.0005 & 0.0302 & 0.0538
\\ 
& \multicolumn{1}{r|}{10,000} & 0.0000 & 0.0189 & \multicolumn{1}{r|}{0.0546}
& -0.0002 & 0.0188 & \multicolumn{1}{r|}{0.0486} & -0.0002 & 0.0208 & 0.0482
\\ 
& \multicolumn{1}{r|}{100,000} & -0.0001 & 0.0059 & \multicolumn{1}{r|}{
0.0474} & 0.0000 & 0.0059 & \multicolumn{1}{r|}{0.0494} & 0.0000 & 0.0068 & 
0.0508 \\ \hline
\multirow{7}{*}{\begin{turn}{90} $\gamma_1 = 1$ \end{turn}} & 
\multicolumn{1}{r|}{100} & -0.0027 & 0.1521 & \multicolumn{1}{r|}{0.0614} & 
-0.0001 & 0.1538 & \multicolumn{1}{r|}{0.0622} & 0.0014 & 0.1847 & 0.0624 \\ 
& \multicolumn{1}{r|}{1,000} & 0.0001 & 0.0480 & \multicolumn{1}{r|}{0.0520}
& -0.0007 & 0.0481 & \multicolumn{1}{r|}{0.0542} & -0.0003 & 0.0584 & 0.0570
\\ 
& \multicolumn{1}{r|}{2,000} & 0.0002 & 0.0338 & \multicolumn{1}{r|}{0.0514}
& -0.0006 & 0.0334 & \multicolumn{1}{r|}{0.0512} & 0.0001 & 0.0417 & 0.0572
\\ 
& \multicolumn{1}{r|}{5,000} & -0.0002 & 0.0213 & \multicolumn{1}{r|}{0.0474}
& 0.0003 & 0.0216 & \multicolumn{1}{r|}{0.0532} & 0.0007 & 0.0257 & 0.0498
\\ 
& \multicolumn{1}{r|}{10,000} & -0.0003 & 0.0150 & \multicolumn{1}{r|}{0.0466
} & 0.0002 & 0.0152 & \multicolumn{1}{r|}{0.0542} & 0.0001 & 0.0183 & 0.0518
\\ 
& \multicolumn{1}{r|}{100,000} & -0.0001 & 0.0047 & \multicolumn{1}{r|}{
0.0482} & 0.0000 & 0.0047 & \multicolumn{1}{r|}{0.0474} & 0.0000 & 0.0057 & 
0.0500 \\ \hline
\multirow{7}{*}{\begin{turn}{90} $\gamma_2 = 1$ \end{turn}} & 
\multicolumn{1}{r|}{100} & 0.0011 & 0.1081 & \multicolumn{1}{r|}{0.0592} & 
0.0013 & 0.1079 & \multicolumn{1}{r|}{0.0622} & -0.0002 & 0.1323 & 0.0674 \\ 
& \multicolumn{1}{r|}{1,000} & -0.0003 & 0.0345 & \multicolumn{1}{r|}{0.0594}
& 0.0003 & 0.0342 & \multicolumn{1}{r|}{0.0556} & 0.0006 & 0.0409 & 0.0500
\\ 
& \multicolumn{1}{r|}{2,000} & 0.0000 & 0.0243 & \multicolumn{1}{r|}{0.0534}
& 0.0006 & 0.0235 & \multicolumn{1}{r|}{0.0450} & -0.0001 & 0.0292 & 0.0576
\\ 
& \multicolumn{1}{r|}{5,000} & 0.0001 & 0.0152 & \multicolumn{1}{r|}{0.0490}
& 0.0001 & 0.0152 & \multicolumn{1}{r|}{0.0552} & -0.0002 & 0.0179 & 0.0470
\\ 
& \multicolumn{1}{r|}{10,000} & 0.0002 & 0.0106 & \multicolumn{1}{r|}{0.0454}
& -0.0002 & 0.0107 & \multicolumn{1}{r|}{0.0528} & -0.0002 & 0.0131 & 0.0526
\\ 
& \multicolumn{1}{r|}{100,000} & 0.0001 & 0.0033 & \multicolumn{1}{r|}{0.0442
} & 0.0000 & 0.0033 & \multicolumn{1}{r|}{0.0448} & 0.0000 & 0.0040 & 0.0486
\\ \hline
\end{tabular}
}
\end{center}
\par
{\footnotesize \textit{Notes:} The data generating process is %
\eqref{eq:mc_dgp}. \textit{high variance} and \textit{low variance}
parametrization are described in \eqref{eq:mc_dgp_para}. ``Baseline'',
``Categorical $x$'' and ``Categorical $u$'' refer to DGP 1 to 3 as in
Section \ref{subsec:dgp}. Generically, bias, RMSE and size are calculated by 
$R^{-1}\sum_{r=1}^R \left( \hat{ \theta}^{(r)} - \theta_0 \right)$, $\sqrt{%
R^{-1}\sum_{r= 1}^R \left( \hat{\theta}^{(r)} -\theta_0 \right)^2}$, and $%
R^{-1}\sum_{r=1}^R \mathbf{1}\left[ \left\vert \hat{\theta}^{(r)} - \theta_0
\right\vert / \hat{\sigma}_{\hat{\theta}}^{(r)} > \mathrm{cv}_{0.05} \right] 
$, respectively, for true parameter $\theta_0$, its estimate $\hat{\theta}%
^{(r)}$, the estimated standard error of $\hat{\theta}^{(r)}$, $\hat{\sigma}%
_{\hat{\theta}}^{(r)}$, and the critical value $\mathrm{cv}_{0.05} =
\Phi^{-1}\left( 0.975 \right)$ across $R = 5,000$ replications, where $%
\Phi\left( \cdot \right)$ is the cumulative distribution function of
standard normal distribution.}
\end{table}

\begin{figure}[tbp]
\caption{Empirical power functions for the least square estimator $\hat{%
\mathbf{\protect\phi}}$ with the \textit{high variance} parametrization ($%
\mathrm{var}\left( \protect\beta_i \right) = 0.25$)}
\label{fig:power_function_gamma_high}
\begin{center}
\begin{subfigure}[b]{\textwidth}
            \centering
            \caption{Baseline}
            \includegraphics[width=\textwidth, height=1.9in]{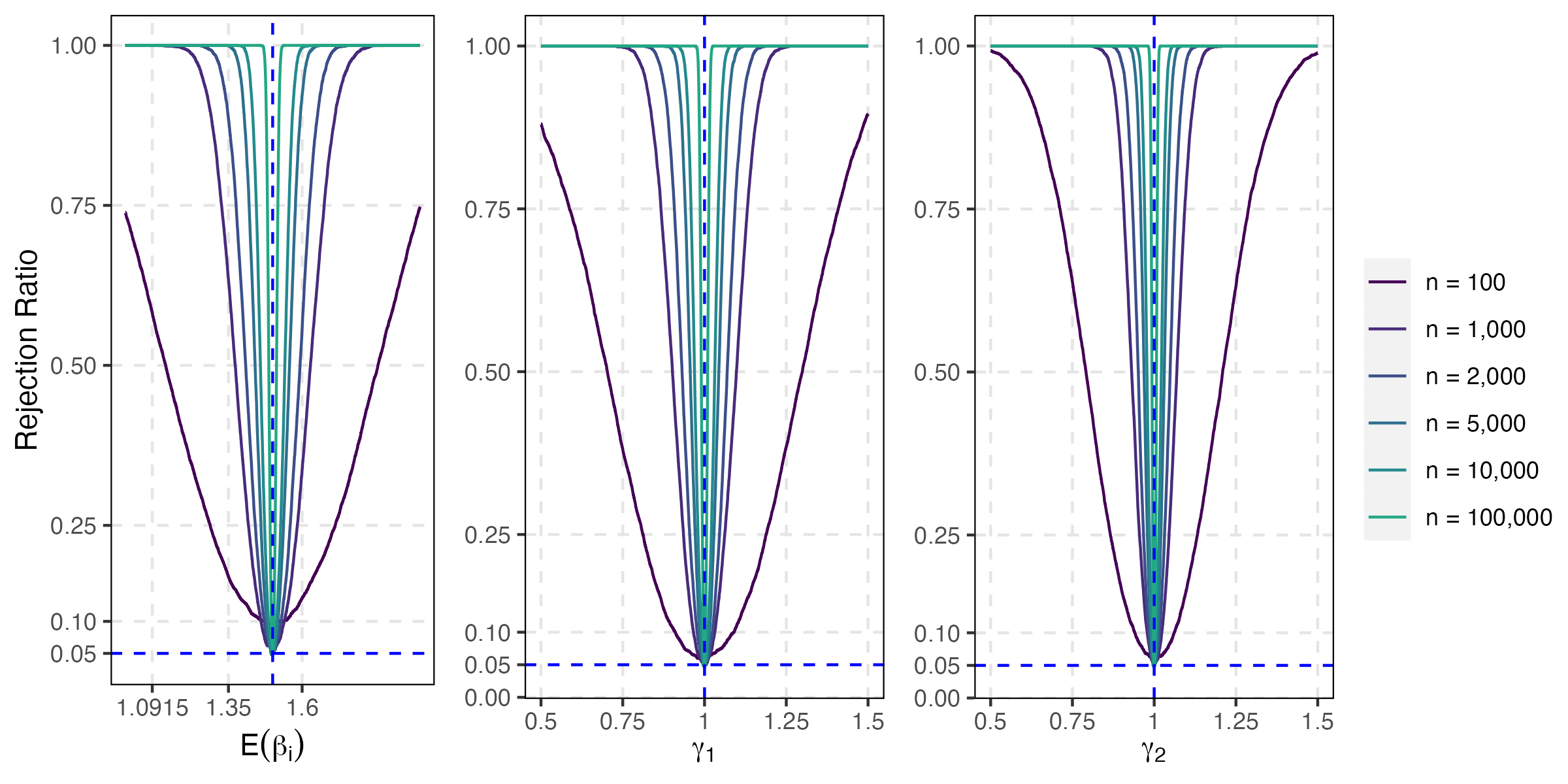}
            \label{fig:power_function_gamma_baseline_high}
        \end{subfigure}\\[0pt]
\begin{subfigure}[b]{\textwidth}
            \centering
            \caption{Categorical $x$}
            \includegraphics[width=\textwidth, height=1.9in]{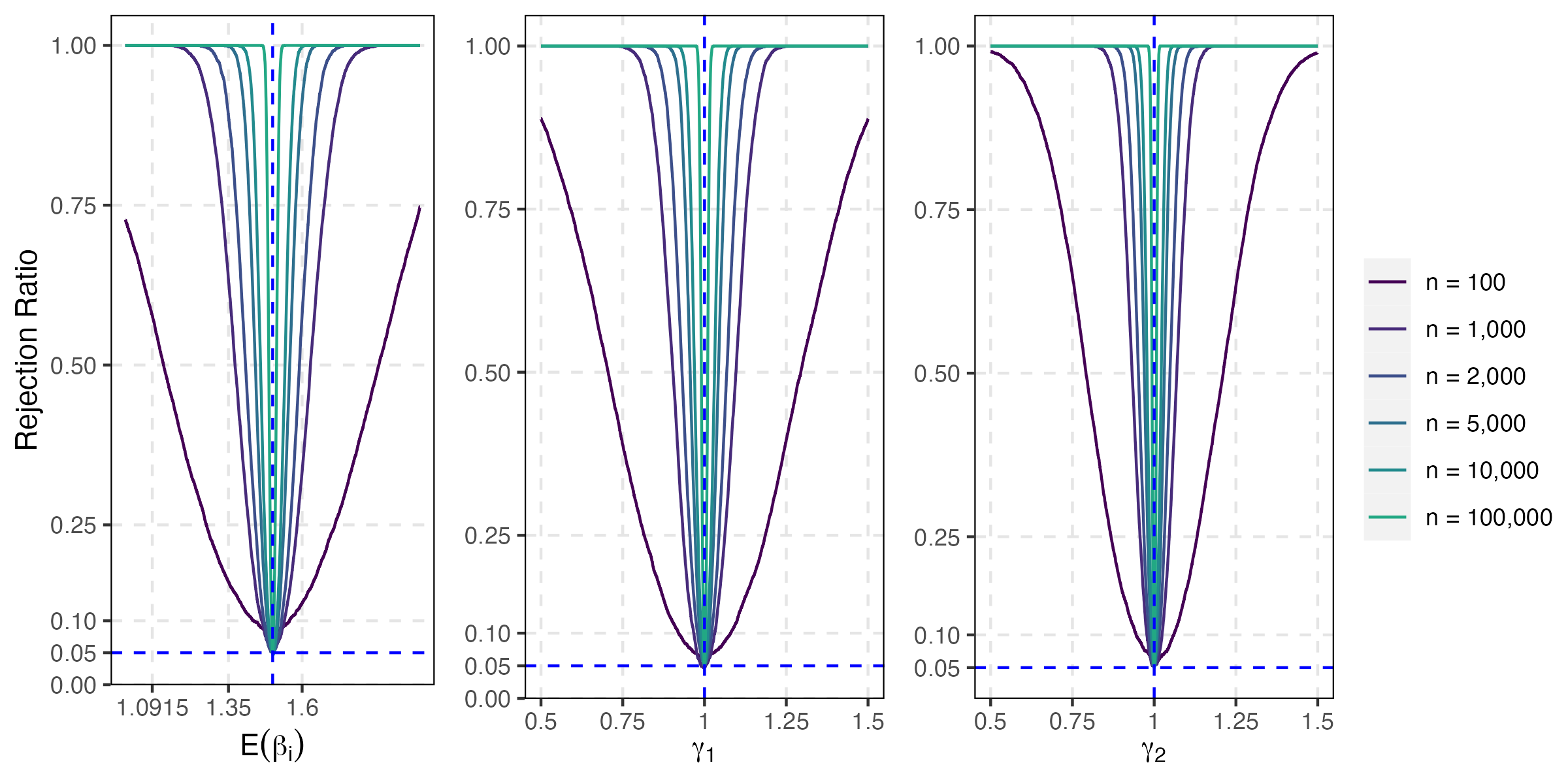}
            \label{fig:power_function_gamma_cat_x_high}
        \end{subfigure}\\[0pt]
\begin{subfigure}[b]{\textwidth}
            \centering
            \caption{Categorical $u$}
            \includegraphics[width=\textwidth, height=1.9in]{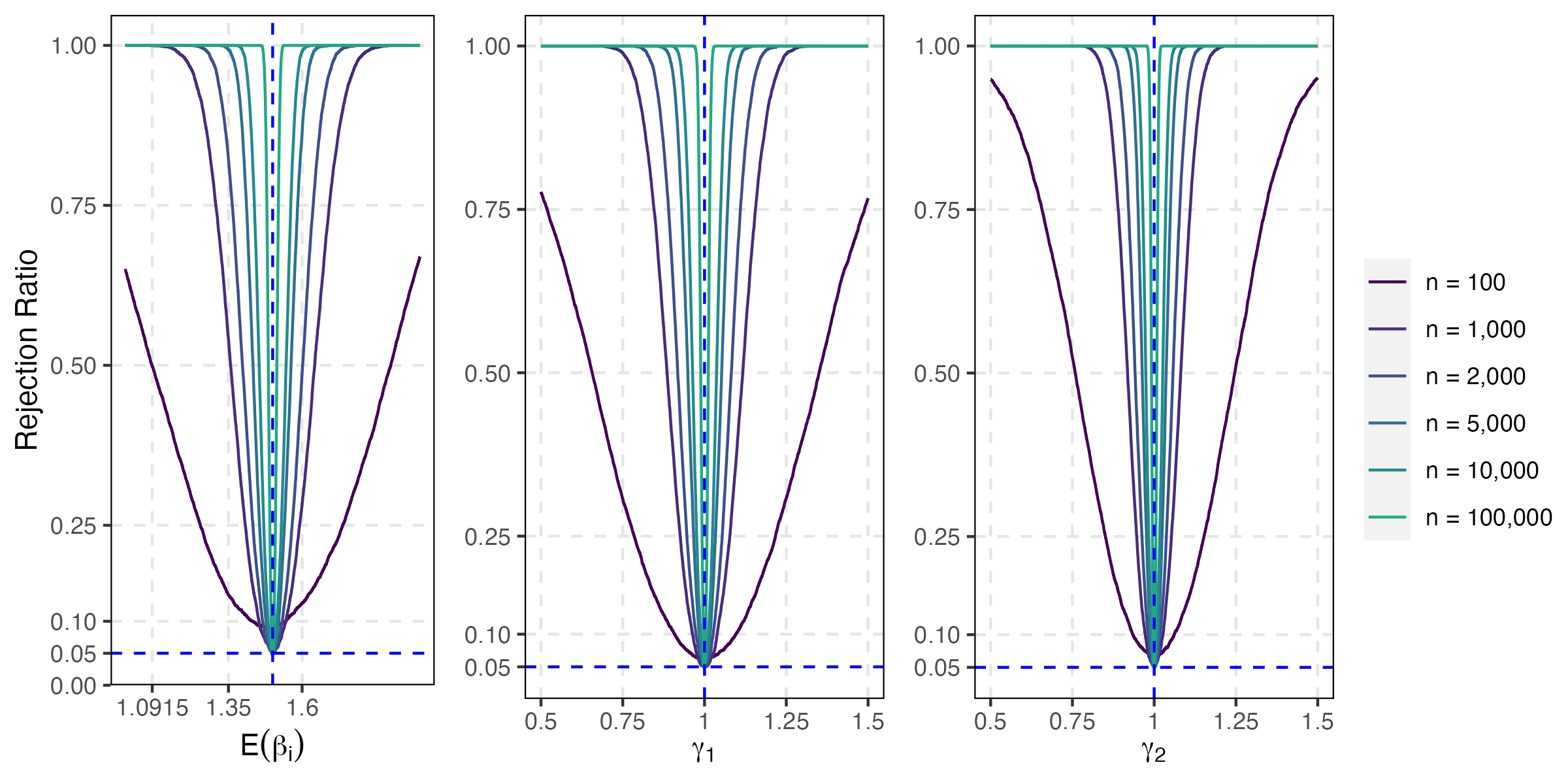}
            \label{fig:power_function_gamma_cat_u_high}
        \end{subfigure}
\end{center}
\par
{\footnotesize \textit{Notes:} The data generating process is %
\eqref{eq:mc_dgp} with \textit{high variance} parametrization that is
described in \eqref{eq:mc_dgp_para}. ``Baseline'', ``Categorical $x$'' and
``Categorical $u$'' refer to DGP 1 to 3 as in Section \ref{subsec:dgp}.
Generically, power is calculated by $R^{-1}\sum_{r=1}^R \mathbf{1}\left[
\left\vert \hat{\theta}^{(r)} - \theta_\delta \right\vert / \hat{\sigma}_{%
\hat{\theta}}^{(r)} > \mathrm{cv}_{0.05} \right] $, for $\theta_\delta$ in a
symmetric neighborhood of the true parameter $\theta_0$, the estimate $\hat{%
\theta}^{(r)}$, the estimated standard error of $\hat{\theta}^{(r)}$, $\hat{%
\sigma}_{\hat{\theta}}^{(r)}$, and the critical value $\mathrm{cv}_{0.05} =
\Phi^{-1}\left( 0.975 \right)$ across $R = 5,000$ replications, where $%
\Phi\left( \cdot \right)$ is the cumulative distribution function of
standard normal distribution. }
\end{figure}

\begin{figure}[tbp]
\caption{Empirical power functions for the least square estimator $\hat{%
\mathbf{\protect\phi}}$ with the \textit{low variance} parametrization ($%
\mathrm{var}\left( \protect\beta_i \right) = 0.15$)}
\label{fig:power_function_gamma_low}
\begin{center}
\begin{subfigure}[b]{\textwidth}
            \centering
            \caption{Baseline}
            \includegraphics[width=\textwidth, height=1.9in]{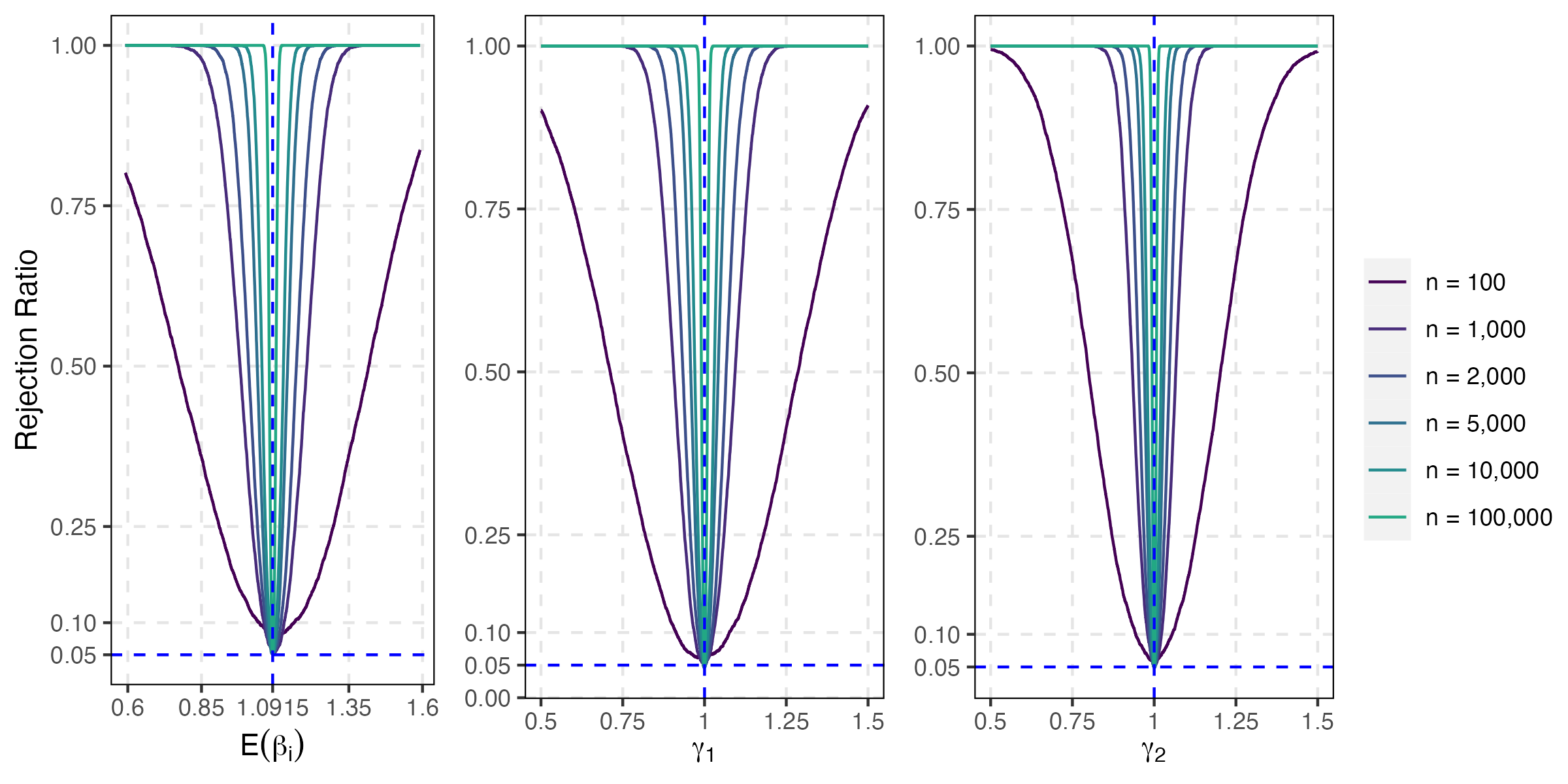}
            \label{fig:power_function_gamma_baseline_low}
        \end{subfigure}\\[0pt]
\begin{subfigure}[b]{\textwidth}
            \centering
            \caption{Categorical $x$}
            \includegraphics[width=\textwidth, height=1.9in]{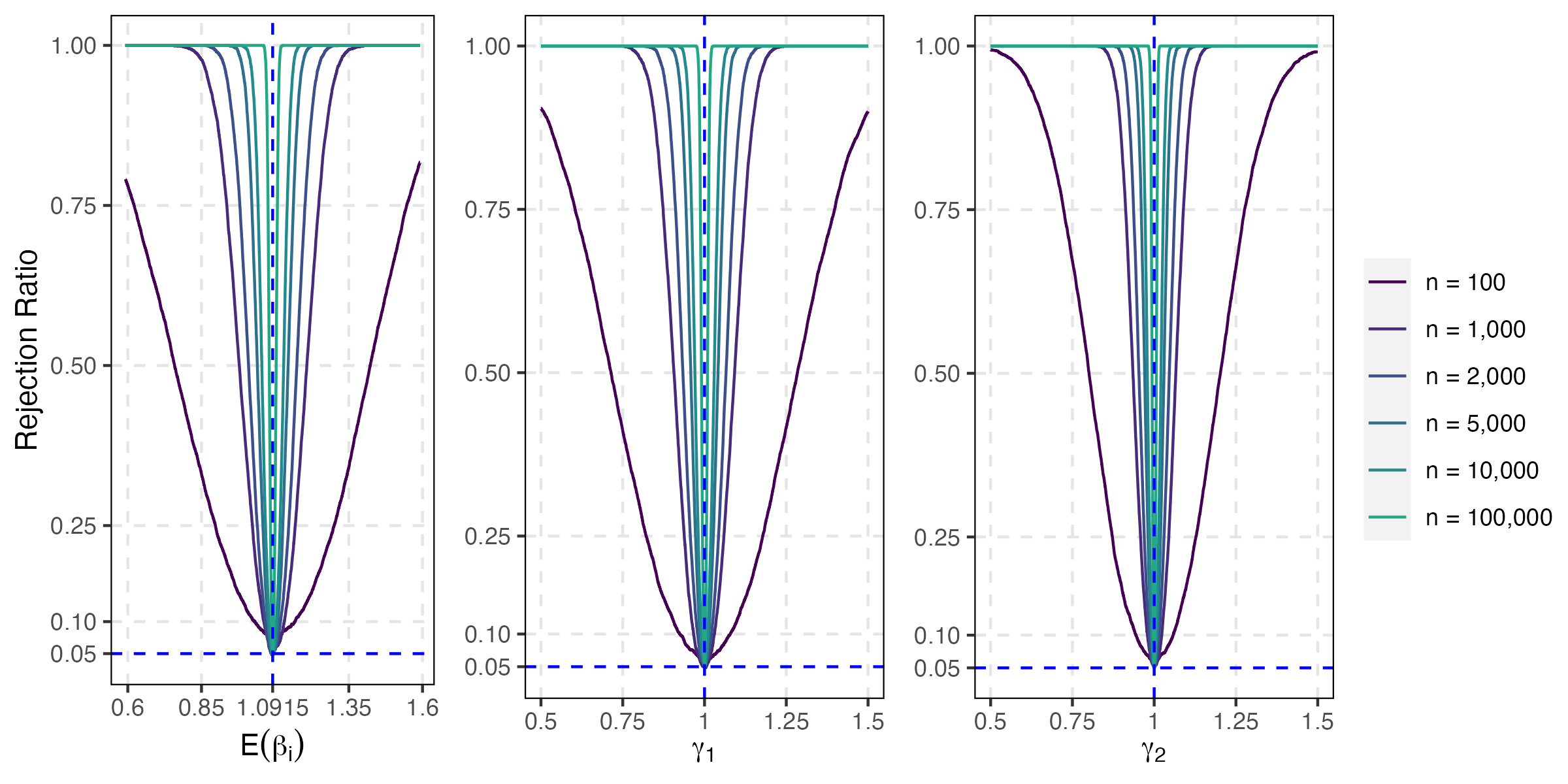}
            \label{fig:power_function_gamma_cat_x_low}
        \end{subfigure}\\[0pt]
\begin{subfigure}[b]{\textwidth}
            \centering
            \caption{Categorical $u$}
            \includegraphics[width=\textwidth, height=1.9in]{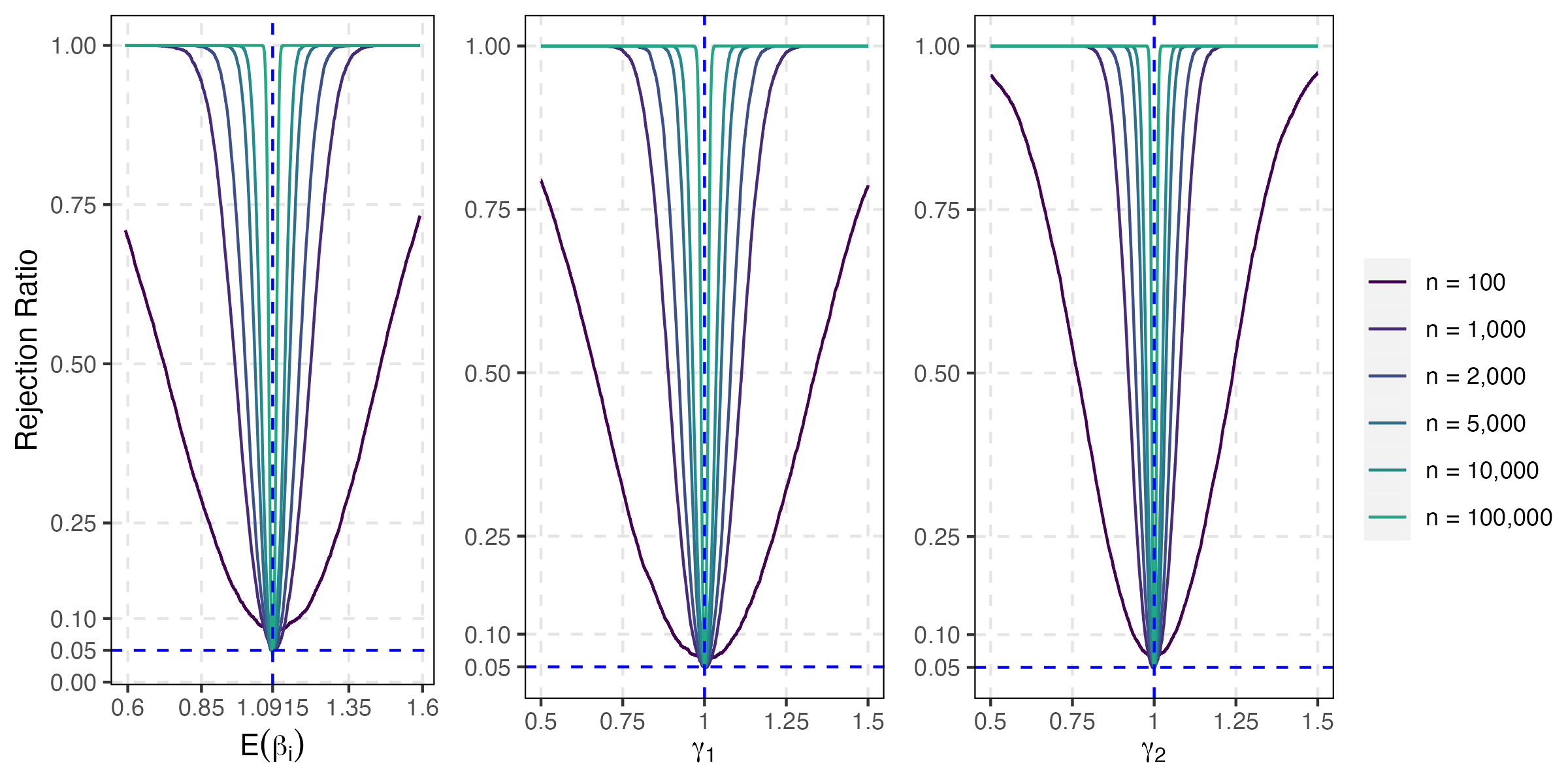}
            \label{fig:power_function_gamma_cat_u_low}
        \end{subfigure}
\end{center}
\par
{\footnotesize \textit{Notes:} The data generating process is %
\eqref{eq:mc_dgp} with \textit{low variance} parametrization that is
described in \eqref{eq:mc_dgp_para}. ``Baseline'', ``Categorical $x$'' and
``Categorical $u$'' refer to DGP 1 to 3 as in Section \ref{subsec:dgp}.
Generically, power is calculated by $R^{-1}\sum_{r=1}^R \mathbf{1}\left[
\left\vert \hat{\theta}^{(r)} - \theta_\delta \right\vert / \hat{\sigma}_{%
\hat{\theta}}^{(r)} > \mathrm{cv}_{0.05} \right] $, for $\theta_\delta$ in a
symmetric neighborhood of the true parameter $\theta_0$, the estimate $\hat{%
\theta}^{(r)}$, the estimated standard error of $\hat{\theta}^{(r)}$, $\hat{%
\sigma}_{\hat{\theta}}^{(r)}$, and the critical value $\mathrm{cv}_{0.05} =
\Phi^{-1}\left( 0.975 \right)$ across $R = 5,000$ replications, where $%
\Phi\left( \cdot \right)$ is the cumulative distribution function of
standard normal distribution. }
\end{figure}

\begin{table}[tbp]
\caption{Bias, RMSE and size of the GMM estimator for distributional
parameters of $\protect\beta$ }
\label{tab:mc_S4}
\begin{center}
{\small \ 
\begin{tabular}{rrrrrrrrrrr}
\hline
\multicolumn{2}{r|}{DGP} & \multicolumn{3}{c|}{Baseline} & 
\multicolumn{3}{c|}{Categorical $x$} & \multicolumn{3}{c}{Categorical $u$}
\\ \hline
\multicolumn{2}{r|}{Sample size $n$} & \multicolumn{1}{c}{Bias} & 
\multicolumn{1}{c}{RMSE} & \multicolumn{1}{c|}{Size} & \multicolumn{1}{c}{
Bias} & \multicolumn{1}{c}{RMSE} & \multicolumn{1}{c|}{Size} & 
\multicolumn{1}{c}{Bias} & \multicolumn{1}{c}{RMSE} & \multicolumn{1}{c}{Size
} \\ \hline
\multicolumn{11}{c}{\textit{high variance}: $\mathrm{var}\left( \beta_i
\right) = 0.25$} \\ \hline
\multirow{7}{*}{\begin{turn}{90} $\pi=0.5$ \end{turn}} & \multicolumn{1}{r|}{
100} & 0.0457 & 0.2291 & \multicolumn{1}{r|}{0.1737} & 0.0363 & 0.2410 & 
\multicolumn{1}{r|}{0.2130} & 0.0235 & 0.2361 & 0.2231 \\ 
& \multicolumn{1}{r|}{1,000} & 0.0018 & 0.1019 & \multicolumn{1}{r|}{0.1308}
& 0.0033 & 0.1178 & \multicolumn{1}{r|}{0.1437} & -0.0270 & 0.1741 & 0.2033
\\ 
& \multicolumn{1}{r|}{2,000} & 0.0017 & 0.0688 & \multicolumn{1}{r|}{0.1084}
& 0.0015 & 0.0826 & \multicolumn{1}{r|}{0.1199} & -0.0174 & 0.1273 & 0.1545
\\ 
& \multicolumn{1}{r|}{5,000} & -0.0003 & 0.0416 & \multicolumn{1}{r|}{0.0936}
& -0.0015 & 0.0495 & \multicolumn{1}{r|}{0.0908} & -0.0089 & 0.0810 & 0.1048
\\ 
& \multicolumn{1}{r|}{10,000} & 0.0002 & 0.0301 & \multicolumn{1}{r|}{0.0774}
& -0.0006 & 0.0351 & \multicolumn{1}{r|}{0.0780} & -0.0052 & 0.0582 & 0.0864
\\ 
& \multicolumn{1}{r|}{100,000} & -0.0001 & 0.0096 & \multicolumn{1}{r|}{
0.0550} & 0.0002 & 0.0114 & \multicolumn{1}{r|}{0.0576} & -0.0009 & 0.0194 & 
0.0582 \\ \hline
\multirow{7}{*}{\begin{turn}{90} $\beta_L = 1$ \end{turn}} & 
\multicolumn{1}{r|}{100} & 0.1415 & 0.4749 & \multicolumn{1}{r|}{0.2472} & 
0.1099 & 0.5110 & \multicolumn{1}{r|}{0.2138} & 0.1151 & 0.5961 & 0.1820 \\ 
& \multicolumn{1}{r|}{1,000} & 0.0207 & 0.1242 & \multicolumn{1}{r|}{0.1501}
& 0.0200 & 0.1454 & \multicolumn{1}{r|}{0.1433} & -0.0256 & 0.2373 & 0.1225
\\ 
& \multicolumn{1}{r|}{2,000} & 0.0129 & 0.0819 & \multicolumn{1}{r|}{0.1344}
& 0.0116 & 0.1007 & \multicolumn{1}{r|}{0.1355} & -0.0094 & 0.1486 & 0.1094
\\ 
& \multicolumn{1}{r|}{5,000} & 0.0048 & 0.0512 & \multicolumn{1}{r|}{0.1052}
& 0.0027 & 0.0607 & \multicolumn{1}{r|}{0.1000} & -0.0053 & 0.0897 & 0.0850
\\ 
& \multicolumn{1}{r|}{10,000} & 0.0031 & 0.0365 & \multicolumn{1}{r|}{0.0854}
& 0.0021 & 0.0428 & \multicolumn{1}{r|}{0.0900} & -0.0020 & 0.0633 & 0.0714
\\ 
& \multicolumn{1}{r|}{100,000} & 0.0002 & 0.0112 & \multicolumn{1}{r|}{0.0534
} & 0.0007 & 0.0135 & \multicolumn{1}{r|}{0.0584} & -0.0002 & 0.0207 & 0.0574
\\ \hline
\multirow{7}{*}{\begin{turn}{90} $\beta_H = 2$ \end{turn}} & 
\multicolumn{1}{r|}{100} & -0.0996 & 0.5609 & \multicolumn{1}{r|}{0.2014} & 
-0.0873 & 0.6154 & \multicolumn{1}{r|}{0.1963} & -0.1071 & 0.6996 & 0.1866
\\ 
& \multicolumn{1}{r|}{1,000} & -0.0193 & 0.1407 & \multicolumn{1}{r|}{0.1864}
& -0.0128 & 0.1581 & \multicolumn{1}{r|}{0.1661} & -0.0319 & 0.2400 & 0.2093
\\ 
& \multicolumn{1}{r|}{2,000} & -0.0099 & 0.0893 & \multicolumn{1}{r|}{0.1486}
& -0.0099 & 0.1094 & \multicolumn{1}{r|}{0.1467} & -0.0239 & 0.1663 & 0.1673
\\ 
& \multicolumn{1}{r|}{5,000} & -0.0053 & 0.0519 & \multicolumn{1}{r|}{0.1092}
& -0.0072 & 0.0622 & \multicolumn{1}{r|}{0.1082} & -0.0127 & 0.1019 & 0.1156
\\ 
& \multicolumn{1}{r|}{10,000} & -0.0020 & 0.0362 & \multicolumn{1}{r|}{0.0878
} & -0.0033 & 0.0430 & \multicolumn{1}{r|}{0.0880} & -0.0080 & 0.0718 & 
0.0986 \\ 
& \multicolumn{1}{r|}{100,000} & -0.0005 & 0.0114 & \multicolumn{1}{r|}{
0.0530} & -0.0003 & 0.0134 & \multicolumn{1}{r|}{0.0548} & -0.0017 & 0.0236
& 0.0646 \\ \hline
\multicolumn{11}{c}{\textit{low variance}: $\mathrm{var}\left( \beta_i
\right) = 0.15$} \\ \hline
\multirow{7}{*}{\begin{turn}{90} $\pi=0.3$ \end{turn}} & \multicolumn{1}{r|}{
100} & 0.2175 & 0.3084 & \multicolumn{1}{r|}{0.2183} & 0.2227 & 0.3187 & 
\multicolumn{1}{r|}{0.2464} & 0.2294 & 0.3157 & 0.2500 \\ 
& \multicolumn{1}{r|}{1,000} & 0.0170 & 0.1536 & \multicolumn{1}{r|}{0.1873}
& 0.0307 & 0.1837 & \multicolumn{1}{r|}{0.2063} & 0.0511 & 0.2295 & 0.2493
\\ 
& \multicolumn{1}{r|}{2,000} & 0.0014 & 0.1010 & \multicolumn{1}{r|}{0.1426}
& 0.0105 & 0.1290 & \multicolumn{1}{r|}{0.1601} & 0.0181 & 0.1815 & 0.2102
\\ 
& \multicolumn{1}{r|}{5,000} & -0.0002 & 0.0590 & \multicolumn{1}{r|}{0.1084}
& 0.0010 & 0.0737 & \multicolumn{1}{r|}{0.1158} & 0.0085 & 0.1232 & 0.1468
\\ 
& \multicolumn{1}{r|}{10,000} & -0.0001 & 0.0415 & \multicolumn{1}{r|}{0.0894
} & 0.0005 & 0.0515 & \multicolumn{1}{r|}{0.0928} & 0.0067 & 0.0906 & 0.1046
\\ 
& \multicolumn{1}{r|}{100,000} & -0.0001 & 0.0129 & \multicolumn{1}{r|}{
0.0594} & 0.0003 & 0.0158 & \multicolumn{1}{r|}{0.0536} & 0.0108 & 0.0349 & 
0.0776 \\ \hline
\multirow{7}{*}{\begin{turn}{90} $\beta_L = 0.5$ \end{turn}} & 
\multicolumn{1}{r|}{100} & 0.3365 & 0.5905 & \multicolumn{1}{r|}{0.2426} & 
0.3153 & 0.6042 & \multicolumn{1}{r|}{0.2432} & 0.3384 & 0.6746 & 0.2005 \\ 
& \multicolumn{1}{r|}{1,000} & 0.0352 & 0.2334 & \multicolumn{1}{r|}{0.1560}
& 0.0290 & 0.2813 & \multicolumn{1}{r|}{0.1544} & 0.0131 & 0.4141 & 0.1233
\\ 
& \multicolumn{1}{r|}{2,000} & 0.0175 & 0.1414 & \multicolumn{1}{r|}{0.1310}
& 0.0131 & 0.1835 & \multicolumn{1}{r|}{0.1382} & -0.0157 & 0.2988 & 0.1037
\\ 
& \multicolumn{1}{r|}{5,000} & 0.0085 & 0.0830 & \multicolumn{1}{r|}{0.1082}
& 0.0041 & 0.1052 & \multicolumn{1}{r|}{0.1118} & -0.0057 & 0.1798 & 0.0928
\\ 
& \multicolumn{1}{r|}{10,000} & 0.0055 & 0.0577 & \multicolumn{1}{r|}{0.0966}
& 0.0031 & 0.0730 & \multicolumn{1}{r|}{0.0934} & 0.0019 & 0.1231 & 0.0760
\\ 
& \multicolumn{1}{r|}{100,000} & 0.0005 & 0.0180 & \multicolumn{1}{r|}{0.0596
} & 0.0011 & 0.0222 & \multicolumn{1}{r|}{0.0582} & 0.0130 & 0.0443 & 0.0962
\\ \hline
\multirow{7}{*}{\begin{turn}{90} $\beta_H = 1.345$ \end{turn}} & 
\multicolumn{1}{r|}{100} & 0.0023 & 0.4727 & \multicolumn{1}{r|}{0.1377} & 
0.0238 & 0.5290 & \multicolumn{1}{r|}{0.1453} & 0.0185 & 0.6500 & 0.1461 \\ 
& \multicolumn{1}{r|}{1,000} & -0.0081 & 0.1265 & \multicolumn{1}{r|}{0.1737}
& 0.0042 & 0.1621 & \multicolumn{1}{r|}{0.1655} & 0.0120 & 0.2353 & 0.1738
\\ 
& \multicolumn{1}{r|}{2,000} & -0.0092 & 0.0828 & \multicolumn{1}{r|}{0.1428}
& -0.0026 & 0.1045 & \multicolumn{1}{r|}{0.1475} & 0.0029 & 0.1607 & 0.1710
\\ 
& \multicolumn{1}{r|}{5,000} & -0.0048 & 0.0489 & \multicolumn{1}{r|}{0.1028}
& -0.0041 & 0.0586 & \multicolumn{1}{r|}{0.1034} & 0.0006 & 0.0970 & 0.1172
\\ 
& \multicolumn{1}{r|}{10,000} & -0.0025 & 0.0340 & \multicolumn{1}{r|}{0.0808
} & -0.0024 & 0.0412 & \multicolumn{1}{r|}{0.0942} & 0.0019 & 0.0706 & 0.0958
\\ 
& \multicolumn{1}{r|}{100,000} & -0.0004 & 0.0105 & \multicolumn{1}{r|}{
0.0486} & -0.0002 & 0.0125 & \multicolumn{1}{r|}{0.0548} & 0.0073 & 0.0262 & 
0.0696 \\ \hline
\end{tabular}
}
\end{center}
\par
{\footnotesize \textit{Notes:} The data generating process is %
\eqref{eq:mc_dgp}. \textit{high variance} and \textit{low variance}
parametrization are described in \eqref{eq:mc_dgp_para}. ``Baseline'',
``Categorical $x$'' and ``Categorical $u$'' refer to DGP 1 to 3 as in
Section \ref{subsec:dgp}. Generically, bias, RMSE and size are calculated by 
$R^{-1}\sum_{r=1}^R \left( \hat{ \theta}^{(r)} - \theta_0 \right)$, $\sqrt{%
R^{-1}\sum_{r= 1}^R \left( \hat{\theta}^{(r)} -\theta_0 \right)^2}$, and $%
R^{-1}\sum_{r=1}^R \mathbf{1}\left[ \left\vert \hat{\theta}^{(r)} - \theta_0
\right\vert / \hat{\sigma}_{\hat{\theta}}^{(r)} > \mathrm{cv}_{0.05} \right] 
$, respectively, for true parameter $\theta_0$, its estimate $\hat{\theta}%
^{(r)}$, the estimated standard error of $\hat{\theta}^{(r)}$, $\hat{\sigma}%
_{\hat{\theta}}^{(r)}$, and the critical value $\mathrm{cv}_{0.05} =
\Phi^{-1}\left( 0.975 \right)$ across $R = 5,000$ replications, where $%
\Phi\left( \cdot \right)$ is the cumulative distribution function of
standard normal distribution. }
\end{table}

\begin{figure}[tbp]
\caption{Empirical power functions for the GMM estimator of distributional
parameters of $\protect\beta$ with the \textit{high variance}
parametrization($\mathrm{var}\left( \protect\beta_i \right) = 0.25$)}
\label{fig:power_function_S4_high}
\begin{center}
\begin{subfigure}[b]{\textwidth}
            \centering
            \caption{Baseline}
            \includegraphics[width=\textwidth, height=1.9in]{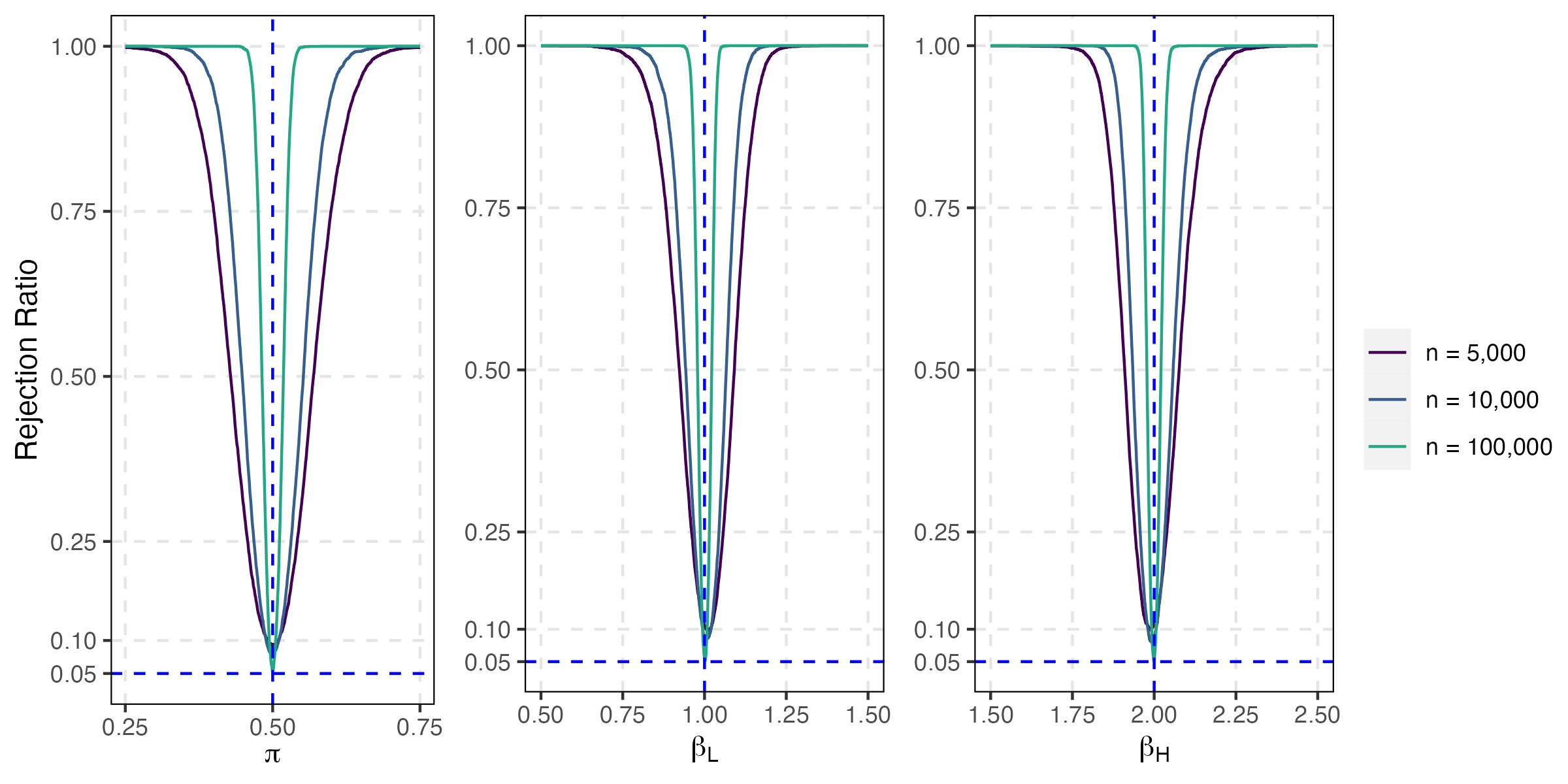}
            \label{fig:power_function_S4_baseline_high}
        \end{subfigure}\\[0pt]
\begin{subfigure}[b]{\textwidth}
            \centering
            \caption{Categorical $x$}
            \includegraphics[width=\textwidth, height=1.9in]{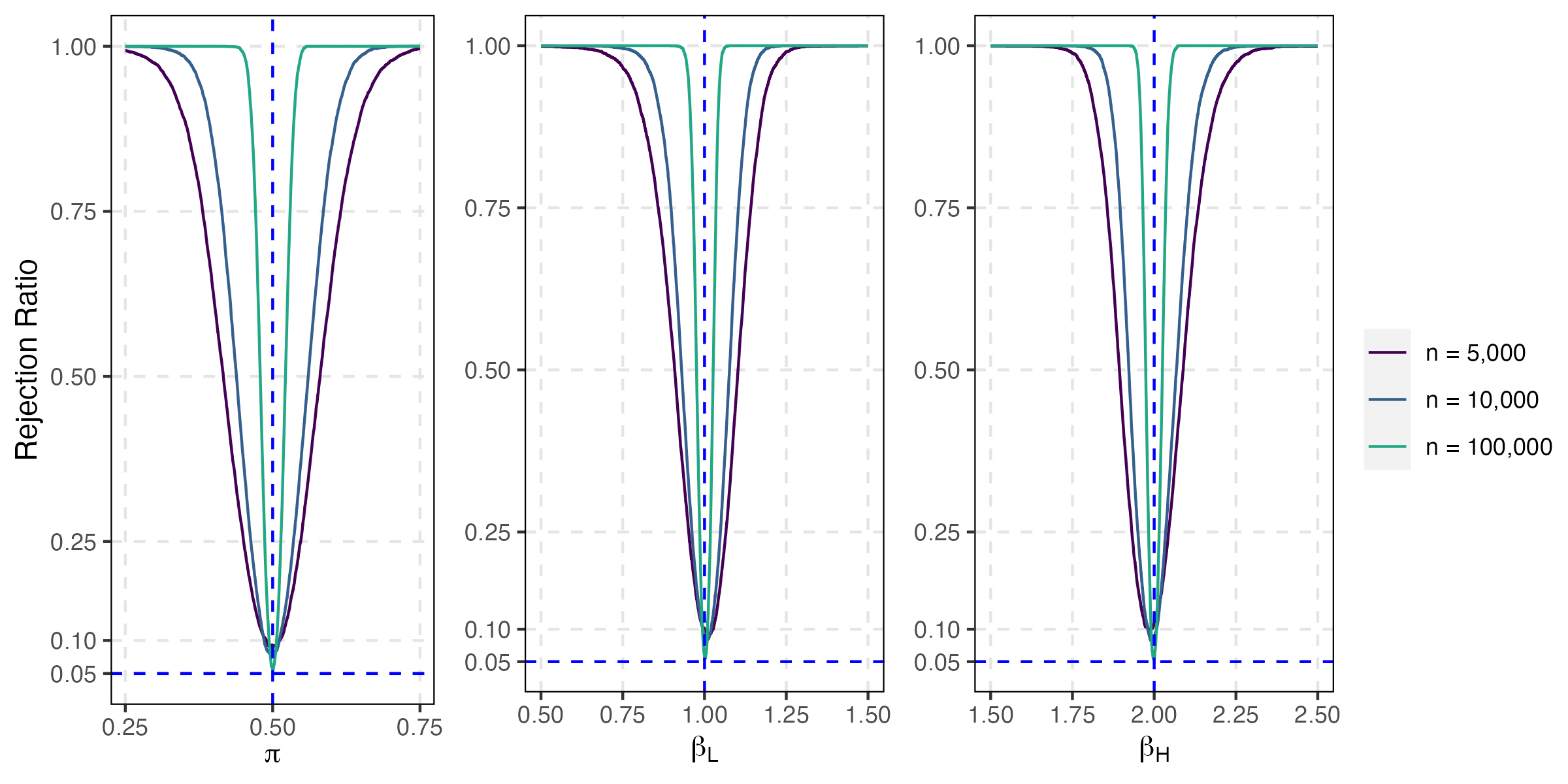}
            \label{fig:power_function_S4_cat_x_high}
        \end{subfigure}\\[0pt]
\begin{subfigure}[b]{\textwidth}
            \centering
            \caption{Categorical $u$}
            \includegraphics[width=\textwidth, height=1.9in]{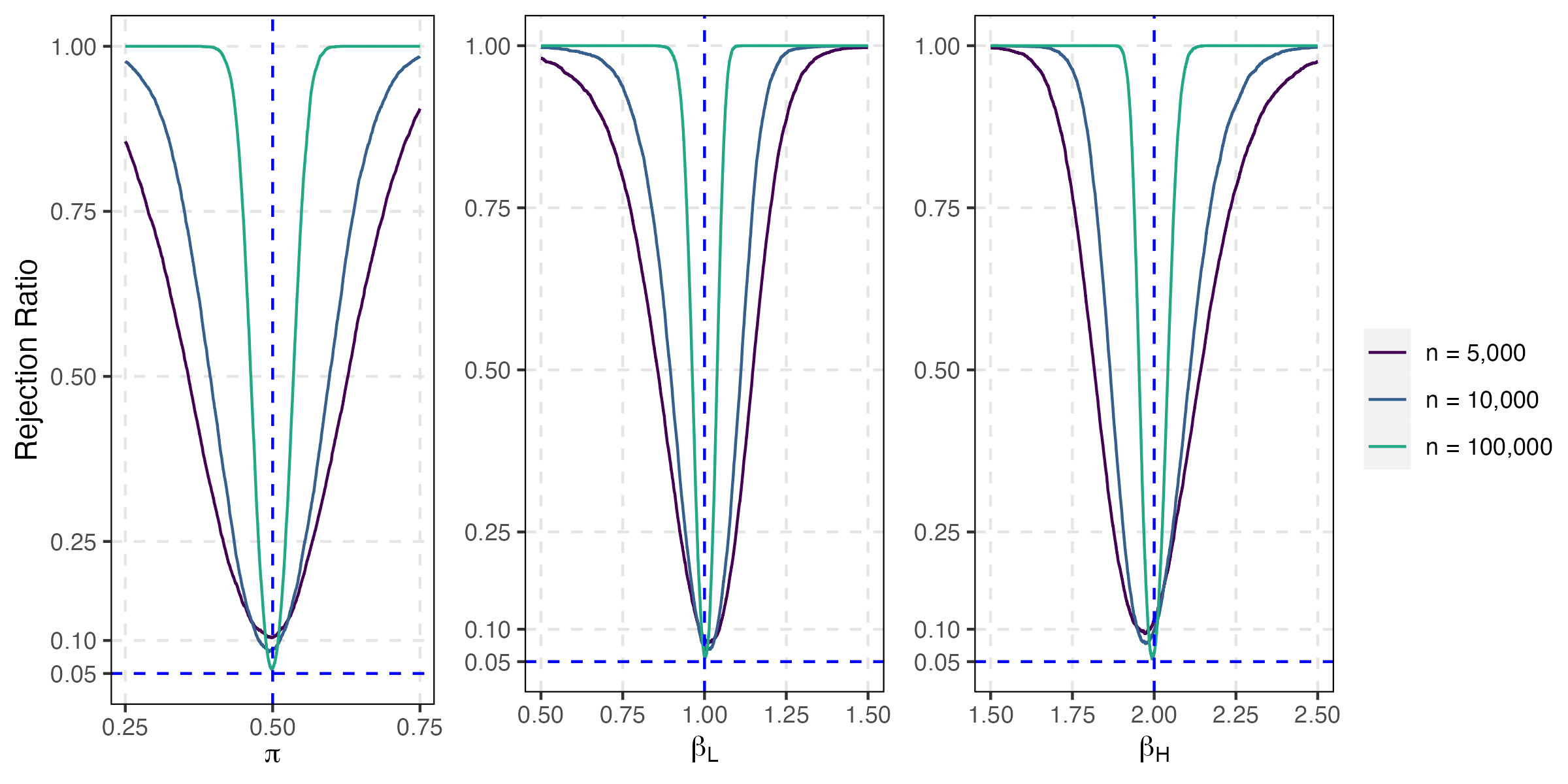}
            \label{fig:power_function_S4_cat_u_high}
        \end{subfigure}
\end{center}
\par
{\footnotesize \textit{Notes:} The data generating process is %
\eqref{eq:mc_dgp} with \textit{high variance} parametrization that is
described in \eqref{eq:mc_dgp_para}. ``Baseline'', ``Categorical $x$'' and
``Categorical $u$'' refer to DGP 1 to 3 as in Section \ref{subsec:dgp}. The
model is estimated with $S = 4$, the highest order of moments of $x_i$ used
in estimation. Generically, power is calculated by $R^{-1}\sum_{r=1}^R 
\mathbf{1}\left[ \left\vert \hat{\theta}^{(r)} - \theta_\delta \right\vert / 
\hat{\sigma}_{\hat{\theta}}^{(r)} > \mathrm{cv}_{0.05} \right] $, for $%
\theta_\delta$ in a symmetric neighborhood of the true parameter $\theta_0$,
the estimate $\hat{\theta}^{(r)}$, the estimated standard error of $\hat{%
\theta}^{(r)}$, $\hat{\sigma}_{\hat{\theta}}^{(r)}$, and the critical value $%
\mathrm{cv}_{0.05} = \Phi^{-1}\left( 0.975 \right)$ across $R = 5,000$
replications, where $\Phi\left( \cdot \right)$ is the cumulative
distribution function of standard normal distribution. }
\end{figure}

\begin{figure}[tbp]
\caption{Empirical power functions for the GMM estimator of distributional
parameters of $\protect\beta$ with the \textit{low variance} parametrization
($\mathrm{var}\left( \protect\beta_i \right) = 0.15$)}
\label{fig:power_function_S4_low}
\begin{center}
\begin{subfigure}[b]{\textwidth}
            \centering
            \caption{Baseline}
            \includegraphics[width=\textwidth, height=1.9in]{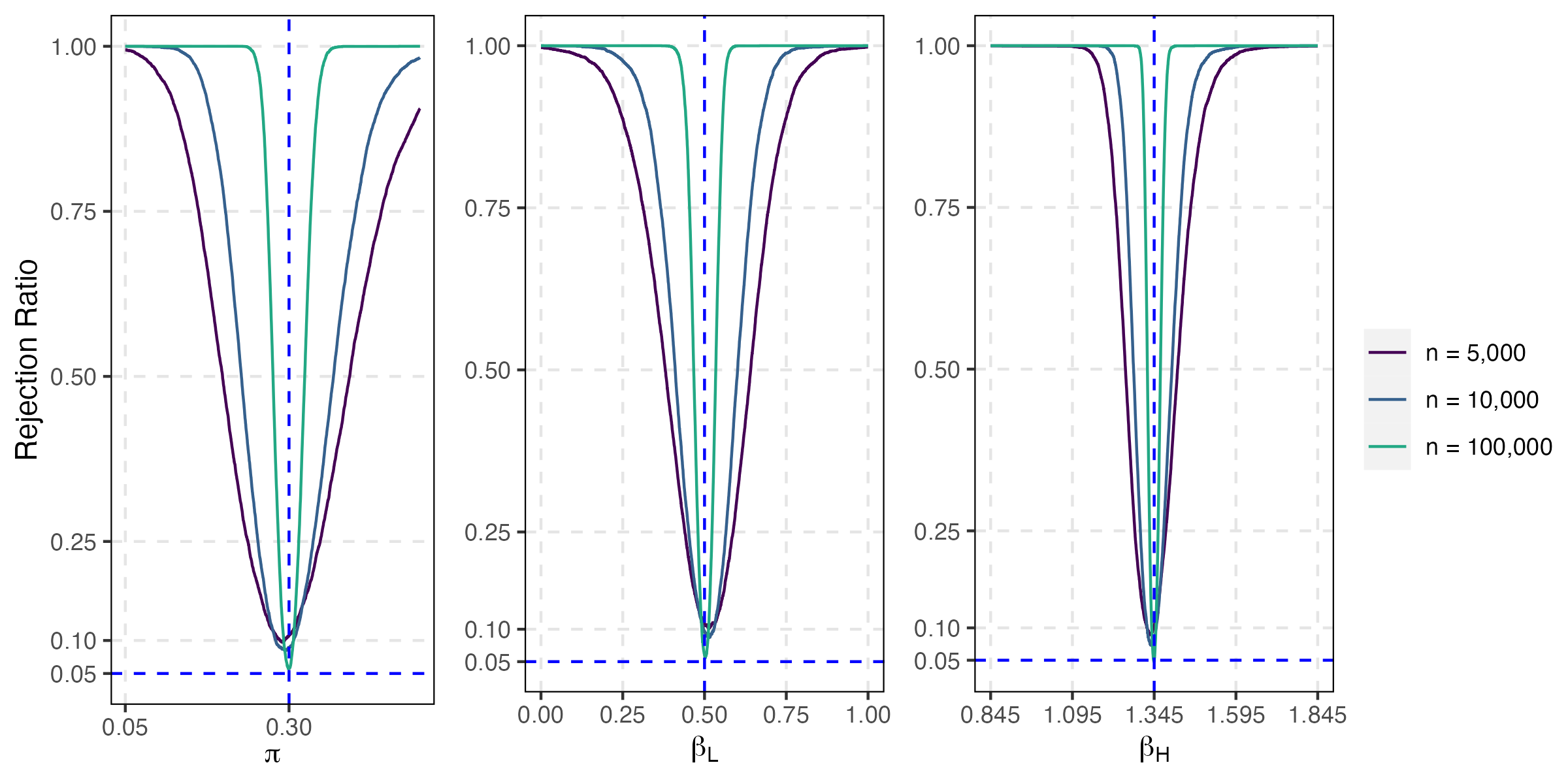}
            \label{fig:power_function_S4_baseline_low}
        \end{subfigure}\\[0pt]
\begin{subfigure}[b]{\textwidth}
            \centering
            \caption{Categorical $x$}
            \includegraphics[width=\textwidth, height=1.9in]{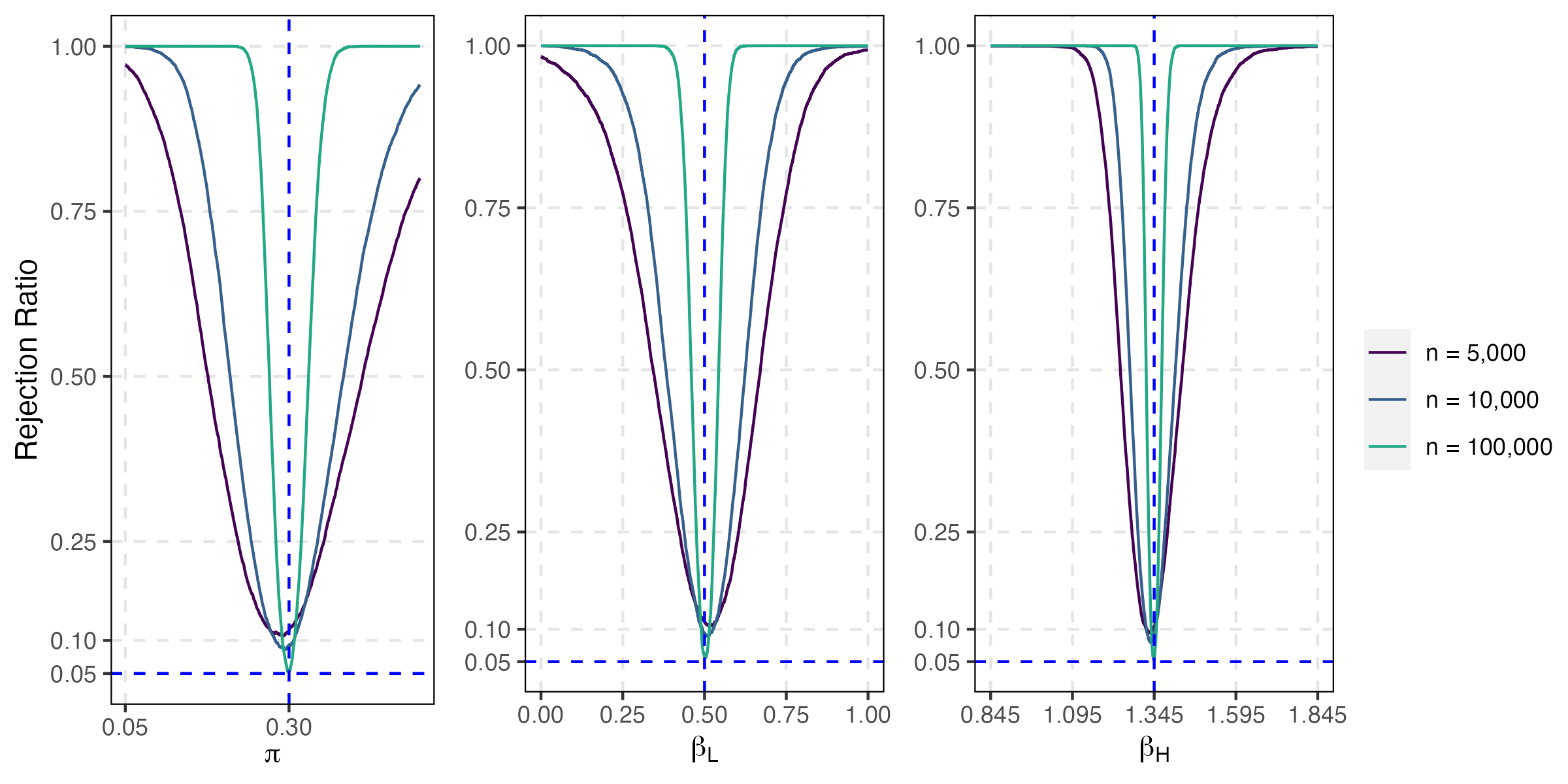}
            \label{fig:power_function_S4_cat_x_low}
        \end{subfigure}\\[0pt]
\begin{subfigure}[b]{\textwidth}
            \centering
            \caption{Categorical $u$}
            \includegraphics[width=\textwidth, height=1.9in]{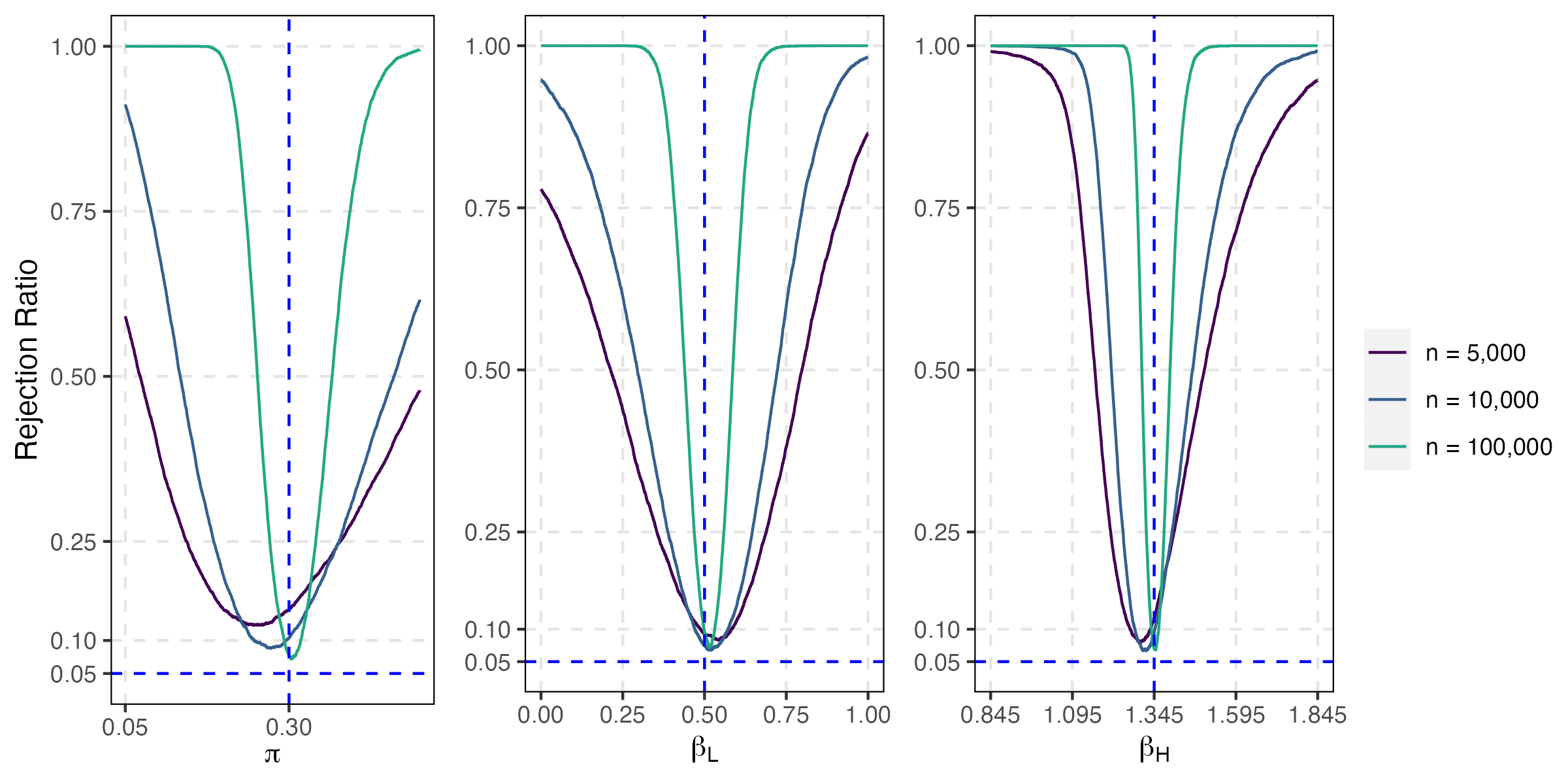}
            \label{fig:power_function_S4_cat_u_low}
        \end{subfigure}
\end{center}
\par
{\footnotesize \textit{Notes:} The data generating process is %
\eqref{eq:mc_dgp} with \textit{low variance} parametrization that is
described in \eqref{eq:mc_dgp_para}. ``Baseline'', ``Categorical $x$'' and
``Categorical $u$'' refer to DGP 1 to 3 as in Section \ref{subsec:dgp}. The
model is estimated with $S = 4$, the highest order of moments of $x_i$ used
in estimation. Generically, power is calculated by $R^{-1}\sum_{r=1}^R 
\mathbf{1}\left[ \left\vert \hat{\theta}^{(r)} - \theta_\delta \right\vert / 
\hat{\sigma}_{\hat{\theta}}^{(r)} > \mathrm{cv}_{0.05} \right] $, for $%
\theta_\delta$ in a symmetric neighborhood of the true parameter $\theta_0$,
the estimate $\hat{\theta}^{(r)}$, the estimated standard error of $\hat{%
\theta}^{(r)}$, $\hat{\sigma}_{\hat{\theta}}^{(r)}$, and the critical value $%
\mathrm{cv}_{0.05} = \Phi^{-1}\left( 0.975 \right)$ across $R = 5,000$
replications, where $\Phi\left( \cdot \right)$ is the cumulative
distribution function of standard normal distribution. }
\end{figure}

For each sample size $n = 100$, $1,000$, $2,000$, $5,000$, $10,000$ and $%
100,000$ we run $5,000$ replications of experiments for DGP 1 (baseline),
DGP 2 (categorical $x$) and DGP 3 (categorical $u$) with \textit{high
variance} and \textit{low variance} parametrization, as set out in %
\eqref{eq:mc_dgp_para}.

We first investigate the finite sample performance of $\hat{\mathbf{\phi }}$%
, as an estimator of $\mathbf{\phi }=\left( \mathrm{E}\left( \beta
_{i}\right) ,\mathbf{\gamma }^{\prime }\right) ^{\prime }$. Bias, root mean
squared errors (RMSE) for estimation of $\mathrm{E}\left( \beta _{i}\right) $%
, $\gamma _{1}$ and $\gamma _{2}$, as well as size of testing of the null
values at the 5 per cent nominal value are reported in Table \ref%
{tab:mc_bias_rmse_gamma}. In addition, we plot the associated empirical
power functions in Figure \ref{fig:power_function_gamma_high} and \ref%
{fig:power_function_gamma_low}, for cases of high and low $var(\beta _{i})$.
The results show that $\hat{\mathbf{\phi }}$ has very good small sample
properties with small bias and RMSEs, with size very close to the nominal
value of 5 per cent across all DGPs and parametrization, even when sample
size is relatively small. The power of the test increases steadily as the
sample size increases.

Then, we turn to the GMM estimator for the distributional parameters of $%
\beta _{i}$ proposed in Section \ref{subsec:estimation_beta}. The bias,
RMSE, and the test size based on the asymptotic distribution given in
Theorem \ref{thm:normality}, for $\pi $, $\beta _{L}$ and $\beta _{H}$, are
reported in Table \ref{tab:mc_S4}. The empirical power functions are
reported in Figure \ref{fig:power_function_S4_high} and \ref%
{fig:power_function_S4_low}. The reported results are based on $S=4$, where $%
S$ $(>2K-1=3)$ denotes the highest order of moments of $x_{i}$ included in
estimation.\footnote{%
We also tried estimation based on a larger number of moments (using $S=5$
and $S=6$). In the case of current Monte Carlo results, adding more moments
does not seem to add much to the precision of the estimates and could be
counter-productive when $n$ is not sufficiently large. The results are
available in Section \ref{subsec:mc_s5_s6} in the online supplement.}

The upper panel of this table reports the results of the high variance and
the lower panel for the low variance parametrization, as set out in (\ref%
{eq:mc_dgp_para}). For all parameters and under all DGPs, the bias and RMSE
decline steadily with the sample size as predicted by Theorem \ref%
{thm:consisteny}, and confirm the robustness of the GMM estimates to the
heterogeneity in the regressor and the error processes. But for a given
sample size, the relative precision of the estimates depends on the
variability of $\beta _{i}$, as characterized by the true value of $\mathrm{%
var} (\beta _{i})$. The precision of the estimates with \textit{high variance%
} parametrization is relatively higher than that with \textit{low variance}
parametrization. This is to be expected since, unlike $\mathrm{E} (\beta
_{i}),$ the distributional parameters are only identified if $\mathrm{var}
(\beta _{i})>0$. As shown in \eqref{sys_3} and \eqref{sys_4} for the current
case of $K=2$, $\mathrm{var}(\beta _{i})$ is in the denominator when we
recover the distributional parameters from the moments of $\beta _{i}$. When 
$\mathrm{var}(\beta _{i})$ is small, estimation errors in the moments of $%
\beta _{i}$ can be amplified in the estimation of $\pi $, $\beta _{L}$ and $%
\beta _{H}$. On the other hand, the larger the variance the more precisely $%
\pi $, $\beta _{H}$ and $\beta _{L}$ can be estimated for a given $n$.%
\footnote{%
Section \ref{subsec:mc_highvar} in the online supplement presents
parametrization with $\mathrm{var}\left( \beta_i \right) = 6.35$ and $18.95$%
, which further confirms the pattern that the larger the variance the more
precisely $\pi $, $\beta _{H}$ and $\beta _{L}$ can be estimated for a given 
$n$.} The size and power also depends on the parametrization. With both 
\textit{high variance} and \textit{low variance} parametrization, we can
achieve correct size and reasonable power when $n$ is quite large ($%
n=100,000 $). We plot the empirical power functions for $n\geq 5,000$ for $%
\pi $, $\beta _{H}$ and $\beta _{L}$ since the size is far above 5 per cent
for smaller values of $n$, and power comparisons are not meaningful in such
cases.

\begin{remark}
\label{rem:practical} Note that GMM estimators of moments of $\beta _{i}$,
namely $\mathbf{m}_{\mathbf{\beta }}$, can be obtained using the moment
conditions in \eqref{eq:mc_limit},and the transformations $\mathbf{m}_{%
\mathbf{\beta }}=h\left( \mathbf{\theta }\right) $ in \eqref{eq:mbeta_mat}
are required only to derive the estimators of $\mathbf{\theta }$, the
parameters of the underlying categorical distribution. The Monte Carlo
results in Section \ref{subsec:mc_gmm_moments} in the online supplement show
that $\mathbf{m}_{\mathbf{\beta }}$ can be accurately estimated with
relatively small sample sizes. In the estimation of both $\mathbf{m}_{%
\mathbf{\beta }}$ and $\mathbf{\theta }$, the same set of moment conditions
are included, so the estimation of distributional parameters $\mathbf{\theta 
}$ essentially relies on the relation $\mathbf{\theta }=h^{-1}\left( \mathbf{%
m}_{\mathbf{\beta }}\right) $. Sampling uncertainties in the estimation of $%
\mathbf{m}_{\mathbf{\beta }}$, particularly in higher order moments, are
potentially amplified through the inverse transformation $h^{-1}$ that
involves matrix inversion, which causes the difficulties in estimation and
inference of $\mathbf{\theta }$ when sample sizes are small. This is
analogous to the problem of precision matrix estimation from an estimated
covariance matrix. In practice, estimation of the categorical parameters is
recommended for applications where the sample size is relatively large,
otherwise it is advisable to focus on estimates of the lower order moments
of $\beta _{i}$.
\end{remark}


\section{Heterogeneous return to education: An empirical application \label%
{sec:Empirical-Application}}

Since the pioneering work by \citet{becker1962investment,becker1964book} on
the effects of investments in human capital, estimating returns to education
has been one of the focal points of labor economics research. In his
pioneering contribution \citet{mincer1974book} models the logarithm of
earnings as a function of years of education and years of potential labor
market experience (age minus years of education minus six), which can be
written in a generic form: 
\begin{equation}
\log \text{wage}_{i}=\alpha _{i}+\beta _{i}\text{edu}_{i}+\phi \left( 
\mathbf{z}_{i}\right) +\varepsilon _{i},  \label{eq:heckman2018eq1}
\end{equation}%
as in \citet*[Equation (1)]{heckman2018returns}, where $\mathbf{z}_{i}$
includes the labor market experience and other relevant control variables.
The above wage equation, also known as the \textquotedblleft Mincer
equation\textquotedblright, has become of the workhorse of the empirical
works on estimating the return to education. In the most widely used
specification of the Mincer equation \eqref{eq:heckman2018eq1}, 
\begin{equation*}
\phi \left( \mathbf{z}_{i}\right) =\rho _{1}\text{exper}_{i}+\rho _{2}\text{%
exper}_{i}^{2}+\tilde{\mathbf{z}}_{i}^{\prime }\tilde{\mathbf{\gamma }},
\end{equation*}%
where $\tilde{\mathbf{z}}_{i}$ is the vector of control variables other than
potential labor market experience.

Along with the advancement of empirical research on this topic, there has
been a growing awareness of the importance of heterogeneity in individual
cognitive and non-cognitive abilities \citep{heckman2001jpenobellecture} and
their significance for explaining the observed heterogeneity in return to
education. Accordingly, it is important to allow the parameters of the wage
equation to differ across individuals. In equation \eqref{eq:heckman2018eq1}
we allow $\alpha _{i}$ and $\beta _{i}$ to differ across individuals, but
assume that $\phi \left( \mathbf{z}_{i}\right) $ can be approximated as
non-linear functions of experience and other control variables with
homogeneous coefficients.

Specifically, following \citet{lemieux2006postnber,lemieux2006postsecondary}
we also allow for time variations in the parameters of the wage equation and
consider the following categorical coefficient model over a given
cross-section sample indexed by $t$:\footnote{%
Some investigators have suggested including higher powers of the experience
variable in the wage equation. \citet{lemieux2006mincerbookchapter}, for
example, proposes using a quartic rather than a quadratic function. As a
robustness check we also provide estimation results with quartic experience
specification in Section \ref{sec:empirical_supp} in the online supplement.} 
\begin{equation}
\log \text{wage}_{it}=\alpha _{it}+\beta _{it}\text{edu}_{it}+\rho _{1t}%
\text{exper}_{it}+\rho _{2t}\text{exper}_{it}^{2}+\tilde{\mathbf{z}}%
_{it}^{\prime }\tilde{\mathbf{\gamma }}_{t}+\varepsilon _{it},
\label{eq:ccrm_spec_alpha_i}
\end{equation}%
where the return to education follows the categorical distribution, 
\begin{equation*}
\beta _{it}=%
\begin{cases}
b_{tL} & \text{w.p. }\pi _{t}, \\ 
b_{tH} & \text{w.p. }1-\pi _{t},%
\end{cases}%
\end{equation*}%
and $\tilde{\mathbf{z}}_{it}$ includes gender, martial status and race. $%
\alpha _{it}=\alpha _{t}+\delta _{it}$ where $\delta _{it}$ is mean $0$
random variable assumed to be distributed independently of $\text{edu}_{it}$
and $\mathbf{z}_{it}=\left( \text{exper}_{it},\text{exper}_{it}^{2}\text{, }%
\tilde{\mathbf{z}}_{t}^{\prime }\right) ^{\prime }$. Let $u_{it}=\varepsilon
_{it}+\delta _{it}$, and write \eqref{eq:ccrm_spec_alpha_i} as 
\begin{equation}
\log \text{wage}_{it}=\alpha _{t}+\beta _{it}\text{edu}_{it}+\rho _{1t}\text{%
exper}_{it}+\rho _{2t}\text{exper}_{it}^{2}+\tilde{\mathbf{z}}_{it}^{\prime }%
\tilde{\mathbf{\gamma }}_{t}+u_{it}.  \label{eq:ccrm_spec}
\end{equation}%
The correlation between $\alpha _{it}$ and $\text{edu}_{it}$ in %
\eqref{eq:heckman2018eq1} is the source of \textquotedblleft ability
bias\textquotedblright\ \citep{griliches1977estimating}. Given the pure
cross-sectional nature of our analysis, we do not allow for the endogeneity
from \textquotedblleft ability bias\textquotedblright\ or dynamics. To allow
for non-zero correlations between $\alpha _{it}$, edu$_{it}$ and $\mathbf{z}%
_{it}$, a panel data approach is required, which has its own challenges, as
education and experience variables tend to very slow moving (if at all) for
many individuals in the panel. Time delays between changes in education and
experience, and the wage outcomes also further complicate the interpretation
of the mean estimates of $\beta _{it}$ which we shall be reporting. To
partially address the possible dynamic spillover effects, we provide
estimates of the distribution of $\beta _{it}$ using cross-sectional data
from two different sample periods, and investigate the extent to which the
distribution of return to education has changed over time, by gender and the
level of educational achievements.\footnote{%
Time variations in return to education has also been investigated in the
literature as a possible explanation of increasing wage inequality in the
U.S. See, for example, the papers by %
\citet{lemieux2006postnber,lemieux2006postsecondary}.}

We estimate the categorical distribution of the return to education in %
\eqref{eq:ccrm_spec} using the May and Outgoing Rotation Group (ORG)
supplements of the Current Population Survey (CPS) data, as in %
\citet{lemieux2006postnber,lemieux2006postsecondary}.\footnote{%
The data is retrieved from %
\url{https://www.openicpsr.org/openicpsr/project/116216/version/V1/view}.}
We pool observations from 1973 to 1975 for the first sample period, $%
t=\left\{ 1973-1975\right\} $ and observations from 2001 to 2003 for the
second sample period, $t=\left\{ 2001-2003\right\} $. Following %
\citet{lemieux2006postnber}, we consider sub-samples of those with less than
12 years of education, \textquotedblleft high school or less", and those
with more than 12 years of education, \textquotedblleft postsecondary
education", as well as the combined sample. We also present results by
gender. The summary statistics are reported in Table \ref%
{tab:Descriptive_Statistics}. As to be expected, the mean log wages are
higher for those with postsecondary education (for male and female), with
the number of years of schooling and experience rising by about one year
across the two sub-period samples. There are also important differences
across male and female, and the two educational groupings, which we hope to
capture in our estimation.

\begin{table}[htbp]
\caption{Summary Statistics of the May and Outgoing Rotation Group (ORG)
supplements of the Current Population Survey (CPS) data across two periods,
1973 - 75 and 2001 - 03, by years of education and gender}
\label{tab:Descriptive_Statistics}\vspace{-10pt}
\par
\begin{center}
{\small 
\begin{tabular}{rccccccc}
\hline
& \multicolumn{3}{c}{1973 - 75} &  & \multicolumn{3}{c}{2001 - 03} \\ 
\cline{2-4}\cline{6-8}
& \multicolumn{1}{c}{High School} & Postsecondary & \multirow{2}{*}{All} & 
& High School & Postsecondary & \multirow{2}{*}{All} \\ 
& \multicolumn{1}{c}{or Less} & Education &  &  & or Less & Education &  \\ 
\hline
& \multicolumn{7}{c}{\textit{Both male and female}} \\ \hline
\texttt{log wage} & 1.59 & 1.94 & 1.69 &  & 1.47 & 1.88 & 1.71 \\ 
& {\footnotesize \textit{(0.50)}} & {\footnotesize \textit{(0.53)}} & 
{\footnotesize \textit{(0.53)}} &  & {\footnotesize \textit{(0.47)}} & 
{\footnotesize \textit{(0.57)}} & {\footnotesize \textit{(0.57)}} \\ 
\texttt{edu.} & 10.64 & 15.21 & 12.02 &  & 11.29 & 14.96 & 13.41 \\ 
& {\footnotesize \textit{(2.11)}} & {\footnotesize \textit{(1.65)}} & 
{\footnotesize \textit{(2.89)}} &  & {\footnotesize \textit{(1.68)}} & 
{\footnotesize \textit{(1.82)}} & {\footnotesize \textit{(2.53)}} \\ 
\texttt{age} & 36.74 & 34.90 & 36.18 &  & 37.96 & 39.87 & 39.06 \\ 
& {\footnotesize \textit{(13.85)}} & {\footnotesize \textit{(11.58)}} & (%
{\footnotesize \textit{{13.23})}} &  & {\footnotesize \textit{(12.93)}} & 
{\footnotesize \textit{(11.33)}} & {\footnotesize \textit{(12.07)}} \\ 
\texttt{expr.} & 20.10 & 13.69 & 18.17 &  & 20.67 & 18.91 & 19.65 \\ 
& {\footnotesize \textit{(14.44)}} & {\footnotesize \textit{(11.41)}} & 
{\footnotesize \textit{(13.91)}} &  & {\footnotesize \textit{(12.95)}} & 
{\footnotesize \textit{(11.17)}} & {\footnotesize \textit{(11.98)}} \\ 
\texttt{marriage} & 0.67 & 0.70 & 0.68 &  & 0.52 & 0.60 & 0.57 \\ 
& {\footnotesize \textit{(0.47)}} & {\footnotesize \textit{(0.46)}} & 
{\footnotesize \textit{(0.47)}} &  & {\footnotesize \textit{(0.50)}} & 
{\footnotesize \textit{(0.49)}} & {\footnotesize \textit{(0.50)}} \\ 
\texttt{nonwhite} & 0.11 & 0.08 & 0.10 &  & 0.15 & 0.14 & 0.15 \\ 
& {\footnotesize \textit{(0.32)}} & {\footnotesize \textit{(0.27)}} & 
{\footnotesize \textit{(0.30)}} &  & {\footnotesize \textit{(0.36)}} & 
{\footnotesize \textit{(0.35)}} & {\footnotesize \textit{(0.35)}} \\ 
$n$ & 77,899 & 33,733 & 111,632 &  & 216,136 & 295,683 & 511,819 \\ \hline
& \multicolumn{7}{c}{\textit{Male}} \\ \hline
\texttt{log wage} & 1.76 & 2.07 & 1.86 &  & 1.57 & 2.00 & 1.81 \\ 
& {\footnotesize \textit{(0.48)}} & {\footnotesize \textit{(0.53)}} & 
{\footnotesize \textit{(0.52)}} &  & {\footnotesize \textit{(0.48)}} & 
{\footnotesize \textit{(0.58)}} & {\footnotesize \textit{(0.58)}} \\ 
\texttt{edu.} & 10.44 & 15.29 & 12.00 &  & 11.19 & 15.02 & 13.31 \\ 
& {\footnotesize \textit{(2.26)}} & {\footnotesize \textit{(1.69)}} & 
{\footnotesize \textit{(3.08)}} &  & {\footnotesize \textit{(1.82)}} & 
{\footnotesize \textit{(1.84)}} & {\footnotesize \textit{(2.64)}} \\ 
\texttt{age} & 36.79 & 35.29 & 36.31 &  & 37.21 & 40.24 & 38.89 \\ 
& {\footnotesize \textit{(13.82)}} & {\footnotesize \textit{(11.24)}} & 
{\footnotesize \textit{(13.07)}} &  & {\footnotesize \textit{(12.70)}} & 
{\footnotesize \textit{(11.30)}} & {\footnotesize \textit{(12.04)}} \\ 
\texttt{expr.} & 20.35 & 14.00 & 18.32 &  & 20.02 & 19.22 & 19.58 \\ 
& {\footnotesize \textit{(14.49)}} & {\footnotesize \textit{(11.06)}} & 
{\footnotesize \textit{(13.81)}} &  & {\footnotesize \textit{(12.75)}} & 
{\footnotesize \textit{(11.08)}} & {\footnotesize \textit{(11.86)}} \\ 
\texttt{marriage} & 0.73 & 0.76 & 0.74 &  & 0.53 & 0.64 & 0.59 \\ 
& {\footnotesize \textit{(0.44)}} & {\footnotesize \textit{(0.43)}} & 
{\footnotesize \textit{(0.44)}} &  & {\footnotesize \textit{(0.50)}} & 
{\footnotesize \textit{(0.48)}} & {\footnotesize \textit{(0.49)}} \\ 
\texttt{nonwhite} & 0.10 & 0.06 & 0.09 &  & 0.14 & 0.13 & 0.13 \\ 
& {\footnotesize \textit{(0.30)}} & {\footnotesize \textit{(0.24)}} & 
{\footnotesize \textit{(0.29)}} &  & {\footnotesize \textit{(0.34)}} & 
{\footnotesize \textit{(0.33)}} & {\footnotesize \textit{(0.34)}} \\ 
$n$ & 44,299 & 20,851 & 65,150 &  & 116,129 & 144,138 & 260,267 \\ \hline
& \multicolumn{7}{c}{\textit{Female}} \\ \hline
\texttt{log wage} & 1.35 & 1.71 & 1.45 &  & 1.77 & 1.36 & 1.61 \\ 
& {\footnotesize \textit{(0.41)}} & {\footnotesize \textit{(0.47)}} & 
{\footnotesize \textit{(0.46)}} &  & {\footnotesize \textit{(0.54)}} & 
{\footnotesize \textit{(0.43)}} & {\footnotesize \textit{(0.54)}} \\ 
\texttt{edu.} & 10.89 & 15.08 & 12.05 &  & 14.90 & 11.42 & 13.52 \\ 
& {\footnotesize \textit{(1.87)}} & {\footnotesize \textit{(1.59)}} & 
{\footnotesize \textit{(2.60)}} &  & {\footnotesize \textit{(1.79)}} & 
{\footnotesize \textit{(1.49)}} & {\footnotesize \textit{(2.40)}} \\ 
\texttt{age} & 36.67 & 34.27 & 36.01 &  & 38.83 & 39.52 & 39.24 \\ 
& {\footnotesize \textit{(13.88)}} & {\footnotesize \textit{(12.09)}} & 
{\footnotesize \textit{(13.45)}} &  & {\footnotesize \textit{(13.14)}} & 
{\footnotesize \textit{(11.35)}} & {\footnotesize \textit{(12.10)}} \\ 
\texttt{expr.} & 19.78 & 13.19 & 17.96 &  & 18.61 & 21.41 & 19.73 \\ 
& {\footnotesize \textit{(14.36)}} & {\footnotesize \textit{(11.94)}} & 
{\footnotesize \textit{(14.04)}} &  & {\footnotesize \textit{(11.24)}} & 
{\footnotesize \textit{(13.13)}} & {\footnotesize \textit{(12.11)}} \\ 
\texttt{marriage} & 0.60 & 0.60 & 0.60 &  & 0.56 & 0.51 & 0.54 \\ 
& {\footnotesize \textit{(0.49)}} & {\footnotesize \textit{(0.49)}} & 
{\footnotesize \textit{(0.49)}} &  & {\footnotesize \textit{(0.50)}} & 
{\footnotesize \textit{(0.50)}} & {\footnotesize \textit{(0.50)}} \\ 
\texttt{nonwhite} & 0.13 & 0.10 & 0.12 &  & 0.15 & 0.17 & 0.16 \\ 
& {\footnotesize \textit{(0.33)}} & {\footnotesize \textit{(0.30)}} & 
{\footnotesize \textit{(0.33)}} &  & {\footnotesize \textit{(0.36)}} & 
{\footnotesize \textit{(0.38)}} & {\footnotesize \textit{(0.37)}} \\ 
$n$ & 33,600 & 12,882 & 46,482 &  & 151,545 & 100,007 & 251,552 \\ \hline
\end{tabular}%
}
\end{center}
\par
{\footnotesize \textit{Notes:} ``Postsecondary Education'' stands for the
sub-sample with years of education higher than 12 and ``High School or
Less'' stands for sub-sample with years of education less than or equal to
12). \texttt{edu.} and \texttt{exper.} are in years. \texttt{marriage} and 
\texttt{nonwhite} are dummy variables. $n$ is the sample size. We report
mean and standard deviation (in parentheses) of each variable. The data is
from the May and Outgoing Rotation Group (ORG) supplements of the Current
Population Survey (CPS) data retrived from %
\url{https://www.openicpsr.org/openicpsr/project/116216/version/V1/view}.}
\end{table}

\begin{table}[htbp]
\caption{Estimates of the distribution of the return to education across two
periods, 1973 - 75 and 2001 - 03, by years of education and gender}
\label{tab:Return-to-Education}
\begin{center}
\begin{tabular}{rrrrrrrrr}
\hline
& \multicolumn{2}{c}{High School or Less} &  & \multicolumn{2}{c}{
Postsecondary Edu.} &  & \multicolumn{2}{c}{All} \\ 
\cline{2-3}\cline{5-6}\cline{8-9}
& \multicolumn{1}{c}{1973 - 75} & \multicolumn{1}{c}{2001 - 03} &  & 
\multicolumn{1}{c}{1973 - 75} & \multicolumn{1}{c}{2001 - 03} &  & 
\multicolumn{1}{c}{1973 - 75} & \multicolumn{1}{c}{2001 - 03} \\ \hline
& \multicolumn{8}{c}{Both Male and Female} \\ \hline
$\pi$ & 0.4843 & 0.5069 &  & 0.4398 & 0.3537 &  & 0.4719 & 0.3463 \\ 
& {\small \textit{(4188.8)}} & {\small \textit{(0.0269)}} &  & {\small 
\textit{(0.0502)}} & {\small \textit{(0.0091)}} &  & {\small \textit{(0.0485)%
}} & \textit{(0.0047)} \\ 
$\beta_L$ & 0.0608 & 0.0382 &  & 0.0624 & 0.0866 &  & 0.0558 & 0.0645 \\ 
& {\small \textit{(5.0939)}} & {\small \textit{(0.0014)}} &  & {\small 
\textit{(0.0035)}} & {\small \textit{(0.0009)}} &  & {\small \textit{(0.0020)%
}} & {\small \textit{(0.0004)}} \\ 
$\beta_H$ & 0.0619 & 0.0920 &  & 0.1103 & 0.1401 &  & 0.0941 & 0.1263 \\ 
& {\small \textit{(4.8132)}} & {\small \textit{(0.0019)}} &  & {\small 
\textit{(0.0032)}} & {\small \textit{(0.0007)}} &  & {\small \textit{(0.0022)%
}} & {\small \textit{(0.0004)}} \\ 
$\beta_H / \beta_L$ & 1.0194 & 2.4102 &  & 1.7680 & 1.6178 &  & 1.6879 & 
1.9567 \\ 
& {\small \textit{(6.2938)}} & {\small \textit{(0.0428)}} &  & {\small 
\textit{(0.0618)}} & {\small \textit{(0.0111)}} &  & {\small \textit{(0.0295)%
}} & {\small \textit{(0.0080)}} \\ 
$\mathrm{E}\left(\beta_i\right)$ & 0.0614 & 0.0647 &  & 0.0893 & 0.1212 &  & 
0.0760 & 0.1049 \\ 
$\mathrm{s.d.}\left(\beta_i\right)$ & 0.0006 & 0.0269 &  & 0.0238 & 0.0256 & 
& 0.0191 & 0.0294 \\ 
$n$ & 77,899 & 216,136 &  & 33,733 & 295,683 &  & 111,632 & 511,819 \\ \hline
& \multicolumn{8}{c}{Male} \\ \hline
$\pi$ & n/a & 0.4939 &  & 0.4706 & 0.3201 &  & 0.4802 & 0.3290 \\ 
& n/a & {\small \textit{(0.0399)}} &  & {\small \textit{(0.0707)}} & {\small 
\textit{(0.0104)}} &  & {\small \textit{(0.0815)}} & {\small \textit{(0.0053)%
}} \\ 
$\beta_L$ & 0.0637 & 0.0404 &  & 0.0534 & 0.0743 &  & 0.0536 & 0.0548 \\ 
& n/a & {\small \textit{(0.0019)}} &  & {\small \textit{(0.0046)}} & {\small 
\textit{(0.0012)}} &  & {\small \textit{(0.0030)}} & {\small \textit{(0.0005)%
}} \\ 
$\beta_H$ & 0.0637 & 0.0911 &  & 0.0995 & 0.1308 &  & 0.0875 & 0.1192 \\ 
& n/a & {\small \textit{(0.0026)}} &  & {\small \textit{(0.0042)}} & {\small 
\textit{(0.0009)}} &  & {\small \textit{(0.0031)}} & {\small \textit{(0.0005)%
}} \\ 
$\beta_H / \beta_L$ & 1.0000 & 2.2526 &  & 1.8641 & 1.7603 &  & 1.6312 & 
2.1772 \\ 
& n/a & {\small \textit{(0.0534)}} &  & {\small \textit{(0.1038)}} & {\small 
\textit{(0.0209)}} &  & {\small \textit{(0.0459)}} & {\small \textit{(0.0144)%
}} \\ 
$\mathrm{E}\left(\beta_i\right)$ & 0.0637 & 0.0661 &  & 0.0778 & 0.1128 &  & 
0.0712 & 0.0980 \\ 
$\mathrm{s.d.}\left(\beta_i\right)$ & 0.0000 & 0.0253 &  & 0.0230 & 0.0264 & 
& 0.0169 & 0.0303 \\ 
$n$ & 44,299 & 116,129 &  & 20,851 & 144,138 &  & 65,150 & 260,267 \\ \hline
& \multicolumn{8}{c}{Female} \\ \hline
$\pi$ & 0.4999 & 0.5166 &  & 0.4526 & 0.3906 &  & 0.4566 & 0.3608 \\ 
& {\small \textit{(0.5047)}} & {\small \textit{(0.0283)}} &  & {\small 
\textit{(0.0829)}} & {\small \textit{(0.0167)}} &  & {\small \textit{(0.0810)%
}} & {\small \textit{(0.0086)}} \\ 
$\beta_L$ & 0.0441 & 0.0348 &  & 0.0823 & 0.0979 &  & 0.0628 & 0.0751 \\ 
& {\small \textit{(0.0133)}} & {\small \textit{(0.0016)}} &  & {\small 
\textit{(0.0053)}} & {\small \textit{(0.0013)}} &  & {\small \textit{(0.0033)%
}} & {\small \textit{(0.0007)}} \\ 
$\beta_H$ & 0.0723 & 0.0972 &  & 0.1310 & 0.1473 &  & 0.1028 & 0.1333 \\ 
& {\small \textit{(0.0159)}} & {\small \textit{(0.0025)}} &  & {\small 
\textit{(0.0055)}} & {\small \textit{(0.0011)}} &  & {\small \textit{(0.0038)%
}} & {\small \textit{(0.0007)}} \\ 
$\beta_H / \beta_L$ & 1.6392 & 2.7934 &  & 1.5913 & 1.5048 &  & 1.6357 & 
1.7756 \\ 
& {\small \textit{(0.1565)}} & {\small \textit{(0.0700)}} &  & {\small 
\textit{(0.0539)}} & {\small \textit{(0.0121)}} &  & {\small \textit{(0.0353)%
}} & \textit{(0.0090)} \\ 
$\mathrm{E}\left(\beta_i\right)$ & 0.0582 & 0.0650 &  & 0.1090 & 0.1280 &  & 
0.0845 & 0.1123 \\ 
$\mathrm{s.d.}\left(\beta_i\right)$ & 0.0141 & 0.0312 &  & 0.0242 & 0.0241 & 
& 0.0199 & 0.0280 \\ 
$n$ & 33,600 & 100,007 &  & 12,882 & 151,545 &  & 46,482 & 251,552 \\ \hline
\end{tabular}%
\end{center}
\par
{\footnotesize \textit{Notes:} This table reports the estimates of the
distribution of $\beta _{i}$ with the quadratic in experience specification %
\eqref{eq:ccrm_spec_alpha_i}, using $S = 4$ order moments of $\text{edu}_i$.
\textquotedblleft Postsecondary Edu.\textquotedblright\ stands for the
sub-sample with years of education higher than 12 and \textquotedblleft High
School or Less\textquotedblright\ stands for those with years of education
less than or equal to 12. $\mathrm{s.d.}\left(\beta _{i}\right)$ corresponds
to the square root of estimated $\mathrm{var}\left( \beta _{i}\right) $. $n$
is the sample size. \textquotedblleft n/a" is inserted when the estimates
show homogeneity of $\beta _{i}$ and $\pi $ is not identified and cannot be
estimated.}
\end{table}

\begin{table}[p!]
\caption{Estimates of $\mathbf{\protect\gamma}$ associated with control
variables $\mathbf{z}_i$ with specification \eqref{eq:ccrm_spec_alpha_i}
across two periods, 1973 - 75 and 2001 - 03, by years of education and
gender, which complements Table \protect\ref{tab:Return-to-Education}}
\label{tab:gamma_est_quadratic}
\begin{center}
\begin{tabular}{rrrrrrlrr}
\hline
& \multicolumn{2}{c}{High School or Less} & \multicolumn{1}{l}{} & 
\multicolumn{2}{c}{Postsecondary Edu.} &  & \multicolumn{2}{c}{All} \\ 
\cline{2-3}\cline{5-6}\cline{8-9}
& \multicolumn{1}{c}{1973 - 75} & \multicolumn{1}{c}{2001 - 03} & 
\multicolumn{1}{l}{} & \multicolumn{1}{c}{1973 - 75} & \multicolumn{1}{c}{
2001 - 03} &  & \multicolumn{1}{c}{1973 - 75} & \multicolumn{1}{c}{2001 - 03}
\\ \hline
& \multicolumn{8}{c}{\textit{Both male and female}} \\ \hline
\texttt{exper.} & 0.0305 & 0.0319 &  & 0.0415 & 0.0354 &  & 0.0310 & 0.0321
\\ 
& {\small \textit{(0.0004)}} & {\small \textit{(0.0002)}} &  & {\small 
\textit{(0.0008)}} & {\small \textit{(0.0003)}} &  & {\small \textit{%
(0.0003) }} & {\small \textit{(0.0002)}} \\ 
$\mathtt{exper.}^2$ ($\times 10^2$) & -0.0490 & -0.0505 &  & -0.0826 & 
-0.0652 &  & -0.0499 & -0.0537 \\ 
& {\small \textit{(0.0009)}} & {\small \textit{(0.0005)}} &  & {\small 
\textit{(0.0022)}} & {\small \textit{(0.0007)}} &  & {\small \textit{%
(0.0008) }} & {\small \textit{(0.0005)}} \\ 
\texttt{marriage} & 0.1120 & 0.0751 &  & 0.0886 & 0.0770 &  & 0.1085 & 0.0818
\\ 
& {\small \textit{(0.0036)}} & {\small \textit{(0.0020)}} &  & {\small 
\textit{(0.0059)}} & {\small \textit{(0.0020)}} &  & {\small \textit{%
(0.0031) }} & {\small \textit{(0.0014)}} \\ 
\texttt{nonwhite} & -0.0922 & -0.0775 &  & -0.0424 & -0.0571 &  & -0.0715 & 
-0.0667 \\ 
& {\small \textit{(0.0047)}} & {\small \textit{(0.0024)}} &  & {\small 
\textit{(0.0088)}} & {\small \textit{(0.0025)}} &  & {\small \textit{%
(0.0042) }} & {\small \textit{(0.0018)}} \\ 
\texttt{gender} & 0.4157 & 0.2298 &  & 0.2962 & 0.2023 &  & 0.3892 & 0.2167
\\ 
& {\small \textit{(0.0029)}} & {\small \textit{(0.0017)}} &  & {\small 
\textit{(0.0050)}} & {\small \textit{(0.0018)}} &  & {\small \textit{%
(0.0025) }} & {\small \textit{(0.0013)}} \\ 
$n$ & 77,899 & 216,136 &  & 33,733 & 295,683 &  & 111,632 & 511,819 \\ \hline
& \multicolumn{8}{c}{\textit{Male}} \\ \hline
\texttt{exper.} & 0.0369 & 0.0366 &  & 0.0516 & 0.0405 &  & 0.0389 & 0.0371
\\ 
& {\small \textit{(0.0005)}} & {\small \textit{(0.0003)}} &  & {\small 
\textit{(0.0011)}} & {\small \textit{(0.0005)}} &  & {\small \textit{%
(0.0005) }} & {\small \textit{(0.0003)}} \\ 
$\mathtt{exper.}^2$ ($\times 10^2$) & -0.0589 & -0.0589 &  & -0.1016 & 
-0.0752 &  & -0.0635 & -0.0629 \\ 
& {\small \textit{(0.0012)}} & {\small \textit{(0.0008)}} &  & {\small 
\textit{(0.0029)}} & {\small \textit{(0.0011)}} &  & {\small \textit{%
(0.0010) }} & {\small \textit{(0.0007)}} \\ 
\texttt{marriage} & 0.1940 & 0.1123 &  & 0.1497 & 0.1344 &  & 0.1828 & 0.1316
\\ 
& {\small \textit{(0.0053)}} & {\small \textit{(0.0028)}} &  & {\small 
\textit{(0.0085)}} & {\small \textit{(0.0031)}} &  & {\small \textit{%
(0.0045) }} & {\small \textit{(0.0021)}} \\ 
\texttt{nonwhite} & -0.1241 & -0.1165 &  & -0.1172 & -0.1010 &  & -0.1178 & 
-0.1093 \\ 
& {\small \textit{(0.0065)}} & {\small \textit{(0.0035)}} &  & {\small 
\textit{(0.0127)}} & {\small \textit{(0.0039)}} &  & {\small \textit{%
(0.0058) }} & {\small \textit{(0.0027)}} \\ 
$n$ & 44,299 & 116,129 &  & 20,851 & 144,138 &  & 65,150 & 260,267 \\ \hline
& \multicolumn{8}{c}{\textit{Female}} \\ \hline
\texttt{exper.} & 0.0223 & 0.0265 &  & 0.0271 & 0.0313 &  & 0.0208 & 0.0272
\\ 
& {\small \textit{(0.0006)}} & {\small \textit{(0.0003)}} &  & {\small 
\textit{(0.0011)}} & {\small \textit{(0.0004)}} &  & {\small \textit{%
(0.0005) }} & {\small \textit{(0.0003)}} \\ 
$\mathtt{exper.}^2$ ($\times 10^2$) & -0.0376 & -0.0411 &  & -0.0564 & 
-0.0576 &  & -0.0338 & -0.0450 \\ 
& {\small \textit{(0.0013)}} & {\small \textit{(0.0008)}} &  & {\small 
\textit{(0.0030)}} & {\small \textit{(0.0010)}} &  & {\small \textit{%
(0.0012) }} & {\small \textit{(0.0006)}} \\ 
\texttt{marriage} & 0.0115 & 0.0317 &  & -0.0005 & 0.0262 &  & 0.0118 & 
0.0322 \\ 
& {\small \textit{(0.0048)}} & {\small \textit{(0.0028)}} &  & {\small 
\textit{(0.0079)}} & {\small \textit{(0.0026)}} &  & {\small \textit{%
(0.0041) }} & {\small \textit{(0.0019)}} \\ 
\texttt{nonwhite} & -0.0581 & -0.0441 &  & 0.0395 & -0.0236 &  & -0.0202 & 
-0.0315 \\ 
& {\small \textit{(0.0065)}} & {\small \textit{(0.0033)}} &  & {\small 
\textit{(0.0117)}} & {\small \textit{(0.0033)}} &  & {\small \textit{%
(0.0058) }} & {\small \textit{(0.0024)}} \\ 
$n$ & 33,600 & 100,007 &  & 12,882 & 151,545 &  & 46,482 & 251,552 \\ \hline
\end{tabular}%
\end{center}
\par
{\footnotesize \textit{Notes:} This table reports the estimates of $\mathbf{%
\ \gamma}$ in \eqref{eq:ccrm_spec_alpha_i}. ``Postsecondary Edu.'' stands
for the sub-sample with years of education higher than 12 and ``High School
or Less'' stands for those with years of education less than or equal to 12.
The standard error of estimates of coefficients associated with control
variables are estimated based on Theorem \ref{lem:gamma_est_consistency} and
reported in parentheses. $n$ is the sample size.}
\end{table}

We treat the cross-section observations in the two sample periods, $%
t=\left\{ 1973-1975\right\} $ and $\left\{ 2001-2003\right\} $, as \textit{%
repeated} cross-sections, rather than a panel data since the data in these
two periods do not cover the same individuals, and represent random samples
from the population of wage earners in two periods. It should also be noted
that sample sizes $(n_{t})$, although quite large, are much larger during $%
\left\{ 2001-2003\right\} $, which could be a factor when we come to compare
estimates from the two sample periods. For example, for both male and female 
$n_{73-75}=111,632$ as compared to $n_{01-03}=511,819$, a difference which
becomes more pronounced when we consider the number observations in
postsecondary/female category - which rises from $12,882$ for the first
period to $100,007$ in the second period.

We report estimates of $\pi_{t}$, $\beta _{L,t}$ and $\beta _{H,t}$, as well
as corresponding mean and standard deviations (denoted by s.d.($\hat{\beta}%
_{it}$)) of the return to education ($\beta_{it}$) for $t=\left\{
1973-1975\right\} $ and $\left\{ 2001-2003\right\} $. For a given $\pi _{t}$%
, the ratio $\beta _{H,t}/\beta _{L,t}$ provides a measure of within group
heterogeneity and allows us to augment information on changes in mean with
changes in the distribution of return of education. The estimates for the
distribution of the return to education ($\beta _{it}$) are summarized in
Table \ref{tab:Return-to-Education}, with the estimation results for control
variables (such as experience, experienced squared, and other individual
specific characteristic) reported in Table \ref{tab:gamma_est_quadratic}.

As can be seen from Table \ref{tab:Return-to-Education}, estimates of $%
\mathrm{s.d.}\left( \beta _{it}\right) $ are strictly positive for all
sub-groups, except for the \textquotedblleft high school or
less\textquotedblright\ group during the first sample period. For this group
during the first period the estimate of $\mathrm{s.d.}\left( \beta
_{it}\right) $ for the male sub-sample is zero, $\pi $ is not identified,
and we have identical estimates for $\beta _{L}$ and $\beta _{H}$. For this
sub-sample, the associated estimates and their standard errors are shown as
unavailable ($n/a$). In case of the female sub-sample as well as both male
and female sub-samples where the estimates of s.d.($\hat{\beta}_{it}$) are
close to zero and $\pi $ is poorly estimated, only the mean of the return to
education is informative. In the case of the samples where the estimates of $%
\mathrm{s.d.}\left( \beta _{it}\right) $ are strictly positive, the estimate
of the ratio $\beta _{H,t}/\beta _{L,t}$ provides a good measure of within
group heterogeneity of return to education. The estimates of $\beta
_{H,t}/\beta _{L,t}$, lie between $1.50$ to $2.79$, with the high estimate
obtained for the females with high school or less education during $\left\{ {%
2001-03}\right\} $, and the low estimate is obtained for females with
postsecondary education during the same period.

As our theory suggests the mean estimates of return to education, $E\left(
\beta _{it}\right) $, are very precisely estimated and inferences involving
them tend to be robust to conditional error heteroskedasticity. The results
in Table \ref{tab:Return-to-Education} show that estimates of $E\left( \beta
_{it}\right) $ have increased over the two sample periods $t=\left\{
1973-75\right\} $ to $t=\left\{ 2001-03\right\} $, regardless of gender or
educational grouping. The postsecondary educational group show larger
increases in the estimates of $E\left( \beta _{it}\right) $ as compared to
those with high school or less. Estimates of $\mathrm{E}\left( \beta
_{it}\right) $ increases by $36$ per cent for the postsecondary group while
the estimates of mean return to education rises only by around $5$ per cent
in the case of those with high school or less. This result holds for both
genders. Comparing the mean returns across the two educational groups, we
find that mean return to education of individuals with postsecondary
education is $45$ per cent higher than those with high school or less in the 
$\{1973 - 75\}$ period, but this gap increases to $87$ per cent in the
second period, $\left\{ 2001-03\right\} $. Similar patterns are observed in
the sub-samples by gender. The estimates suggest rising between group
heterogeneity, which is mainly due to the increasing returns to education
for the postsecondary group.

Turning to within group heterogeneity, we focus on the estimates of $\beta
_{H,t}/\beta _{L,t}$ and first note that over the two periods, within group
heterogeneity has been rising mainly in the case of those with high school
or less, for both male and female. For the combined male and female samples
and the male sub-sample, there is little evidence of within group
heterogeneity for the first period $\left\{ 1973-75\right\} $. However, for
the second\ period $\left\{ 2001-03\right\} $ we find a sizeable degree of
within group heterogeneity where $\beta _{H,t}/\beta _{L,t}$ is estimated to
be around $2.41$, with $\text{s.d.}\left( \beta _{it}\right) \approx 0.03$.
For the female sub-sample with high school or less, little evidence of
heterogeneity was found for the first period, estimates of $\beta
_{H,t}/\beta _{L,t}$ increases to $2.79$ for the second sample period, that
corresponds to a commensurate rise in $\text{s.d.}\left( \beta _{i}\right) $
to $0.032$. The pattern of within group heterogeneity is very different for
those with postsecondary educational. For this group we in fact observe a
slight decline in the estimates of $\beta _{H,t}/\beta _{L,t}$ by gender and
over two sample periods.

Overall, our between and within estimates of mean return to education are in
line with the evidence of rising wage inequality documented in the
literature \citep{corak2013income}.


\section{Conclusion\label{sec:Conclusion}}

In this paper we consider random coefficient models for repeated
cross-sections in which the random coefficients follow categorical
distributions. Identification is established using moments of the random
coefficients in terms of the moments of the underlying observations. We
propose two-step generalized method of moments to estimate the parameters of
the categorical distributions. The consistency and asymptotic normality of
the GMM estimators are established without the IID assumption typically
assumed in the literature. Small sample properties of the proposed estimator
are investigated by means of Monte Carlo experiments and shown to be robust
to heterogeneously generated regressors and errors, although relatively
large samples are required to estimate the parameters of the underling
categorical distributions. This is largely due to the highly non-linear
mapping between the parameters of the categorical distribution and the
higher order moments of the coefficients. This problem is likely to become
more pronounced with a larger number of categories and coefficients.

In the empirical application, we apply the model to study the evolution of
returns to education over two sub-periods, also considered in the literature
by \citet{lemieux2006postnber}. Our estimates show that mean (ex post)
returns to education have risen over the periods from 1973 - 75 to 2001 -
2003 mainly in the case of individuals with postsecondary education, and
this result is robust by gender. We find evidence of within group
heterogeneity in the case of high school or less educational group as
compared to those with postsecondary education.

In our model specification, the number of categories, $K$, is treated as a
tuning parameter and assumed to be known. An information criterion, as in %
\citet{bonhomme_manersa2015fixedeffect} and \citet*{ssp2016classo}, to
determine $K$ could be considered. Further investigation of models with
multiple regressors subject to parameter heterogeneity is also required.
These and other related issues are topics for future research.

\newpage \appendix

\begin{center}
{\LARGE Appendix}
\end{center}

\setcounter{table}{0} \renewcommand{\thetable}{A.\arabic{table}} %
\setcounter{section}{0} \renewcommand{\thesection}{A.\arabic{section}} %
\setcounter{figure}{0} \renewcommand{\thefigure}{A.\arabic{figure}}

\section{Proofs\label{sec:Proofs}}

We include proofs and technical details in this section.

\begin{proof}[Proof of Theorem \protect\ref{lem: identification_moments}]
Sum \eqref{eq:mc_n_r} over $i$ and rearrange terms,

\begin{equation}
\left( \frac{1}{n}\sum_{i=1}^{n}\mathrm{E}\left( x_{i}^{r}\right) \right) 
\mathrm{E}\left( \beta _{i}^{r}\right) +\frac{1}{n}\sum_{i=1}^{n}\mathrm{E}%
\left( u_{i}^{r}\right) =\frac{1}{n}\sum_{i=1}^{n}\mathrm{E}\left( \tilde{y}%
_{i}^{r}\right) -\sum_{q=2}^{r-1}\binom{r}{q}\left( \frac{1}{n}\sum_{i=1}^{n}%
\mathrm{E}\left( x_{i}^{r-q}\right) \mathrm{E}\left( u_{i}^{q}\right)
\right) \mathrm{E}\left( \beta _{i}^{r-q}\right) .  \label{eq:mc_n_sum_r}
\end{equation}

Note that 
\begin{equation*}
\frac{1}{n}\sum_{i=1}^n \mathrm{E}\left(x_i^{r-q}\right) \mathrm{E} \left(
u_{i}^{q}\right) = \left( \frac{1}{n}\sum_{i=1}^n \mathrm{E}%
\left(x_i^{r-q}\right)\right) \sigma_q + \frac{1}{n}\sum_{i=1}^n \mathrm{E}%
\left(x_i^{r-q}\right) \left( \mathrm{E} \left( u_{i}^{q}\right) - \sigma_q
\right),
\end{equation*}
and 
\begin{equation*}
\left\vert \frac{1}{n}\sum_{i=1}^n \mathrm{E}\left(x_i^{r-q}\right) \left( 
\mathrm{E} \left( u_{i}^{q}\right) - \sigma_q \right) \right\vert \leq
\sup_i \left\vert \mathrm{E}\left(x_i^{r-q}\right) \right\vert \left\vert 
\frac{1}{n}\sum_{i=1}^n \left( \mathrm{E} \left( u_{i}^{q}\right) - \sigma_q
\right) \right\vert = O(n^{-1/2}),
\end{equation*}
by Assumption \ref{assu:identification_regularity_condition}(b) and \ref%
{assu:conv_moment}(b), then by taking $n\to \infty$ on both sides of %
\eqref{eq:mc_n_sum_r}, we have \eqref{eq:mc_limit_r}. Similar steps for %
\eqref{eq:mc_n_2r} give \eqref{eq:mc_limit_2r}.
\end{proof}

\medskip

\begin{proof}[Proof of Theorem \protect\ref{prop:identification_of_beta_L_H}]

Let $m_{r}=\mathrm{E}\left( \beta _{i}^{r}\right) $, $r=1,2,\cdots ,2K-1$,
which are taken as known. We show that 
\begin{equation}
m_{r}=\sum_{k=1}^{K}\pi _{k}b_{k}^{r},  \label{eq:id_general}
\end{equation}%
$r=0,1,2,\cdots ,2K-1$, has a unique solution $\mathbf{\theta }=\left( 
\mathbf{\pi }^{\prime },\mathbf{b}^{\prime }\right) ^{\prime }$, with $%
b_{1}<b_{2}<\cdots <b_{K}$ and $\pi _{k}\in \left( 0,1\right) $ imposed.

Let 
\begin{equation}
q\left( \lambda \right) =\prod_{k=1}^{K}\left( \lambda -b_{k}\right)
=\lambda ^{K}+\left( -1\right) ^{1}b^{\ast}_{1}\lambda ^{K-1}+\cdots +\left(
-1\right) ^{K}b^{\ast}_{K} ,  \label{eq:characteristic_polynomial}
\end{equation}%
be the polynomial with $K$ distinct roots $b_{1}$, $b_{2}$, $\cdots $, $%
b_{K} $. Note that for each $k$, $\left( b_{k}^{r}\right) _{r=0}^{2K-1}$
satisfies the linear homogeneous recurrence relation, 
\begin{equation}
b_{k}^{K+r}=b^{\ast}_{1}b_{k}^{K+r-1}+\left( -1\right)
^{1}b^{\ast}_{2}b_{k}^{K+r-2}+\cdots +\left( -1\right)
^{K-1}b^{\ast}_{K}b_{k}^{r},  \label{eq:linear_recurrence_relation}
\end{equation}%
for $r=0,1,\cdots K-1$, since $q$ is the characteristic polynomial of the
linear recurrence relation (\ref{eq:linear_recurrence_relation}) and $b_{k}$
is a root of $q$ \citep[Chapter 5.2]{rosen2011discretemath}. $\left(
m_{r}\right) _{r=0}^{2K-1}$ is a linear combination of $\left(
b_{1}^{r}\right) _{r=0}^{2K-1}$, $\left( b_{2}^{r}\right) _{r=0}^{2K-1}$, $%
\cdots $, $\left( b_{K}^{r}\right) _{r=0}^{2K-1}$ by (\ref{eq:id_general}),
then $\left( m_{r}\right) _{r=0}^{2K-1}$ also satisfies the linear
recurrence relation \eqref{eq:linear_recurrence_relation}, i.e., 
\begin{equation}
m_{K+r}=b^{\ast}_{1}m_{K+r-1}+\left( -1\right)
^{1}b^{\ast}_{2}m_{K+r-2}+\cdots +\left( -1\right) ^{K-1}b^{\ast}_{K}m_{r},
\label{eq:linear_recurrence_m}
\end{equation}%
for $r=0,1,\cdots ,K-1$. \eqref{eq:linear_recurrence_m} is a linear system
of $K $ equations in terms of $\left( b^{\ast}_{k}\right) _{k=1}^{K}$. In
matrix form, 
\begin{equation}
\mathbf{M} \mathbf{D} \mathbf{b}^{\ast} = \mathbf{m},
\label{eq: linear_recurrence_matrix}
\end{equation}
where 
\begin{equation*}
\mathbf{M} = 
\begin{pmatrix}
1 & m_1 & \cdots & m_{K-1} \\ 
m_{1} & m_{2} & \cdots & m_{K} \\ 
\vdots & \vdots & \ddots & \vdots \\ 
m_{K-1} & m_{K} & \cdots & m_{2K-2}%
\end{pmatrix}%
,
\end{equation*}
$\mathbf{D} = \mathrm{diag}\left( \left( -1\right) ^{K-1},\left( -1\right)
^{K-2},\cdots ,1\right)$, $\mathbf{b}^\ast = \left( b^{\ast}_K,
b^{\ast}_{K-1}, \cdots, b^{\ast}_1 \right)^{\prime}$, and $\mathbf{m} =
\left( m_{K}, m_{K+1}, \cdots, m_{2K-1} \right)^{\prime}$.

Denote $\mathbf{\psi }_{k}=\left( 1,b_{k},b_{k}^{2}\cdots
,b_{k}^{K-1}\right) ^{\prime }$ and $\mathbf{\Psi }=\left( \mathbf{\psi }%
_{1},\mathbf{\psi }_{2},\cdots ,\mathbf{\psi }_{K}\right) $. Then 
\begin{equation*}
\mathbf{M}_{k}=%
\begin{pmatrix}
1 & b_{k} & \cdots & b_{k}^{K-1} \\ 
b_{k} & b_{k}^{2} & \cdots & b_{k}^{K} \\ 
\vdots & \vdots & \ddots & \vdots \\ 
b_{k}^{K-1} & b_{k}^{K} & \cdots & b_{k}^{2K-2}%
\end{pmatrix}%
=\mathbf{\psi }_{k}\mathbf{\psi }_{k}^{\prime } ,
\end{equation*}%
and $\mathbf{M}=\sum_{k=1}^{K}\pi _{k}\mathbf{M}_{k}=\mathbf{\Psi }\mathrm{%
diag}\left( \mathbf{\pi }\right) \mathbf{\Psi } ^{\prime }$. Note that $%
\mathbf{\Psi }^{\prime }$ is a Vandermonde matrix then $\det \left( \Psi
\right) =\prod_{1\leq k<k^{\prime }\leq K}\left( b_{k^{\prime
}}-b_{k}\right) >0$ since $b_{1}<b_{2}<\cdots <b_{K}$. 
\begin{align*}
\det \left( \mathbf{MD}\right) & =\det \left( \mathbf{\Psi }\mathrm{diag}%
\left( \mathbf{\pi }\right) \mathbf{\Psi }^{\prime }\right) \det \left( 
\mathbf{D}\right) \\
& =\left( \prod_{1\leq k<k^{\prime }\leq K}\left( b_{k^{\prime
}}-b_{k}\right) \right) ^{2}\left( \prod_{k=1}^{K}\pi _{k}\right) \left(
\left( -1\right) ^{\frac{1}{2}K\left( K-1\right) }\right) \neq 0,
\end{align*}%
since $\pi _{k}\in \left( 0,1\right) $ for any $k$. Then we can identify $%
\left(b^{\ast}_{k}\right) _{k=1}^{K}$ by $\left( m_{r}\right) _{r=0}^{2K-1}$
in \eqref{eq: linear_recurrence_matrix}, and hence the characteristic
polynomial is determined, and we can identify $\left( b_{k}\right)
_{k=1}^{K} $ by \eqref{eq:characteristic_polynomial}.

Since both $\left(b_{k}\right)_{k=1}^{K}$ and $\left(m_{r}%
\right)_{r=1}^{2K-1}$ are identified, the first $K$ equations of (\ref%
{eq:id_general}) is 
\begin{equation*}
\mathbf{\Psi}^{\prime}\mathbf{\pi}=\left(1,m_{1},m_{2},\cdots,m_{K-1}%
\right)^{\prime},
\end{equation*}
and $\mathbf{\pi}$ is identified by inverting the Vandermonde matrix $%
\mathbf{\Psi}^{\prime}$, which completes the proof.
\end{proof}

\medskip

\begin{proof}[Proof of Theorem \protect\ref{thm:consisteny}]
Denote 
\begin{equation*}
\Phi_0\left( \mathbf{\theta }, \mathbf{\sigma}, \mathbf{\gamma } \right) = 
\mathbf{g}_0\left( \mathbf{\theta }, \mathbf{\sigma}, \mathbf{\gamma }
\right)^\prime \mathbf{A} \mathbf{g}_0\left( \mathbf{\theta }, \mathbf{\sigma%
}, \mathbf{\gamma } \right),
\end{equation*}
where we stack the left-hand side of \eqref{eq:mc_limit} and transform $%
\mathbf{m}_\beta = h\left( \mathbf{\theta} \right)$ to get $\mathbf{g}%
_0\left(\mathbf{\theta }, \mathbf{\sigma}, \mathbf{\gamma } \right)$. We
suppress and the argument $\hat{\gamma}$ and denote $\mathbf{\eta} = \left( 
\mathbf{\theta}^\prime, \mathbf{\sigma}^\prime \right)^\prime$ for notation
simplicity and proceed by verifying the conditions of 
\citet[Theorem
2.1]{newey1994large}. Theorem \ref{prop:identification_of_beta_L_H} provides
the identification results which together with the positive definiteness of $%
\mathbf{A}$ verifies that $\Phi _{0}\left( \mathbf{\eta }, \mathbf{\gamma }%
\right) $ is uniquely minimized to 0 at $\mathbf{\eta }_{0}$. The
compactness of the parameter space holds by Assumption \ref{assu:Consistency
Assumption}(a). Note that $\mathbf{g}_0\left( \mathbf{\eta }, \mathbf{\gamma 
}\right) $ is a polynomial in $\mathbf{\eta }$, which is continuous in $%
\mathbf{\eta }$. $\mathbf{g}_0\left( \mathbf{\eta }, \mathbf{\gamma }\right) 
$ is bounded on $\Theta \times \mathcal{S}$. We proceed by verify the
uniform convergence condition. The additive terms in $\hat{\mathbf{g}}%
_n\left( \mathbf{\eta }, \hat{\mathbf{\gamma }}\right) - \mathbf{g}_0\left( 
\mathbf{\eta },\mathbf{\gamma }\right) $ are of the form $H_{n,1} h^{\left(
r, q\right) }\left( \mathbf{\eta}\right)$ or $H_{n,2}$, where 
\begin{align*}
\left\vert H_{n, 1} \right\vert & = \left\vert \frac{1}{n} \sum_{i=1}^{n}
x_{i}^{r-q+s_{r}} - \rho_{0, r-q+s_r} \right\vert \\
& \leq \left\vert \frac{1}{n} \sum_{i=1}^{n} x_{i}^{r-q+s_{r}} - \frac{1}{n}
\sum_{i=1}^{n} \mathrm{E}\left(x_{i}^{r-q+s_{r}}\right) \right\vert +
\left\vert \frac{1}{n} \sum_{i=1}^{n} \mathrm{E}\left(x_{i}^{r-q+s_{r}}%
\right) - \rho_{0, r-q+s_r} \right\vert \\
& = O_{p}\left(n^{-1/2}\right),
\end{align*}
$h^{\left( r,q\}\right) }\left( \mathbf{\eta }\right) $ is a polynomial in $%
\mathbf{\eta }$, and 
\begin{align*}
\left\vert H_{n,2} \right\vert & = \left\vert \frac{1}{n}\sum_{i=1}^{n} \hat{%
\tilde{y}} _{i}^{r}x_{i}^{s_{r}} - \rho_{r,s_r} \right\vert \\
& \leq \left\vert \frac{1}{n}\sum_{i=1}^{n} \hat{\tilde{y}}
_{i}^{r}x_{i}^{s_{r}} - \frac{1}{n}\sum_{i=1}^{n}\mathrm{E}\left( \tilde{y}%
_{i}^{r}x_{i}^{s_{r}}\right) \right\vert + \left\vert \frac{1}{n}%
\sum_{i=1}^{n} \mathrm{E}\left( \tilde{y}_{i}^{r}x_{i}^{s_{r}} \right) -
\rho_{r,s_r} \right\vert \\
& = O_{p}\left( n^{-1/2}\right).
\end{align*}
$H_{n,1} = O_{p}\left(n^{-1/2}\right)$ and $H_{n,2} =
O_{p}\left(n^{-1/2}\right)$ are due to Assumption \ref{assu:conv_moment}(a)
and \ref{assu:Consistency Assumption}(c).

By the compactness of $\Theta \times \mathcal{S} $, $\sup_{\mathbf{\eta }\in
\Theta\times \mathcal{S} }h^{\left( r,q\right) }\left( \mathbf{\eta }\right)
<C<\infty $ for some positive constant $C$. By triangle inequality, we have 
\begin{equation}
\sup_{\mathbf{\eta }\in \Theta \times \mathcal{S} }\left\Vert \hat{\mathbf{g}%
}_n \left( \mathbf{\eta },\hat{\mathbf{\gamma }}\right) - \mathbf{g}_0\left( 
\mathbf{\eta },\mathbf{\gamma }\right) \right\Vert \rightarrow _{p}0 ,
\label{eq:uni-conv-g-func}
\end{equation}
as $n\rightarrow \infty $. Following the proof of 
\citet[Theorem
    2.1]{newey1994large}, 
\begin{align*}
& \left\vert \hat{\Phi}_{n}\left( \mathbf{\eta },\hat{\mathbf{\ \gamma }}%
\right) -\Phi_0\left( \mathbf{\eta },\mathbf{\gamma } \right) \right\vert \\
\leq & \left\vert \left[ \hat{\mathbf{g}}_n\left( \mathbf{\eta },\hat{ 
\mathbf{\gamma }}\right) -\mathbf{g}_0 \left( \mathbf{\eta },\mathbf{\
\gamma }\right) \right] ^{\prime }\mathbf{A}_{n}\left[ \hat{\mathbf{g}}%
_n\left( \mathbf{\eta },\hat{\mathbf{\gamma }}\right) -\mathbf{g}_0 \left( 
\mathbf{\eta },\mathbf{\gamma }\right) \right] \right\vert +\left\vert 
\mathbf{g}_0 \left( \mathbf{\eta },\mathbf{\gamma }\right) ^{\prime }\left( 
\mathbf{A}_{n}+\mathbf{A}_{n}^{\prime }\right) \left[ \hat{\mathbf{g}}%
_{n}\left( \mathbf{\eta },\hat{\mathbf{\gamma }} \right) -\mathbf{g}_0
\left( \mathbf{\eta },\mathbf{\gamma }\right) \right] \right\vert \\
& \quad \quad +\left\vert \mathbf{g}_0 \left( \mathbf{\eta },\mathbf{\
\gamma }\right) ^{\prime }\left( \mathbf{A}_{n}-\mathbf{A}\right) \mathbf{g}%
_0 \left( \mathbf{\eta },\mathbf{\gamma }\right) \right\vert \\
\leq & \left\Vert \hat{\mathbf{g}}_n\left( \mathbf{\eta },\hat{\mathbf{\
\gamma }}\right) -\mathbf{g}_0 \left( \mathbf{\eta },\mathbf{\gamma }
\right) \right\Vert ^{2}\left\Vert \mathbf{A}_{n}\right\Vert +2\left\Vert 
\mathbf{g}_0 \left( \mathbf{\eta },\mathbf{\gamma }\right) \right\Vert
\left\Vert \hat{\mathbf{g}}_n\left( \mathbf{\eta },\hat{ \mathbf{\gamma }}%
\right) -\mathbf{g}_0 \left( \mathbf{\eta },\mathbf{\ \gamma }\right)
\right\Vert \left\Vert \mathbf{A}_{n}\right\Vert +\left\Vert \mathbf{g}_0
\left( \mathbf{\eta },\mathbf{\gamma }\right) \right\Vert ^{2}\left\Vert 
\mathbf{A}_{n}-\mathbf{A}\right\Vert.
\end{align*}
By \eqref{eq:uni-conv-g-func} and the boundedness of $\mathbf{g}_0 $, $%
\sup_{ \mathbf{\eta }\in \eta }\left\vert \hat{\Phi}_{n}\left( \mathbf{\
\eta },\hat{\mathbf{\gamma }}\right) -\Phi _{n}\left( \mathbf{\ \eta },%
\mathbf{\gamma }\right) \right\vert \rightarrow _{p}0$, which completes the
proof.
\end{proof}

\medskip

\begin{proof}[Proof of Theorem \protect\ref{thm:normality}]
We denote $\mathbf{\eta }=\left( \mathbf{\theta }^{\prime },\mathbf{\sigma }%
^{\prime }\right) ^{\prime }$ for notation simplicity. The first-order
condition, $\nabla _{\mathbf{\eta }}\hat{\mathbf{g}}_{n}\left( \hat{\mathbf{%
\eta }},\hat{\mathbf{\gamma }}\right) \mathbf{A}_{n}\hat{\mathbf{g}}%
_{n}\left( \hat{\mathbf{\eta }},\hat{\mathbf{\gamma }}\right) =\mathbf{0}$,
holds with probability 1. Denote $\hat{\mathbf{G}}\left( \mathbf{\eta },%
\mathbf{\gamma }\right) =\nabla _{\mathbf{\eta }}\hat{\mathbf{g}}_{n}\left( 
\mathbf{\eta },\mathbf{\gamma }\right) $ and expand $\hat{\mathbf{g}}%
_{n}\left( \mathbf{\hat{\eta}},\hat{\mathbf{\gamma }}\right) $ in the
first-order condition around $\mathbf{\eta }_{0}$, we have 
\begin{align*}
\sqrt{n}\left( \hat{\mathbf{\eta }}-\mathbf{\eta }_{0}\right) & =-\left[ 
\hat{\mathbf{G}}\left( \hat{\mathbf{\eta }},\hat{\mathbf{\gamma }}\right)
^{\prime }\mathbf{A}_{n}\hat{\mathbf{G}}\left( \bar{\mathbf{\eta }},\hat{%
\mathbf{\gamma }}\right) \right] ^{-1}\hat{\mathbf{G}}\left( \hat{\mathbf{%
\eta }},\hat{\mathbf{\ \gamma }}\right) ^{\prime }\mathbf{A}_{n}\left( \sqrt{%
n}\hat{\mathbf{g}}_{n}\left( \mathbf{\eta }_{0},\hat{\mathbf{\gamma }}%
\right) \right) \\
& =-\left[ \hat{\mathbf{G}}\left( \hat{\mathbf{\eta }},\hat{\mathbf{\gamma }}%
\right) ^{\prime }\mathbf{A}_{n}\hat{\mathbf{G}}\left( \bar{\mathbf{\eta }},%
\hat{\mathbf{\gamma }}\right) \right] ^{-1}\hat{\mathbf{G}}\left( \hat{%
\mathbf{\eta }},\hat{\mathbf{\ \gamma }}\right) ^{\prime }\mathbf{A}_{n}%
\left[ \sqrt{n}\hat{\mathbf{g}}_{n}\left( \mathbf{\eta }_{0},\mathbf{\gamma }%
_{0}\right) +\nabla _{\gamma }\hat{\mathbf{g}}_{n}\left( \mathbf{\eta }_{0},%
\bar{\mathbf{\gamma }}\right) \sqrt{n}\left( \hat{\mathbf{\gamma }}-\mathbf{%
\gamma }_{0}\right) \right],
\end{align*}%
where $\bar{\mathbf{\eta }}$ and $\bar{\mathbf{\gamma }}$ are between $\hat{%
\mathbf{\eta }}$ and $\mathbf{\eta }_{0}$; and $\hat{\mathbf{\gamma }}$ and $%
\mathbf{\gamma }_{0}$, respectively. Note that by term-by-term convergence,
we have $\hat{\mathbf{G}}\left( \hat{\mathbf{\eta }},\hat{\mathbf{\gamma }}%
\right) ,\hat{\mathbf{G}}\left( \bar{\mathbf{\eta }},\hat{\mathbf{\gamma }}%
\right) \rightarrow _{p}\mathbf{G}_{0}$ and $\nabla _{\mathbf{\gamma }}\hat{%
\mathbf{g} }_{n}\left( \mathbf{\eta }_{0},\bar{\mathbf{\gamma }}\right)
\rightarrow_{p} \nabla _{\mathbf{\gamma }}\mathbf{g}_{0}\left( \mathbf{\eta }%
_{0,}\mathbf{\gamma }_{0}\right) =\mathbf{G}_{0, \gamma }$. By Assumption %
\ref{assu:Consistency Assumption}(b), $\mathbf{A}_{n}\rightarrow _{p}\mathbf{%
A}$. By Assumption \ref{assu:normality}(a) and (b) and Slutsky theorem, 
\begin{equation*}
\sqrt{n}\left( \hat{\mathbf{\eta }}-\mathbf{\eta }_{0}\right) \rightarrow
_{d}\left( \mathbf{G}_{0}^{\prime }\mathbf{A}\mathbf{G}_{0}\right) ^{-1}%
\mathbf{G}_{0}^{\prime }\mathbf{A}\left( \mathbf{\zeta }+\mathbf{G}_{0,
\gamma}\mathbf{\zeta }_{\gamma }\right) ,
\end{equation*}%
which completes the proof.
\end{proof}

\medskip

\begin{proof}[Further details for Example 4]
We need to verify the invertibility of the matrix 
\begin{equation*}
\mathbf{B}=%
\begin{pmatrix}
1 & 1 & 0 & 0 \\ 
0 & 0 & 1 & 1 \\ 
1 & 0 & 1 & 0 \\ 
b_{1L}b_{2L} & b_{1L}b_{2H} & b_{1H}b_{2L} & b_{1H}b_{2H}%
\end{pmatrix}%
.
\end{equation*}%
The span of first three rows of $\mathbf{B}$ is 
\begin{equation*}
\mathcal{S}=\left\{ \left( \alpha _{1}+\alpha _{3},\alpha _{1},\alpha
_{2}+\alpha _{3},\alpha _{3}\right) ^{\prime }:\alpha _{1},\alpha
_{2},\alpha _{3}\in \mathbb{R}\right\} .
\end{equation*}%
$\left( b_{1L}b_{2L},b_{1L}b_{2H},b_{1H}b_{2L},b_{1H}b_{2H}\right) ^{\prime
}\notin \mathcal{S}$ is equivalent to $b_{1H}b_{2H}-b_{1H}b_{2L}\neq
b_{1L}b_{2H}-b_{1L}b_{2L}$. This can be verified by 
\begin{equation*}
\left( b_{1H}b_{2H}-b_{1H}b_{2L}\right) -\left(
b_{1L}b_{2H}-b_{1L}b_{2L}\right) =\left( b_{1H}-b_{1L}\right) \left(
b_{2H}-b_{2L}\right) >0,
\end{equation*}%
given that $b_{1L}<b_{1H}$ and $b_{2L}<b_{2H}$ hold.
\end{proof}




\clearpage

\pagebreak

\begin{center}
Online Supplement to

\bigskip

\textbf{\ Identification and Estimation of Categorical Random Coefficient
Models}

\bigskip

by

\bigskip

Zhan Gao and M. Hashem Pesaran

\bigskip

February 2023 \thispagestyle{empty}

\pagebreak
\end{center}

\setcounter{table}{0} \renewcommand{\thetable}{S.\arabic{table}} %
\setcounter{section}{0} \renewcommand{\thesection}{S.\arabic{section}} %
\setcounter{figure}{0} \renewcommand{\thefigure}{S.\arabic{figure}} %
\setcounter{page}{1} \renewcommand{\thepage}{S.\arabic{page}}

\section{Introduction}

This online supplement is composed of four sections. Section \ref%
{suppsec:Proofs} provides additional proofs and technical details omitted
from the main text. Section \ref{suppsec:simulation} provides additional
simulation results. Section \ref{sec:empirical_supp} gives additional
empirical results. Details of the computational algorithm used are described
in Section \ref{sec:computation}.


\section{Proofs\label{suppsec:Proofs}}

We include omitted proofs and technical details in this section.

\begin{proof}[Proof of Theorem 3]
From \eqref{eq: model_w_phi}, we have 
\begin{equation*}
\frac{1}{n}\sum_{i=1}^{n}\mathbf{w}_{i}y_{i}=\frac{1}{n}\sum_{i=1}^{n}%
\mathbf{w}_{i}\mathbf{w}_{i}^{\prime }\mathbf{\phi }+\frac{1}{n}%
\sum_{i=1}^{n}\mathbf{w}_{i}\xi _{i},
\end{equation*}%
where $\mathbf{\phi }=\mathrm{E}\left( \mathbf{\phi }_{i}\right) =\left( 
\mathrm{E}\left( \beta _{i}\right) ,\mathbf{\gamma }^{\prime }\right)
^{\prime },$ and $\xi _{i}=u_{i}+x_{i}v_{i}$, which can be written
equivalently as 
\begin{equation*}
\mathbf{q}_{n,wy}=\mathbf{Q}_{n,ww}\mathbf{\phi }+\frac{1}{n}\sum_{i=1}^{n}%
\mathbf{w}_{i}\xi _{i}.
\end{equation*}%
Taking expectations of both sides and rearrange terms, we have 
\begin{equation*}
\mathbf{\phi }=\mathrm{E}\left( \mathbf{Q}_{n,ww}\right) ^{-1}\mathrm{E}%
\left( \mathbf{q}_{n,wy}\right) .
\end{equation*}%
Consider 
\begin{align*}
\mathbf{\hat{\phi}}-\mathbf{\phi }& =\mathbf{Q}_{n,ww}^{-1}\mathbf{q}_{n,wy}-%
\mathrm{E}\left( \mathbf{Q}_{n,ww}\right) ^{-1}\mathrm{E}\left( \mathbf{q}%
_{n,wy}\right)  \\
& =\left[ \mathbf{Q}_{n,ww}^{-1}-E\left( \mathbf{Q}_{n,ww}\right)
^{-1}+E\left( \mathbf{Q}_{n,ww}\right) ^{-1}\right] \left[ \mathbf{q}_{n,wy}-%
\mathrm{E}\left( \mathbf{q}_{n,wy}\right) +\mathrm{E}\left( \mathbf{q}%
_{n,wy}\right) \right] -\mathrm{E}\left( \mathbf{Q}_{n,ww}\right) ^{-1}%
\mathrm{E}\left( \mathbf{q}_{n,wy}\right)  \\
& =\left[ \mathbf{Q}_{n,ww}^{-1}-E\left( \mathbf{Q}_{n,ww}\right) ^{-1}%
\right] \left[ \mathbf{q}_{n,wy}-\mathrm{E}\left( \mathbf{q}_{n,wy}\right) %
\right] +\left[ \mathbf{Q}_{n,ww}^{-1}-\mathrm{E}\left( \mathbf{Q}%
_{n,ww}\right) ^{-1}\right] \mathrm{E}\left( \mathbf{q}_{n,wy}\right)  \\
& \quad \quad \quad +\mathrm{E}\left( \mathbf{Q}_{n,ww}\right) ^{-1}\left[ 
\mathbf{q}_{n,wy}-\mathrm{E}\left( \mathbf{q}_{n,wy}\right) \right] .
\end{align*}%
Then, 
\begin{align*}
\left\Vert \hat{\mathbf{\phi }}-\mathbf{\phi }\right\Vert & \leq \left\Vert 
\mathbf{Q}_{n,ww}^{-1}-\mathrm{E}\left( \mathbf{Q}_{n,ww}\right)
^{-1}\right\Vert \left\Vert \mathbf{q}_{n,wy}-\mathrm{E}\left( \mathbf{q}%
_{n,wy}\right) \right\Vert +\left\Vert \mathbf{Q}_{n,ww}^{-1}-\mathrm{E}%
\left( \mathbf{Q}_{n,ww}\right) ^{-1}\right\Vert \left\Vert \mathrm{E}\left( 
\mathbf{q}_{n,wy}\right) \right\Vert  \\
& \quad \quad +\left\Vert \mathrm{E}\left( \mathbf{Q}_{n,ww}\right)
^{-1}\right\Vert \left\Vert \mathbf{q}_{n,wy}-\mathrm{E}\left( \mathbf{q}%
_{n,wy}\right) \right\Vert .
\end{align*}%
By Assumption \ref{assu:identification_regularity_condition}(c), we have $%
\left\Vert \mathbf{Q}_{n,ww}^{-1}-\mathrm{E}\left( \mathbf{Q}_{n,ww}\right)
^{-1}\right\Vert =O_{p}\left( n^{-1/2}\right) $, $\left\Vert \mathbf{q}%
_{n,wy}-\mathrm{E}\left( \mathbf{q}_{n,wy}\right) \right\Vert =O_{p}\left(
n^{-1/2}\right) $, and by Assumption \ref%
{assu:identification_regularity_condition}(b), $\left\Vert \mathrm{E}\left( 
\mathbf{q}_{n,wy}\right) \right\Vert $ and $\left\Vert \mathrm{E}\left( 
\mathbf{Q}_{n,ww}\right) ^{-1}\right\Vert $ are bounded. Thus, 
\begin{equation}
\left\Vert \hat{\mathbf{\phi }}-\mathbf{\phi }\right\Vert =O_{p}\left(
n^{-1/2}\right) .  \label{phihatorder}
\end{equation}

To establish the asymptotic distribution of $\hat{\mathbf{\phi }}$, we first
note that 
\begin{equation*}
\sqrt{n}\left( \hat{\mathbf{\phi }}-\mathbf{\phi }\right) =\mathbf{Q}%
_{n,ww}^{-1}\left( n^{-1/2}\sum_{i=1}^{n}\mathbf{w}_{i}\xi _{i}\right) .
\end{equation*}%
By Assumption \ref{assu:gamma_est_normality}, we have 
\begin{equation*}
\mathrm{var}\left( n^{-1/2}\sum_{i=1}^{n}\mathbf{w}_{i}\xi _{i}\right) =%
\frac{1}{n}\sum_{i=1}^{n}\mathrm{var}\left( \mathbf{w}_{i}\xi _{i}\right) =%
\frac{1}{n}\sum_{i=1}^{n}\mathrm{E}\left( \mathbf{w}_{i}\mathbf{w}%
_{i}^{\prime }\xi _{i}^{2}\right) \rightarrow \mathbf{V}_{w\xi }\succ 0.
\end{equation*}%
Note that $\xi _{i}=u_{i}+x_{i}v_{i}$, and $\mathbf{w}_{i}$ is distributed
independently of $u_{i}$ and $v_{i}$. Then%
\begin{equation*}
\mathbf{w}_{i}\xi _{i}=\mathbf{w}_{i}\left( u_{i}+x_{i}v_{i}\right) =\mathbf{%
w}_{i}u_{i}+\left( \mathbf{w}_{i}x_{i}\right) v_{i},
\end{equation*}%
and by Minkowski's inequality, for $r=2+\delta $ with $0<\delta <1$, 
\begin{equation*}
\left[ E\left\Vert \mathbf{w}_{i}\xi _{i}\right\Vert ^{r}\right] ^{1/r}\leq %
\left[ E\left\Vert \mathbf{w}_{i}u_{i}\right\Vert ^{r}\right] ^{1/r}+\left[
E\left\Vert \left( \mathbf{w}_{i}x_{i}\right) v_{i}\right\Vert ^{r}\right]
^{1/r}.
\end{equation*}%
Due to the independence of $u_{i}$ and $v_{i}$ from $\mathbf{w}_{i}$, we have%
\begin{equation*}
\mathrm{E}\left( \left\Vert \mathbf{w}_{i}u_{i}\right\Vert ^{r}\right) \leq
E\left\Vert \mathbf{w}_{i}\right\Vert ^{r}E\left\Vert u_{i}\right\Vert ^{r},%
\text{ and }E\left\Vert \left( \mathbf{w}_{i}x_{i}^{\prime }\right)
v_{i}\right\Vert ^{r}\leq E\left\Vert \mathbf{w}_{i}x_{i}\right\Vert
^{r}E\left\Vert v_{i}\right\Vert ^{r}.
\end{equation*}%
Also, $E\left\Vert \mathbf{w}_{i}x_{i}\right\Vert ^{r}\leq E\left\Vert
\left( x_{i}^{2},x_{i}\mathbf{z}_{i}^{\prime }\right) ^{\prime }\right\Vert
^{r}\leq E\left\Vert \mathbf{w}_{i}\mathbf{w}_{i}^{\prime }\right\Vert
^{r}\leq E\left\Vert \mathbf{w}_{i}\right\Vert ^{2r}$, where $2<r<3$, and
hence $2r<6$. By Assumptions \ref{assu:identification_regularity_condition}%
(a.ii) and \ref{assu:identification_regularity_condition}(b.ii), we have $%
\sup_{i}\mathrm{E}\left( \left\Vert \mathbf{w}_{i}\right\Vert ^{6}\right) <C$%
, $\sup_{i}\mathrm{E}\left( \left\Vert u_{i}\right\Vert ^{3}\right) <C$, and 
$\mathrm{E}\left( \left\Vert v_{i}\right\Vert ^{3}\right) \leq \max_{1\leq
k\leq K}\left\vert b_{k}-\mathrm{E}\left( \beta _{i}\right) \right\vert
^{3}<C.$ Then, we verified that $\sup_{i}\mathrm{E}\left( \left\Vert \mathbf{%
w}_{i}u_{i}\right\Vert ^{r}\right) <C$, and $E\left\Vert \left( \mathbf{w}%
_{i}x_{i}^{\prime }\right) v_{i}\right\Vert ^{r}<C$, and hence the Lyapunov
condition that $\sup_{i}\mathrm{E}\left( \left\Vert \mathbf{w}_{i}\xi
_{i}\right\Vert ^{r}\right) <C$, where $r=2+\delta \in (2,3)$. By the
central limit theorem for independent but not necessarily identically
distributed random vectors (see Pesaran (2015, Theorem 18) or Hansen (2022,
Theorem 6.5)), we have 
\begin{equation*}
\frac{1}{\sqrt{n}}\sum_{i=1}^{n}\mathbf{w}_{i}\xi _{i}\rightarrow _{d}N(%
\mathbf{0},\mathbf{V}_{w\xi }),
\end{equation*}%
as $n\rightarrow \infty $, and by Assumption \ref%
{assu:identification_regularity_condition} and continuous mapping theorem, 
\begin{equation*}
\sqrt{n}(\hat{\mathbf{\phi }}-\mathbf{\phi })\rightarrow _{d}N\left( \mathbf{%
0},\mathbf{Q}_{ww}^{-1}\mathbf{V}_{w\xi }\mathbf{Q}_{ww}^{-1}\right) .
\end{equation*}

We then turn to the consistent estimation of the variance matrix. By
Assumption \ref{assu:gamma_est_normality}, we have 
\begin{align}
\left\Vert \hat{\mathbf{V}}_{w\xi }-\mathbf{V}_{w\xi }\right\Vert &
=\left\Vert \frac{1}{n}\sum_{i=1}^{n}\mathbf{w}_{i}\mathbf{w}_{i}^{\prime }%
\hat{\xi}_{i}^{2}-\frac{1}{n}\sum_{i=1}^{n}\mathrm{E}\left( \mathbf{w}_{i}%
\mathbf{w}_{i}\xi _{i}^{2}\right) +\frac{1}{n}\sum_{i=1}^{n}\mathrm{E}\left( 
\mathbf{w}_{i}\mathbf{w}_{i}\xi _{i}^{2}\right) -\mathbf{V}_{w\xi
}\right\Vert   \notag \\
& \leq \left\Vert \frac{1}{n}\sum_{i=1}^{n}\mathbf{w}_{i}\mathbf{w}%
_{i}^{\prime }\xi _{i}^{2}-\frac{1}{n}\sum_{i=1}^{n}\mathrm{E}\left( \mathbf{%
w}_{i}\mathbf{w}_{i}\xi _{i}^{2}\right) \right\Vert +\left\Vert \frac{1}{n}%
\sum_{i=1}^{n}\mathrm{E}\left( \mathbf{w}_{i}\mathbf{w}_{i}\xi
_{i}^{2}\right) -\mathbf{V}_{w\xi }\right\Vert   \notag \\
& \quad \quad +\left\Vert \frac{1}{n}\sum_{i=1}^{n}\mathbf{w}_{i}\mathbf{w}%
_{i}^{\prime }\left( \hat{\xi}_{i}^{2}-\xi _{i}^{2}\right) \right\Vert  
\notag \\
& \leq \frac{1}{n}\sum_{i=1}^{n}\left\Vert \mathbf{w}_{i}\right\Vert
^{2}\left\vert \hat{\xi}_{i}^{2}-\xi _{i}^{2}\right\vert +O_{p}(n^{-1/2}).
\label{eq:bound_V_hat}
\end{align}%
Note that $\hat{\xi}_{i}=\xi _{i}-\left( \hat{\mathbf{\phi }}-\mathbf{\phi }%
\right) ^{\prime }\mathbf{w}_{i}$, then 
\begin{align}
\left\vert \hat{\xi}_{i}^{2}-\xi _{i}^{2}\right\vert & \leq 2\left\vert \xi
_{i}\mathbf{w}_{i}^{\prime }\left( \hat{\mathbf{\phi }}-\mathbf{\phi }%
\right) \right\vert +\left( \hat{\mathbf{\phi }}-\mathbf{\phi }\right)
^{\prime }\left( \mathbf{w}_{i}\mathbf{w}_{i}^{\prime }\right) \left( \hat{%
\mathbf{\phi }}-\mathbf{\phi }\right)   \notag \\
& \leq 2\left\vert \xi _{i}\right\vert \left\Vert \mathbf{w}_{i}\right\Vert
\left\Vert \hat{\mathbf{\phi }}-\mathbf{\phi }\right\Vert +\left\Vert 
\mathbf{w}_{i}\right\Vert ^{2}\left\Vert \hat{\mathbf{\phi }}-\mathbf{\phi }%
\right\Vert ^{2}.  \label{eq:bound_xi_hat_square}
\end{align}%
Combine \eqref{eq:bound_V_hat} and \eqref{eq:bound_xi_hat_square}, we have 
\begin{equation}
\left\Vert \hat{\mathbf{V}}_{w\xi }-\mathbf{V}_{w\xi }\right\Vert \leq
2\left( \frac{1}{n}\sum_{i=1}^{n}\left\Vert \mathbf{w}_{i}\right\Vert
^{3}\left\vert \xi _{i}\right\vert \right) \left\Vert \hat{\mathbf{\phi }}-%
\mathbf{\phi }\right\Vert +\left( \frac{1}{n}\sum_{i=1}^{n}\left\Vert 
\mathbf{w}_{i}\right\Vert ^{4}\right) \left\Vert \hat{\mathbf{\phi }}-%
\mathbf{\phi }\right\Vert ^{2} + O_p\left( n^{-1/2} \right).  \label{eq:bound_V_hat_cont}
\end{equation}
By H\"{o}lder's inequality, 
\begin{equation}
\frac{1}{n}\sum_{i=1}^{n}\left\Vert \mathbf{w}_{i}\right\Vert ^{3}\left\vert
\xi _{i}\right\vert \leq \left( \frac{1}{n}\sum_{i=1}^{n}\left\Vert \mathbf{w%
}_{i}\right\Vert ^{4}\right) ^{3/4}\left( \frac{1}{n}\sum_{i=1}^{n}\xi
_{i}^{4}\right) ^{1/4}.  \label{eq:sum_w3_xi}
\end{equation}%
By Assumption \ref{assu:identification_regularity_condition}(b.iii), $%
n^{-1}\sum_{i=1}^{n}\left\Vert \mathbf{w}_{i}\right\Vert ^{4}=O_{p}(1)$. By
Minkowski inequality, 
\begin{align*}
\left( \frac{1}{n}\sum_{i=1}^{n}\xi _{i}^{4}\right) ^{1/4}& =\left( \frac{1}{%
n}\sum_{i=1}^{n}\left( u_{i}+x_{i}v_{i}\right) ^{4}\right) ^{1/4}\leq \left( 
\frac{1}{n}\sum_{i=1}^{n}u_{i}^{4}\right) ^{1/4}+\left( \frac{1}{n}%
\sum_{i=1}^{n}x_{i}^{4}v_{i}^{4}\right) ^{1/4} \\
& \leq \left( \frac{1}{n}\sum_{i=1}^{n}u_{i}^{4}\right)
^{1/4}+\max_{k}\left\{ \left\vert b_{k}-\mathrm{E}\left( \beta _{i}\right)
\right\vert \right\} \left( \frac{1}{n}\sum_{i=1}^{n}x_{i}^{4}\right) ^{1/4}
\\
& =O_{p}(1),
\end{align*}%
where the last inequality is from Assumptions \ref%
{assu:identification_regularity_condition}(a.iii) and (b.iii) that $%
n^{-1}\sum_{i=1}^{n}u_{i}^{4}=O_{p}\left( 1\right) $, and $%
n^{-1}\sum_{i=1}^{n}x_{i}^{4}\leq n^{-1}\sum_{i=1}^{n}\left\Vert \mathbf{w}%
_{i}\right\Vert ^{4}=O_{p}(1)$. Then we verified in \eqref{eq:sum_w3_xi}
that $n^{-1} \sum_{i=1}^{n}\left\Vert \mathbf{w}_{i}\right\Vert
^{3}\left\vert \xi _{i}\right\vert =O_{p}(1)$. Then using the above results
in \eqref{eq:bound_V_hat_cont}, and noting from (\ref{phihatorder}) that $%
\left\Vert \hat{\mathbf{\phi }}-\mathbf{\phi }\right\Vert =O_{p}\left(
n^{-1/2}\right) $, we have $\left\Vert \hat{\mathbf{V}}_{w\xi }-\mathbf{V}%
_{w\xi }\right\Vert =O_{p}\left( n^{-1/2}\right) $, as required.
\end{proof}

\section{Monte Carlo Simulation\label{suppsec:simulation}}

\subsection{Results with $S = 5$ and $S = 6$\label{subsec:mc_s5_s6}}

Tables \ref{tab:mc_S5} and \ref{tab:mc_S6} present the summary results
corresponding to $S=5$ and $S=6$, for the data generating processes
described in Section 5.1. These results show that adding more moments does
not necessarily improve the estimation accuracy but could be
counter-productive.

\begin{table}[tbp]
\caption{Bias, RMSE and size of the GMM estimator for distributional
parameters of $\protect\beta$ with $S= 5$ }
\label{tab:mc_S5}
\begin{center}
{\small \ 
\begin{tabular}{crlllllllll}
\hline
\multicolumn{2}{r|}{DGP} & \multicolumn{3}{c|}{Baseline} & 
\multicolumn{3}{c|}{Categorical $x$} & \multicolumn{3}{c}{Categorical $u$}
\\ \hline
\multicolumn{2}{c|}{Sample size $n$} & \multicolumn{1}{c}{Bias} & 
\multicolumn{1}{c}{RMSE} & \multicolumn{1}{c|}{Size} & \multicolumn{1}{c}{
Bias} & \multicolumn{1}{c}{RMSE} & \multicolumn{1}{c|}{Size} & 
\multicolumn{1}{c}{Bias} & \multicolumn{1}{c}{RMSE} & \multicolumn{1}{c}{Size
} \\ \hline
\multicolumn{11}{c}{\textit{high variance}: $\mathrm{var}\left( \beta_i
\right) = 0.25$} \\ \hline
\multirow{7}{*}{\begin{turn}{90} $\pi=0.5$ \end{turn}} & \multicolumn{1}{r|}{
100} & 0.0308 & 0.1869 & \multicolumn{1}{l|}{0.1021} & 0.0259 & 0.1986 & 
\multicolumn{1}{l|}{0.1276} & 0.0106 & 0.1944 & 0.1050 \\ 
& \multicolumn{1}{r|}{1,000} & 0.0048 & 0.1235 & \multicolumn{1}{l|}{0.1950}
& 0.0054 & 0.1334 & \multicolumn{1}{l|}{0.2112} & -0.0364 & 0.1638 & 0.2239
\\ 
& \multicolumn{1}{r|}{2,000} & -0.0006 & 0.0875 & \multicolumn{1}{l|}{0.1641}
& -0.0009 & 0.0962 & \multicolumn{1}{l|}{0.1887} & -0.0238 & 0.1172 & 0.2059
\\ 
& \multicolumn{1}{r|}{5,000} & -0.0005 & 0.0484 & \multicolumn{1}{l|}{0.1339}
& -0.0001 & 0.0591 & \multicolumn{1}{l|}{0.1602} & -0.0125 & 0.0740 & 0.1667
\\ 
& \multicolumn{1}{r|}{10,000} & -0.0002 & 0.0334 & \multicolumn{1}{l|}{0.1152
} & -0.0005 & 0.0373 & \multicolumn{1}{l|}{0.1246} & -0.0080 & 0.0519 & 
0.1386 \\ 
& \multicolumn{1}{r|}{100,000} & -0.0002 & 0.0096 & \multicolumn{1}{l|}{
0.0636} & 0.0001 & 0.0116 & \multicolumn{1}{l|}{0.0738} & -0.0008 & 0.0174 & 
0.0766 \\ \hline
\multirow{7}{*}{\begin{turn}{90} $\beta_L = 1$ \end{turn}} & 
\multicolumn{1}{r|}{100} & 0.2234 & 0.4541 & \multicolumn{1}{l|}{0.3205} & 
0.1992 & 0.4777 & \multicolumn{1}{l|}{0.2843} & 0.1780 & 0.5090 & 0.2519 \\ 
& \multicolumn{1}{r|}{1,000} & 0.0503 & 0.1609 & \multicolumn{1}{l|}{0.3060}
& 0.0475 & 0.1812 & \multicolumn{1}{l|}{0.2963} & 0.0100 & 0.2024 & 0.2141
\\ 
& \multicolumn{1}{r|}{2,000} & 0.0265 & 0.1148 & \multicolumn{1}{l|}{0.2501}
& 0.0257 & 0.1262 & \multicolumn{1}{l|}{0.2501} & 0.0088 & 0.1337 & 0.1905
\\ 
& \multicolumn{1}{r|}{5,000} & 0.0108 & 0.0606 & \multicolumn{1}{l|}{0.1926}
& 0.0130 & 0.0702 & \multicolumn{1}{l|}{0.2042} & 0.0031 & 0.0803 & 0.1641
\\ 
& \multicolumn{1}{r|}{10,000} & 0.0054 & 0.0409 & \multicolumn{1}{l|}{0.1408}
& 0.0061 & 0.0456 & \multicolumn{1}{l|}{0.1510} & 0.0008 & 0.0527 & 0.1338
\\ 
& \multicolumn{1}{r|}{100,000} & 0.0004 & 0.0114 & \multicolumn{1}{l|}{0.0716
} & 0.0006 & 0.0134 & \multicolumn{1}{l|}{0.0790} & 0.0002 & 0.0184 & 0.0834
\\ \hline
\multirow{7}{*}{\begin{turn}{90} $\beta_H = 2$ \end{turn}} & 
\multicolumn{1}{r|}{100} & -0.1956 & 0.5486 & \multicolumn{1}{l|}{0.2448} & 
-0.1941 & 0.5638 & \multicolumn{1}{l|}{0.2386} & -0.2029 & 0.5801 & 0.2269
\\ 
& \multicolumn{1}{r|}{1,000} & -0.0418 & 0.2080 & \multicolumn{1}{l|}{0.3299}
& -0.0414 & 0.2300 & \multicolumn{1}{l|}{0.3384} & -0.0752 & 0.2583 & 0.3620
\\ 
& \multicolumn{1}{r|}{2,000} & -0.0264 & 0.1379 & \multicolumn{1}{l|}{0.2799}
& -0.0286 & 0.1554 & \multicolumn{1}{l|}{0.2860} & -0.0529 & 0.1789 & 0.3048
\\ 
& \multicolumn{1}{r|}{5,000} & -0.0113 & 0.0696 & \multicolumn{1}{l|}{0.2008}
& -0.0116 & 0.0883 & \multicolumn{1}{l|}{0.2170} & -0.0254 & 0.1038 & 0.2411
\\ 
& \multicolumn{1}{r|}{10,000} & -0.0053 & 0.0432 & \multicolumn{1}{l|}{0.1502
} & -0.0064 & 0.0520 & \multicolumn{1}{l|}{0.1642} & -0.0156 & 0.0690 & 
0.2002 \\ 
& \multicolumn{1}{r|}{100,000} & -0.0007 & 0.0113 & \multicolumn{1}{l|}{
0.0662} & -0.0004 & 0.0135 & \multicolumn{1}{l|}{0.0764} & -0.0016 & 0.0209
& 0.0818 \\ \hline
\multicolumn{11}{c}{\textit{low variance}: $\mathrm{var}\left( \beta_i
\right) = 0.15$} \\ \hline
\multirow{7}{*}{\begin{turn}{90} $\pi=0.3$ \end{turn}} & \multicolumn{1}{r|}{
100} & 0.2214 & 0.2820 & \multicolumn{1}{l|}{0.1063} & 0.2291 & 0.2942 & 
\multicolumn{1}{l|}{0.1328} & 0.2212 & 0.2876 & 0.1221 \\ 
& \multicolumn{1}{r|}{1,000} & 0.0477 & 0.1746 & \multicolumn{1}{l|}{0.2235}
& 0.0605 & 0.1928 & \multicolumn{1}{l|}{0.2430} & 0.0348 & 0.2039 & 0.2900
\\ 
& \multicolumn{1}{r|}{2,000} & 0.0217 & 0.1198 & \multicolumn{1}{l|}{0.2020}
& 0.0262 & 0.1331 & \multicolumn{1}{l|}{0.2246} & -0.0080 & 0.1608 & 0.2822
\\ 
& \multicolumn{1}{r|}{5,000} & 0.0112 & 0.0709 & \multicolumn{1}{l|}{0.1732}
& 0.0154 & 0.0828 & \multicolumn{1}{l|}{0.1956} & -0.0115 & 0.1072 & 0.2289
\\ 
& \multicolumn{1}{r|}{10,000} & 0.0063 & 0.0465 & \multicolumn{1}{l|}{0.1588}
& 0.0106 & 0.0576 & \multicolumn{1}{l|}{0.1649} & -0.0075 & 0.0761 & 0.1890
\\ 
& \multicolumn{1}{r|}{100,000} & 0.0001 & 0.0130 & \multicolumn{1}{l|}{0.0810
} & 0.0014 & 0.0158 & \multicolumn{1}{l|}{0.0882} & 0.0040 & 0.0280 & 0.0978
\\ \hline
\multirow{7}{*}{\begin{turn}{90} $\beta_L = 0.5$ \end{turn}} & 
\multicolumn{1}{r|}{100} & 0.4245 & 0.5722 & \multicolumn{1}{l|}{0.2938} & 
0.4048 & 0.5818 & \multicolumn{1}{l|}{0.2612} & 0.3827 & 0.6052 & 0.2278 \\ 
& \multicolumn{1}{r|}{1,000} & 0.1300 & 0.2692 & \multicolumn{1}{l|}{0.3058}
& 0.1300 & 0.2890 & \multicolumn{1}{l|}{0.3057} & 0.0882 & 0.3673 & 0.1970
\\ 
& \multicolumn{1}{r|}{2,000} & 0.0763 & 0.1746 & \multicolumn{1}{l|}{0.3147}
& 0.0735 & 0.1903 & \multicolumn{1}{l|}{0.2820} & 0.0149 & 0.2523 & 0.1964
\\ 
& \multicolumn{1}{r|}{5,000} & 0.0378 & 0.1018 & \multicolumn{1}{l|}{0.2690}
& 0.0410 & 0.1155 & \multicolumn{1}{l|}{0.2695} & 0.0034 & 0.1417 & 0.1905
\\ 
& \multicolumn{1}{r|}{10,000} & 0.0202 & 0.0674 & \multicolumn{1}{l|}{0.2344}
& 0.0257 & 0.0822 & \multicolumn{1}{l|}{0.2404} & 0.0013 & 0.0961 & 0.1690
\\ 
& \multicolumn{1}{r|}{100,000} & 0.0013 & 0.0184 & \multicolumn{1}{l|}{0.0952
} & 0.0026 & 0.0221 & \multicolumn{1}{l|}{0.1042} & 0.0060 & 0.0347 & 0.1112
\\ \hline
\multirow{7}{*}{\begin{turn}{90} $\beta_H = 1.345$ \end{turn}} & 
\multicolumn{1}{r|}{100} & -0.0646 & 0.3773 & \multicolumn{1}{l|}{0.1781} & 
-0.0616 & 0.4058 & \multicolumn{1}{l|}{0.1668} & -0.0564 & 0.4357 & 0.1688
\\ 
& \multicolumn{1}{r|}{1,000} & -0.0180 & 0.1523 & \multicolumn{1}{l|}{0.2496}
& -0.0119 & 0.1804 & \multicolumn{1}{l|}{0.2615} & -0.0476 & 0.2022 & 0.2721
\\ 
& \multicolumn{1}{r|}{2,000} & -0.0104 & 0.1021 & \multicolumn{1}{l|}{0.2375}
& -0.0101 & 0.1147 & \multicolumn{1}{l|}{0.2414} & -0.0381 & 0.1448 & 0.2830
\\ 
& \multicolumn{1}{r|}{5,000} & -0.0027 & 0.0549 & \multicolumn{1}{l|}{0.1680}
& -0.0016 & 0.0680 & \multicolumn{1}{l|}{0.1936} & -0.0193 & 0.0927 & 0.2369
\\ 
& \multicolumn{1}{r|}{10,000} & -0.0001 & 0.0368 & \multicolumn{1}{l|}{0.1458
} & 0.0007 & 0.0438 & \multicolumn{1}{l|}{0.1458} & -0.0115 & 0.0634 & 0.1976
\\ 
& \multicolumn{1}{r|}{100,000} & -0.0002 & 0.0102 & \multicolumn{1}{l|}{
0.0726} & 0.0005 & 0.0120 & \multicolumn{1}{l|}{0.0688} & 0.0021 & 0.0214 & 
0.0902 \\ \hline
\end{tabular}
}
\end{center}
\par
{\footnotesize \textit{Notes:} The data generating process is %
\eqref{eq:mc_dgp}. \textit{high variance} and \textit{low variance}
parametrization are described in \eqref{eq:mc_dgp_para}. ``Baseline'',
``Categorical $x$'' and ``Categorical $u$'' refer to DGP 1 to 3 as in
Section \ref{subsec:dgp}. Generically, bias, RMSE and size are calculated by 
$R^{-1}\sum_{r=1}^R \left( \hat{ \theta}^{(r)} - \theta_0 \right)$, $\sqrt{
R^{-1}\sum_{r= 1}^R \left( \hat{\theta}^{(r)} -\theta_0 \right)^2}$, and $%
R^{-1}\sum_{r=1}^R \mathbf{1}\left[ \left\vert \hat{\theta}^{(r)} - \theta_0
\right\vert / \hat{\sigma}_{\hat{\theta}}^{(r)} > \mathrm{cv}_{0.05} \right] 
$, respectively, for true parameter $\theta_0$, its estimate $\hat{\theta}
^{(r)}$, the estimated standard error of $\hat{\theta}^{(r)}$, $\hat{\sigma}
_{\hat{\theta}}^{(r)}$, and the critical value $\mathrm{cv}_{0.05} =
\Phi^{-1}\left( 0.975 \right)$ across $R = 5,000$ replications, where $%
\Phi\left( \cdot \right)$ is the cumulative distribution function of
standard normal distribution. }
\end{table}

\begin{table}[tbp]
\caption{Bias, RMSE and size of the GMM estimator for distributional
parameters of $\protect\beta$ with $S= 6$ }
\label{tab:mc_S6}
\begin{center}
{\small \ 
\begin{tabular}{crlllllllll}
\hline
\multicolumn{2}{r|}{DGP} & \multicolumn{3}{c|}{Baseline} & 
\multicolumn{3}{c|}{Categorical $x$} & \multicolumn{3}{c}{Categorical $u$}
\\ \hline
\multicolumn{2}{c|}{Sample size $n$} & \multicolumn{1}{c}{Bias} & 
\multicolumn{1}{c}{RMSE} & \multicolumn{1}{c|}{Size} & \multicolumn{1}{c}{
Bias} & \multicolumn{1}{c}{RMSE} & \multicolumn{1}{c|}{Size} & 
\multicolumn{1}{c}{Bias} & \multicolumn{1}{c}{RMSE} & \multicolumn{1}{c}{Size
} \\ \hline
\multicolumn{11}{c}{\textit{high variance}: $\mathrm{var}\left( \beta_i
\right) = 0.25$} \\ \hline
\multirow{7}{*}{\begin{turn}{90} $\pi=0.5$ \end{turn}} & \multicolumn{1}{r|}{
100} & 0.0337 & 0.1472 & \multicolumn{1}{l|}{0.0456} & 0.0293 & 0.1645 & 
\multicolumn{1}{l|}{0.0695} & 0.0227 & 0.1498 & 0.0469 \\ 
& \multicolumn{1}{r|}{1,000} & 0.0021 & 0.1405 & \multicolumn{1}{l|}{0.2545}
& 0.0015 & 0.1469 & \multicolumn{1}{l|}{0.2543} & -0.0265 & 0.1635 & 0.2551
\\ 
& \multicolumn{1}{r|}{2,000} & 0.0008 & 0.1071 & \multicolumn{1}{l|}{0.2614}
& 0.0006 & 0.1185 & \multicolumn{1}{l|}{0.2789} & -0.0201 & 0.1281 & 0.2732
\\ 
& \multicolumn{1}{r|}{5,000} & -0.0020 & 0.0661 & \multicolumn{1}{l|}{0.2261}
& -0.0016 & 0.0765 & \multicolumn{1}{l|}{0.2518} & -0.0142 & 0.0836 & 0.2510
\\ 
& \multicolumn{1}{r|}{10,000} & -0.0005 & 0.0444 & \multicolumn{1}{l|}{0.1844
} & -0.0011 & 0.0505 & \multicolumn{1}{l|}{0.2155} & -0.0093 & 0.0587 & 
0.2323 \\ 
& \multicolumn{1}{r|}{100,000} & 0.0000 & 0.0097 & \multicolumn{1}{l|}{0.0732
} & 0.0000 & 0.0118 & \multicolumn{1}{l|}{0.0912} & -0.0020 & 0.0178 & 0.1162
\\ \hline
\multirow{7}{*}{\begin{turn}{90} $\beta_L = 1$ \end{turn}} & 
\multicolumn{1}{r|}{100} & 0.2226 & 0.4373 & \multicolumn{1}{l|}{0.3341} & 
0.2151 & 0.4658 & \multicolumn{1}{l|}{0.3237} & 0.1879 & 0.4841 & 0.2896 \\ 
& \multicolumn{1}{r|}{1,000} & 0.0721 & 0.2081 & \multicolumn{1}{l|}{0.4485}
& 0.0780 & 0.2197 & \multicolumn{1}{l|}{0.4318} & 0.0531 & 0.2283 & 0.3576
\\ 
& \multicolumn{1}{r|}{2,000} & 0.0443 & 0.1464 & \multicolumn{1}{l|}{0.4056}
& 0.0455 & 0.1609 & \multicolumn{1}{l|}{0.4157} & 0.0342 & 0.1536 & 0.3271
\\ 
& \multicolumn{1}{r|}{5,000} & 0.0175 & 0.0806 & \multicolumn{1}{l|}{0.3035}
& 0.0203 & 0.0923 & \multicolumn{1}{l|}{0.3341} & 0.0150 & 0.0933 & 0.2770
\\ 
& \multicolumn{1}{r|}{10,000} & 0.0092 & 0.0510 & \multicolumn{1}{l|}{0.2350}
& 0.0098 & 0.0594 & \multicolumn{1}{l|}{0.2723} & 0.0081 & 0.0629 & 0.2403
\\ 
& \multicolumn{1}{r|}{100,000} & 0.0010 & 0.0114 & \multicolumn{1}{l|}{0.0850
} & 0.0013 & 0.0136 & \multicolumn{1}{l|}{0.0982} & 0.0002 & 0.0186 & 0.1116
\\ \hline
\multirow{7}{*}{\begin{turn}{90} $\beta_H = 2$ \end{turn}} & 
\multicolumn{1}{r|}{100} & -0.2495 & 0.5629 & \multicolumn{1}{l|}{0.2563} & 
-0.2580 & 0.5681 & \multicolumn{1}{l|}{0.2608} & -0.2589 & 0.5782 & 0.2248
\\ 
& \multicolumn{1}{r|}{1,000} & -0.0618 & 0.2530 & \multicolumn{1}{l|}{0.4938}
& -0.0686 & 0.2733 & \multicolumn{1}{l|}{0.4867} & -0.0962 & 0.2814 & 0.4874
\\ 
& \multicolumn{1}{r|}{2,000} & -0.0334 & 0.1729 & \multicolumn{1}{l|}{0.4454}
& -0.0365 & 0.1951 & \multicolumn{1}{l|}{0.4461} & -0.0625 & 0.2017 & 0.4643
\\ 
& \multicolumn{1}{r|}{5,000} & -0.0189 & 0.1010 & \multicolumn{1}{l|}{0.3457}
& -0.0203 & 0.1178 & \multicolumn{1}{l|}{0.3638} & -0.0383 & 0.1223 & 0.3946
\\ 
& \multicolumn{1}{r|}{10,000} & -0.0080 & 0.0634 & \multicolumn{1}{l|}{0.2670
} & -0.0109 & 0.0732 & \multicolumn{1}{l|}{0.3011} & -0.0246 & 0.0830 & 
0.3347 \\ 
& \multicolumn{1}{r|}{100,000} & -0.0013 & 0.0114 & \multicolumn{1}{l|}{
0.0842} & -0.0012 & 0.0141 & \multicolumn{1}{l|}{0.1070} & -0.0043 & 0.0220
& 0.1396 \\ \hline
\multicolumn{11}{c}{\textit{low variance}: $\mathrm{var}\left( \beta_i
\right) = 0.15$} \\ \hline
\multirow{7}{*}{\begin{turn}{90} $\pi=0.3$ \end{turn}} & \multicolumn{1}{r|}{
100} & 0.2374 & 0.2757 & \multicolumn{1}{l|}{0.0591} & 0.2352 & 0.2816 & 
\multicolumn{1}{l|}{0.0829} & 0.2330 & 0.2771 & 0.0801 \\ 
& \multicolumn{1}{r|}{1,000} & 0.1071 & 0.2107 & \multicolumn{1}{l|}{0.2608}
& 0.1114 & 0.2244 & \multicolumn{1}{l|}{0.2775} & 0.0764 & 0.2158 & 0.2772
\\ 
& \multicolumn{1}{r|}{2,000} & 0.0702 & 0.1661 & \multicolumn{1}{l|}{0.2994}
& 0.0786 & 0.1815 & \multicolumn{1}{l|}{0.3258} & 0.0242 & 0.1806 & 0.3291
\\ 
& \multicolumn{1}{r|}{5,000} & 0.0452 & 0.1101 & \multicolumn{1}{l|}{0.3217}
& 0.0519 & 0.1260 & \multicolumn{1}{l|}{0.3466} & 0.0092 & 0.1263 & 0.3329
\\ 
& \multicolumn{1}{r|}{10,000} & 0.0300 & 0.0816 & \multicolumn{1}{l|}{0.3060}
& 0.0390 & 0.0933 & \multicolumn{1}{l|}{0.3389} & 0.0108 & 0.0954 & 0.3161
\\ 
& \multicolumn{1}{r|}{100,000} & 0.0018 & 0.0164 & \multicolumn{1}{l|}{0.1128
} & 0.0041 & 0.0234 & \multicolumn{1}{l|}{0.1482} & 0.0055 & 0.0298 & 0.1688
\\ \hline
\multirow{7}{*}{\begin{turn}{90} $\beta_L = 0.5$ \end{turn}} & 
\multicolumn{1}{r|}{100} & 0.4146 & 0.5479 & \multicolumn{1}{l|}{0.3137} & 
0.4191 & 0.5636 & \multicolumn{1}{l|}{0.2965} & 0.3844 & 0.5678 & 0.2532 \\ 
& \multicolumn{1}{r|}{1,000} & 0.2445 & 0.3459 & \multicolumn{1}{l|}{0.4601}
& 0.2436 & 0.3579 & \multicolumn{1}{l|}{0.4561} & 0.2080 & 0.3872 & 0.3187
\\ 
& \multicolumn{1}{r|}{2,000} & 0.1663 & 0.2539 & \multicolumn{1}{l|}{0.4809}
& 0.1684 & 0.2620 & \multicolumn{1}{l|}{0.4797} & 0.1108 & 0.2830 & 0.3203
\\ 
& \multicolumn{1}{r|}{5,000} & 0.0977 & 0.1648 & \multicolumn{1}{l|}{0.4800}
& 0.1051 & 0.1788 & \multicolumn{1}{l|}{0.4938} & 0.0590 & 0.1731 & 0.3606
\\ 
& \multicolumn{1}{r|}{10,000} & 0.0613 & 0.1182 & \multicolumn{1}{l|}{0.4230}
& 0.0730 & 0.1315 & \multicolumn{1}{l|}{0.4717} & 0.0417 & 0.1251 & 0.3667
\\ 
& \multicolumn{1}{r|}{100,000} & 0.0050 & 0.0242 & \multicolumn{1}{l|}{0.1420
} & 0.0086 & 0.0333 & \multicolumn{1}{l|}{0.1808} & 0.0101 & 0.0386 & 0.1906
\\ \hline
\multirow{7}{*}{\begin{turn}{90} $\beta_H = 1.345$ \end{turn}} & 
\multicolumn{1}{r|}{100} & -0.0817 & 0.3703 & \multicolumn{1}{l|}{0.1601} & 
-0.0883 & 0.3842 & \multicolumn{1}{l|}{0.1687} & -0.0806 & 0.4136 & 0.1614
\\ 
& \multicolumn{1}{r|}{1,000} & -0.0086 & 0.1726 & \multicolumn{1}{l|}{0.3174}
& -0.0144 & 0.1907 & \multicolumn{1}{l|}{0.3295} & -0.0560 & 0.2029 & 0.3239
\\ 
& \multicolumn{1}{r|}{2,000} & 0.0022 & 0.1194 & \multicolumn{1}{l|}{0.3267}
& 0.0029 & 0.1368 & \multicolumn{1}{l|}{0.3401} & -0.0395 & 0.1582 & 0.3736
\\ 
& \multicolumn{1}{r|}{5,000} & 0.0093 & 0.0722 & \multicolumn{1}{l|}{0.2899}
& 0.0099 & 0.0876 & \multicolumn{1}{l|}{0.3254} & -0.0189 & 0.0998 & 0.3570
\\ 
& \multicolumn{1}{r|}{10,000} & 0.0092 & 0.0535 & \multicolumn{1}{l|}{0.2642}
& 0.0117 & 0.0601 & \multicolumn{1}{l|}{0.2889} & -0.0076 & 0.0733 & 0.3141
\\ 
& \multicolumn{1}{r|}{100,000} & -0.0002 & 0.0116 & \multicolumn{1}{l|}{
0.0972} & 0.0012 & 0.0157 & \multicolumn{1}{l|}{0.1326} & 0.0019 & 0.0220 & 
0.1454 \\ \hline
\end{tabular}
}
\end{center}
\par
{\footnotesize \textit{Notes:} The data generating process is %
\eqref{eq:mc_dgp}. \textit{high variance} and \textit{low variance}
parametrization are described in \eqref{eq:mc_dgp_para}. ``Baseline'',
``Categorical $x$'' and ``Categorical $u$'' refer to DGP 1 to 3 as in
Section \ref{subsec:dgp}. Generically, bias, RMSE and size are calculated by 
$R^{-1}\sum_{r=1}^R \left( \hat{ \theta}^{(r)} - \theta_0 \right)$, $\sqrt{
R^{-1}\sum_{r= 1}^R \left( \hat{\theta}^{(r)} -\theta_0 \right)^2}$, and $%
R^{-1}\sum_{r=1}^R \mathbf{1}\left[ \left\vert \hat{\theta}^{(r)} - \theta_0
\right\vert / \hat{\sigma}_{\hat{\theta}}^{(r)} > \mathrm{cv}_{0.05} \right] 
$, respectively, for true parameter $\theta_0$, its estimate $\hat{\theta}
^{(r)}$, the estimated standard error of $\hat{\theta}^{(r)}$, $\hat{\sigma}
_{\hat{\theta}}^{(r)}$, and the critical value $\mathrm{cv}_{0.05} =
\Phi^{-1}\left( 0.975 \right)$ across $R = 5,000$ replications, where $%
\Phi\left( \cdot \right)$ is the cumulative distribution function of
standard normal distribution. }
\end{table}

\subsection{GMM Estimation of Moments of $\protect\beta_i$\label%
{subsec:mc_gmm_moments}}

With the data generating processes described in Section \ref{subsec:dgp}, we
report the bias, RMSE and size of the GMM estimator for moments of $\beta
_{i}$ in Table \ref{tab:mc_gmm_moments_S4}. The GMM estimator for moments of 
$\beta _{i}$ achieve better small sample performance as compared to those
for the distributional parameters $\pi ,\beta _{L}$ and $\beta _{H}$.

\begin{table}[tbp]
\caption{Bias, RMSE and size of the GMM estimator for moments of $\protect%
\beta$ }
\label{tab:mc_gmm_moments_S4}
\begin{center}
{\small \ 
\begin{tabular}{rrrrrrrrrrr}
\hline
\multicolumn{2}{r|}{DGP} & \multicolumn{3}{c|}{Baseline} & 
\multicolumn{3}{c|}{Categorical $x$} & \multicolumn{3}{c}{Categorical $u$}
\\ \hline
\multicolumn{2}{c|}{Sample size $n$} & \multicolumn{1}{c}{Bias} & 
\multicolumn{1}{c}{RMSE} & \multicolumn{1}{c|}{Size} & \multicolumn{1}{c}{
Bias} & \multicolumn{1}{c}{RMSE} & \multicolumn{1}{c|}{Size} & 
\multicolumn{1}{c}{Bias} & \multicolumn{1}{c}{RMSE} & \multicolumn{1}{c}{Size
} \\ \hline
\multicolumn{11}{c}{\textit{high variance}: $\mathrm{var}\left( \beta_i
\right) = 0.25$} \\ \hline
\multirow{7}{*}{\begin{turn}{90} $\mathrm{E}\left(\beta_i\right) = 1.5$
\end{turn}} & \multicolumn{1}{r|}{100} & -0.0080 & 0.2262 & 
\multicolumn{1}{r|}{0.1922} & -0.0117 & 0.2297 & \multicolumn{1}{r|}{0.1940}
& -0.0030 & 0.2418 & 0.1800 \\ 
& \multicolumn{1}{r|}{1,000} & -0.0029 & 0.0663 & \multicolumn{1}{r|}{0.0936}
& -0.0015 & 0.0673 & \multicolumn{1}{r|}{0.0848} & -0.0037 & 0.0725 & 0.0804
\\ 
& \multicolumn{1}{r|}{2,000} & -0.0012 & 0.0431 & \multicolumn{1}{r|}{0.0688}
& -0.0015 & 0.0463 & \multicolumn{1}{r|}{0.0700} & -0.0021 & 0.0494 & 0.0656
\\ 
& \multicolumn{1}{r|}{5,000} & -0.0003 & 0.0263 & \multicolumn{1}{r|}{0.0566}
& -0.0009 & 0.0276 & \multicolumn{1}{r|}{0.0588} & -0.0013 & 0.0303 & 0.0622
\\ 
& \multicolumn{1}{r|}{10,000} & 0.0004 & 0.0183 & \multicolumn{1}{r|}{0.0530}
& -0.0001 & 0.0186 & \multicolumn{1}{r|}{0.0498} & -0.0003 & 0.0206 & 0.0492
\\ 
& \multicolumn{1}{r|}{100,000} & 0.0000 & 0.0056 & \multicolumn{1}{r|}{0.0434
} & 0.0000 & 0.0058 & \multicolumn{1}{r|}{0.0472} & 0.0000 & 0.0066 & 0.0514
\\ \hline
\multirow{7}{*}{\begin{turn}{90} $\mathrm{E}\left(\beta_i^2\right) = 2.5$
\end{turn}} & \multicolumn{1}{r|}{100} & -0.0627 & 0.9082 & 
\multicolumn{1}{r|}{0.3464} & -0.0826 & 0.8821 & \multicolumn{1}{r|}{0.3166}
& -0.0629 & 0.9459 & 0.3122 \\ 
& \multicolumn{1}{r|}{1,000} & -0.0300 & 0.2909 & \multicolumn{1}{r|}{0.1518}
& -0.0275 & 0.2837 & \multicolumn{1}{r|}{0.1382} & -0.0362 & 0.3112 & 0.1512
\\ 
& \multicolumn{1}{r|}{2,000} & -0.0160 & 0.1751 & \multicolumn{1}{r|}{0.0976}
& -0.0188 & 0.1868 & \multicolumn{1}{r|}{0.1074} & -0.0255 & 0.1900 & 0.1048
\\ 
& \multicolumn{1}{r|}{5,000} & -0.0067 & 0.0916 & \multicolumn{1}{r|}{0.0658}
& -0.0090 & 0.0993 & \multicolumn{1}{r|}{0.0710} & -0.0124 & 0.1091 & 0.0754
\\ 
& \multicolumn{1}{r|}{10,000} & -0.0015 & 0.0580 & \multicolumn{1}{r|}{0.0506
} & -0.0036 & 0.0609 & \multicolumn{1}{r|}{0.0530} & -0.0061 & 0.0704 & 
0.0566 \\ 
& \multicolumn{1}{r|}{100,000} & -0.0005 & 0.0179 & \multicolumn{1}{r|}{
0.0462} & -0.0005 & 0.0185 & \multicolumn{1}{r|}{0.0498} & -0.0011 & 0.0219
& 0.0542 \\ \hline
\multirow{7}{*}{\begin{turn}{90} $\mathrm{E}\left(\beta_i^3\right) = 4.5$
\end{turn}} & \multicolumn{1}{r|}{100} & -0.2511 & 2.3755 & 
\multicolumn{1}{r|}{0.3698} & -0.2990 & 2.3416 & \multicolumn{1}{r|}{0.3424}
& -0.2940 & 2.6179 & 0.3522 \\ 
& \multicolumn{1}{r|}{1,000} & -0.1155 & 0.7641 & \multicolumn{1}{r|}{0.1734}
& -0.1092 & 0.7613 & \multicolumn{1}{r|}{0.1606} & -0.1478 & 0.8856 & 0.1904
\\ 
& \multicolumn{1}{r|}{2,000} & -0.0667 & 0.4683 & \multicolumn{1}{r|}{0.1166}
& -0.0745 & 0.5058 & \multicolumn{1}{r|}{0.1234} & -0.1066 & 0.5485 & 0.1378
\\ 
& \multicolumn{1}{r|}{5,000} & -0.0290 & 0.2475 & \multicolumn{1}{r|}{0.0800}
& -0.0365 & 0.2696 & \multicolumn{1}{r|}{0.0788} & -0.0507 & 0.3178 & 0.0942
\\ 
& \multicolumn{1}{r|}{10,000} & -0.0099 & 0.1559 & \multicolumn{1}{r|}{0.0516
} & -0.0163 & 0.1699 & \multicolumn{1}{r|}{0.0602} & -0.0282 & 0.2088 & 
0.0660 \\ 
& \multicolumn{1}{r|}{100,000} & -0.0020 & 0.0488 & \multicolumn{1}{r|}{
0.0462} & -0.0023 & 0.0515 & \multicolumn{1}{r|}{0.0526} & -0.0052 & 0.0653
& 0.0520 \\ \hline
\multicolumn{11}{c}{\textit{low variance}: $\mathrm{var}\left( \beta_i
\right) = 0.15$} \\ \hline
\multirow{7}{*}{\begin{turn}{90} $\mathrm{E}\left(\beta_i\right) = 1.0915$
\end{turn}} & \multicolumn{1}{r|}{100} & 0.0165 & 0.1943 & 
\multicolumn{1}{r|}{0.1618} & 0.0089 & 0.1983 & \multicolumn{1}{r|}{0.1514}
& 0.0169 & 0.2112 & 0.1416 \\ 
& \multicolumn{1}{r|}{1,000} & 0.0045 & 0.0577 & \multicolumn{1}{r|}{0.0800}
& 0.0042 & 0.0584 & \multicolumn{1}{r|}{0.0702} & 0.0033 & 0.0655 & 0.0734
\\ 
& \multicolumn{1}{r|}{2,000} & 0.0019 & 0.0384 & \multicolumn{1}{r|}{0.0594}
& 0.0016 & 0.0410 & \multicolumn{1}{r|}{0.0698} & 0.0010 & 0.0452 & 0.0632
\\ 
& \multicolumn{1}{r|}{5,000} & 0.0008 & 0.0243 & \multicolumn{1}{r|}{0.0562}
& 0.0003 & 0.0250 & \multicolumn{1}{r|}{0.0540} & -0.0003 & 0.0283 & 0.0574
\\ 
& \multicolumn{1}{r|}{10,000} & 0.0007 & 0.0171 & \multicolumn{1}{r|}{0.0502}
& 0.0001 & 0.0175 & \multicolumn{1}{r|}{0.0476} & 0.0000 & 0.0194 & 0.0442
\\ 
& \multicolumn{1}{r|}{100,000} & 0.0000 & 0.0052 & \multicolumn{1}{r|}{0.0430
} & 0.0000 & 0.0054 & \multicolumn{1}{r|}{0.0476} & 0.0000 & 0.0062 & 0.0472
\\ \hline
\multirow{7}{*}{\begin{turn}{90} $\mathrm{E}\left(\beta_i^2\right) = 1.3413$
\end{turn}} & \multicolumn{1}{r|}{100} & -0.0121 & 0.5119 & 
\multicolumn{1}{r|}{0.2440} & -0.0280 & 0.5095 & \multicolumn{1}{r|}{0.2330}
& -0.0236 & 0.5724 & 0.2340 \\ 
& \multicolumn{1}{r|}{1,000} & -0.0061 & 0.1528 & \multicolumn{1}{r|}{0.1232}
& -0.0084 & 0.1566 & \multicolumn{1}{r|}{0.1126} & -0.0163 & 0.1776 & 0.1246
\\ 
& \multicolumn{1}{r|}{2,000} & -0.0072 & 0.0973 & \multicolumn{1}{r|}{0.0836}
& -0.0080 & 0.1053 & \multicolumn{1}{r|}{0.0922} & -0.0143 & 0.1154 & 0.0964
\\ 
& \multicolumn{1}{r|}{5,000} & -0.0037 & 0.0565 & \multicolumn{1}{r|}{0.0658}
& -0.0044 & 0.0603 & \multicolumn{1}{r|}{0.0698} & -0.0088 & 0.0699 & 0.0720
\\ 
& \multicolumn{1}{r|}{10,000} & -0.0018 & 0.0381 & \multicolumn{1}{r|}{0.0582
} & -0.0027 & 0.0401 & \multicolumn{1}{r|}{0.0590} & -0.0054 & 0.0476 & 
0.0618 \\ 
& \multicolumn{1}{r|}{100,000} & -0.0004 & 0.0119 & \multicolumn{1}{r|}{
0.0496} & -0.0005 & 0.0125 & \multicolumn{1}{r|}{0.0538} & -0.0009 & 0.0152
& 0.0506 \\ \hline
\multirow{7}{*}{\begin{turn}{90} $\mathrm{E}\left(\beta_i^3\right) = 1.7407$
\end{turn}} & \multicolumn{1}{r|}{100} & -0.0759 & 0.9761 & 
\multicolumn{1}{r|}{0.2806} & -0.0995 & 1.0052 & \multicolumn{1}{r|}{0.2672}
& -0.1277 & 1.2814 & 0.2718 \\ 
& \multicolumn{1}{r|}{1,000} & -0.0364 & 0.2925 & \multicolumn{1}{r|}{0.1486}
& -0.0396 & 0.3112 & \multicolumn{1}{r|}{0.1456} & -0.0687 & 0.3973 & 0.1720
\\ 
& \multicolumn{1}{r|}{2,000} & -0.0297 & 0.1927 & \multicolumn{1}{r|}{0.1040}
& -0.0310 & 0.2126 & \multicolumn{1}{r|}{0.1178} & -0.0526 & 0.2650 & 0.1324
\\ 
& \multicolumn{1}{r|}{5,000} & -0.0148 & 0.1141 & \multicolumn{1}{r|}{0.0798}
& -0.0168 & 0.1252 & \multicolumn{1}{r|}{0.0860} & -0.0301 & 0.1619 & 0.0964
\\ 
& \multicolumn{1}{r|}{10,000} & -0.0078 & 0.0771 & \multicolumn{1}{r|}{0.0654
} & -0.0097 & 0.0846 & \multicolumn{1}{r|}{0.0722} & -0.0188 & 0.1126 & 
0.0828 \\ 
& \multicolumn{1}{r|}{100,000} & -0.0013 & 0.0242 & \multicolumn{1}{r|}{
0.0478} & -0.0016 & 0.0262 & \multicolumn{1}{r|}{0.0554} & -0.0031 & 0.0360
& 0.0566 \\ \hline
\end{tabular}
}
\end{center}
\par
{\footnotesize \textit{Notes:} The data generating process is %
\eqref{eq:mc_dgp}. $S = 4$ is used. \textit{high variance} and \textit{low
variance} parametrization are described in \eqref{eq:mc_dgp_para}.
``Baseline'', ``Categorical $x$'' and ``Categorical $u$'' refer to DGP 1 to
3 as in Section \ref{subsec:dgp}. Generically, bias, RMSE and size are
calculated by $R^{-1}\sum_{r=1}^R \left( \hat{ \theta}^{(r)} - \theta_0
\right)$, $\sqrt{ R^{-1}\sum_{r= 1}^R \left( \hat{\theta}^{(r)} -\theta_0
\right)^2}$, and $R^{-1}\sum_{r=1}^R \mathbf{1}\left[ \left\vert \hat{\theta}%
^{(r)} - \theta_0 \right\vert / \hat{\sigma}_{\hat{\theta}}^{(r)} > \mathrm{%
cv}_{0.05} \right] $, respectively, for true parameter $\theta_0$, its
estimate $\hat{\theta} ^{(r)}$, the estimated standard error of $\hat{\theta}%
^{(r)}$, $\hat{\sigma} _{\hat{\theta}}^{(r)}$, and the critical value $%
\mathrm{cv}_{0.05} = \Phi^{-1}\left( 0.975 \right)$ across $R = 5,000$
replications, where $\Phi\left( \cdot \right)$ is the cumulative
distribution function of standard normal distribution. }
\end{table}

\subsection{Three Estimators of $\mathrm{E} \left( \protect\beta_i \right) $ 
\label{subsec:three_estimators}}

Table \ref{tab:mc_m1_compare} compares the finite sample performance of
three estimators of $\mathrm{E} \left( \beta_i \right)$ with the data
generating processes described in Section \ref{subsec:dgp}.

\begin{itemize}
\item {\ The OLS estimator $\hat{\mathbf{\phi }}$ studied in Section \ref%
{subsec:Estimation-of-gamma}}

\item {\ The GMM estimator of $\mathrm{E}\left( \beta _{i}\right) $ with
moment conditions given by \eqref{eq:mc_limit}. }

\item {\ $\widehat{\mathrm{E}\left( \beta _{i}\right) }=\hat{\pi}\hat{\beta}%
_{L}+\left( 1-\hat{\pi}\right) \hat{\beta}_{H}$, where $\hat{\pi},\hat{\beta}%
_{L},\hat{\beta}_{H}$ are the GMM estimators of $\pi ,\beta _{L},$ and $%
\beta _{H}$. }
\end{itemize}

According to Table \ref{tab:mc_m1_compare}, three estimators perform
comparably well in terms of bias and RMSE, whereas the OLS estimator, along
with the standard error from Theorem \ref{lem:gamma_est_consistency},
controls size well when $n$ is small.

\begin{table}[tbp]
\caption{Bias, RMSE and size of three estimators for $\mathrm{E}\left( 
\protect\beta_i \right)$ }
\label{tab:mc_m1_compare}
\begin{center}
{\small \ 
\begin{tabular}{rrrrrrrrrrr}
\hline
\multicolumn{2}{r|}{DGP} & \multicolumn{3}{c|}{Baseline} & 
\multicolumn{3}{c|}{Categorical $x$} & \multicolumn{3}{c}{Categorical $u$}
\\ \hline
\multicolumn{2}{r|}{Sample size $n$} & \multicolumn{1}{c}{Bias} & 
\multicolumn{1}{c}{RMSE} & \multicolumn{1}{c|}{Size} & \multicolumn{1}{c}{
Bias} & \multicolumn{1}{c}{RMSE} & \multicolumn{1}{c|}{Size} & 
\multicolumn{1}{c}{Bias} & \multicolumn{1}{c}{RMSE} & \multicolumn{1}{c}{Size
} \\ \hline
\multicolumn{11}{c}{\textit{high variance}: $\mathrm{E}\left(\beta_i\right)
= 1.5$, $\mathrm{var}\left( \beta_i \right) = 0.25$} \\ \hline
\multirow{7}{*}{\begin{turn}{90} OLS \end{turn}} & \multicolumn{1}{r|}{100}
& -0.0024 & 0.2035 & \multicolumn{1}{r|}{0.0966} & -0.0037 & 0.2035 & 
\multicolumn{1}{r|}{0.0858} & -0.0042 & 0.2268 & 0.0920 \\ 
& \multicolumn{1}{r|}{1,000} & -0.0017 & 0.0669 & \multicolumn{1}{r|}{0.0568}
& -0.0002 & 0.0657 & \multicolumn{1}{r|}{0.0540} & -0.0019 & 0.0738 & 0.0540
\\ 
& \multicolumn{1}{r|}{2,000} & -0.0008 & 0.0463 & \multicolumn{1}{r|}{0.0512}
& -0.0015 & 0.0475 & \multicolumn{1}{r|}{0.0534} & -0.0010 & 0.0523 & 0.0522
\\ 
& \multicolumn{1}{r|}{5,000} & -0.0004 & 0.0301 & \multicolumn{1}{r|}{0.0540}
& -0.0008 & 0.0300 & \multicolumn{1}{r|}{0.0546} & -0.0007 & 0.0335 & 0.0560
\\ 
& \multicolumn{1}{r|}{10,000} & 0.0002 & 0.0214 & \multicolumn{1}{r|}{0.0508}
& 0.0000 & 0.0212 & \multicolumn{1}{r|}{0.0510} & 0.0000 & 0.0229 & 0.0456
\\ 
& \multicolumn{1}{r|}{100,000} & -0.0001 & 0.0066 & \multicolumn{1}{r|}{
0.0472} & 0.0000 & 0.0066 & \multicolumn{1}{r|}{0.0460} & 0.0000 & 0.0075 & 
0.0506 \\ \hline
\multirow{7}{*}{\begin{turn}{90} GMM \end{turn}} & \multicolumn{1}{r|}{100}
& -0.0080 & 0.2262 & \multicolumn{1}{r|}{0.1922} & -0.0117 & 0.2297 & 
\multicolumn{1}{r|}{0.1940} & -0.0030 & 0.2418 & 0.1800 \\ 
& \multicolumn{1}{r|}{1,000} & -0.0029 & 0.0663 & \multicolumn{1}{r|}{0.0936}
& -0.0015 & 0.0673 & \multicolumn{1}{r|}{0.0848} & -0.0037 & 0.0725 & 0.0804
\\ 
& \multicolumn{1}{r|}{2,000} & -0.0012 & 0.0431 & \multicolumn{1}{r|}{0.0688}
& -0.0015 & 0.0463 & \multicolumn{1}{r|}{0.0700} & -0.0021 & 0.0494 & 0.0656
\\ 
& \multicolumn{1}{r|}{5,000} & -0.0003 & 0.0263 & \multicolumn{1}{r|}{0.0566}
& -0.0009 & 0.0276 & \multicolumn{1}{r|}{0.0588} & -0.0013 & 0.0303 & 0.0622
\\ 
& \multicolumn{1}{r|}{10,000} & 0.0004 & 0.0183 & \multicolumn{1}{r|}{0.0530}
& -0.0001 & 0.0186 & \multicolumn{1}{r|}{0.0498} & -0.0003 & 0.0206 & 0.0492
\\ 
& \multicolumn{1}{r|}{100,000} & 0.0000 & 0.0056 & \multicolumn{1}{r|}{0.0434
} & 0.0000 & 0.0058 & \multicolumn{1}{r|}{0.0472} & 0.0000 & 0.0066 & 0.0514
\\ \hline
\multirow{7}{*}{\begin{turn}{90} $\hat{\pi}\hat{\beta}_L +
(1-\hat{\pi})\hat{\beta}_H$ \end{turn}} & \multicolumn{1}{r|}{100} & -0.0087
& 0.2922 & \multicolumn{1}{r|}{0.1961} & -0.1232 & 0.2347 & 
\multicolumn{1}{r|}{0.1809} & -0.0037 & 0.2947 & 0.1894 \\ 
& \multicolumn{1}{r|}{1,000} & -0.0012 & 0.0648 & \multicolumn{1}{r|}{0.0709}
& -0.0237 & 0.0783 & \multicolumn{1}{r|}{0.0665} & -0.0023 & 0.0713 & 0.0652
\\ 
& \multicolumn{1}{r|}{2,000} & -0.0004 & 0.0410 & \multicolumn{1}{r|}{0.0556}
& -0.0140 & 0.0537 & \multicolumn{1}{r|}{0.0597} & -0.0015 & 0.0479 & 0.0558
\\ 
& \multicolumn{1}{r|}{5,000} & 0.0000 & 0.0259 & \multicolumn{1}{r|}{0.0536}
& -0.0063 & 0.0296 & \multicolumn{1}{r|}{0.0546} & -0.0011 & 0.0299 & 0.0590
\\ 
& \multicolumn{1}{r|}{10,000} & 0.0004 & 0.0183 & \multicolumn{1}{r|}{0.0526}
& -0.0035 & 0.0205 & \multicolumn{1}{r|}{0.0496} & -0.0003 & 0.0205 & 0.0488
\\ 
& \multicolumn{1}{r|}{100,000} & 0.0000 & 0.0056 & \multicolumn{1}{r|}{0.0436
} & -0.0006 & 0.0062 & \multicolumn{1}{r|}{0.0472} & 0.0000 & 0.0066 & 0.0514
\\ \hline
\multicolumn{11}{c}{\textit{low variance}: $\mathrm{E}\left(\beta_i\right) =
1.0915$, $\mathrm{var}\left( \beta_i \right) = 0.15$} \\ \hline
\multirow{7}{*}{\begin{turn}{90} OLS \end{turn}} & \multicolumn{1}{r|}{100}
& -0.0006 & 0.1829 & \multicolumn{1}{r|}{0.0810} & -0.0023 & 0.1855 & 
\multicolumn{1}{r|}{0.0766} & -0.0025 & 0.2094 & 0.0828 \\ 
& \multicolumn{1}{r|}{1,000} & -0.0005 & 0.0597 & \multicolumn{1}{r|}{0.0610}
& 0.0005 & 0.0590 & \multicolumn{1}{r|}{0.0478} & -0.0006 & 0.0670 & 0.0542
\\ 
& \multicolumn{1}{r|}{2,000} & -0.0002 & 0.0408 & \multicolumn{1}{r|}{0.0516}
& -0.0007 & 0.0427 & \multicolumn{1}{r|}{0.0606} & -0.0004 & 0.0475 & 0.0544
\\ 
& \multicolumn{1}{r|}{5,000} & -0.0002 & 0.0264 & \multicolumn{1}{r|}{0.0530}
& -0.0006 & 0.0266 & \multicolumn{1}{r|}{0.0480} & -0.0005 & 0.0302 & 0.0538
\\ 
& \multicolumn{1}{r|}{10,000} & 0.0000 & 0.0189 & \multicolumn{1}{r|}{0.0546}
& -0.0002 & 0.0188 & \multicolumn{1}{r|}{0.0486} & -0.0002 & 0.0208 & 0.0482
\\ 
& \multicolumn{1}{r|}{100,000} & -0.0001 & 0.0059 & \multicolumn{1}{r|}{
0.0474} & 0.0000 & 0.0059 & \multicolumn{1}{r|}{0.0494} & 0.0000 & 0.0068 & 
0.0508 \\ \hline
\multirow{7}{*}{\begin{turn}{90} GMM \end{turn}} & \multicolumn{1}{r|}{100}
& -0.0121 & 0.5119 & \multicolumn{1}{r|}{0.2440} & -0.0280 & 0.5095 & 
\multicolumn{1}{r|}{0.2330} & -0.0236 & 0.5724 & 0.2340 \\ 
& \multicolumn{1}{r|}{1,000} & -0.0061 & 0.1528 & \multicolumn{1}{r|}{0.1232}
& -0.0084 & 0.1566 & \multicolumn{1}{r|}{0.1126} & -0.0163 & 0.1776 & 0.1246
\\ 
& \multicolumn{1}{r|}{2,000} & -0.0072 & 0.0973 & \multicolumn{1}{r|}{0.0836}
& -0.0080 & 0.1053 & \multicolumn{1}{r|}{0.0922} & -0.0143 & 0.1154 & 0.0964
\\ 
& \multicolumn{1}{r|}{5,000} & -0.0037 & 0.0565 & \multicolumn{1}{r|}{0.0658}
& -0.0044 & 0.0603 & \multicolumn{1}{r|}{0.0698} & -0.0088 & 0.0699 & 0.0720
\\ 
& \multicolumn{1}{r|}{10,000} & -0.0018 & 0.0381 & \multicolumn{1}{r|}{0.0582
} & -0.0027 & 0.0401 & \multicolumn{1}{r|}{0.0590} & -0.0054 & 0.0476 & 
0.0618 \\ 
& \multicolumn{1}{r|}{100,000} & -0.0004 & 0.0119 & \multicolumn{1}{r|}{
0.0496} & -0.0005 & 0.0125 & \multicolumn{1}{r|}{0.0538} & -0.0009 & 0.0152
& 0.0506 \\ \hline
\multirow{7}{*}{\begin{turn}{90} { $\hat{\pi}\hat{\beta}_L +
(1-\hat{\pi})\hat{\beta}_H$} \end{turn}} & \multicolumn{1}{r|}{100} & 0.0166
& 0.2392 & \multicolumn{1}{r|}{0.1496} & 0.0063 & 0.2342 & 
\multicolumn{1}{r|}{0.1412} & 0.0182 & 0.2432 & 0.1586 \\ 
& \multicolumn{1}{r|}{1,000} & 0.0078 & 0.0621 & \multicolumn{1}{r|}{0.0827}
& 0.0068 & 0.0615 & \multicolumn{1}{r|}{0.0677} & 0.0064 & 0.0674 & 0.0693
\\ 
& \multicolumn{1}{r|}{2,000} & 0.0024 & 0.0388 & \multicolumn{1}{r|}{0.0559}
& 0.0021 & 0.0414 & \multicolumn{1}{r|}{0.0672} & 0.0019 & 0.0454 & 0.0627
\\ 
& \multicolumn{1}{r|}{5,000} & 0.0009 & 0.0241 & \multicolumn{1}{r|}{0.0554}
& 0.0003 & 0.0247 & \multicolumn{1}{r|}{0.0524} & 0.0001 & 0.0282 & 0.0548
\\ 
& \multicolumn{1}{r|}{10,000} & 0.0007 & 0.0170 & \multicolumn{1}{r|}{0.0502}
& 0.0002 & 0.0174 & \multicolumn{1}{r|}{0.0478} & 0.0003 & 0.0193 & 0.0438
\\ 
& \multicolumn{1}{r|}{100,000} & 0.0000 & 0.0052 & \multicolumn{1}{r|}{0.0430
} & 0.0000 & 0.0054 & \multicolumn{1}{r|}{0.0480} & 0.0004 & 0.0063 & 0.0494
\\ \hline
\end{tabular}
}
\end{center}
\par
{\footnotesize \textit{Notes:} The data generating process is %
\eqref{eq:mc_dgp}. \textit{high variance} and \textit{low variance}
parametrization are described in \eqref{eq:mc_dgp_para}. ``Baseline'',
``Categorical $x$'' and ``Categorical $u$'' refer to DGP 1 to 3 as in
Section \ref{subsec:dgp}. Generically, bias, RMSE and size are calculated by 
$R^{-1}\sum_{r=1}^R \left( \hat{ \theta}^{(r)} - \theta_0 \right)$, $\sqrt{
R^{-1}\sum_{r= 1}^R \left( \hat{\theta}^{(r)} -\theta_0 \right)^2}$, and $%
R^{-1}\sum_{r=1}^R \mathbf{1}\left[ \left\vert \hat{\theta}^{(r)} - \theta_0
\right\vert / \hat{\sigma}_{\hat{\theta}}^{(r)} > \mathrm{cv}_{0.05} \right] 
$, respectively, for true parameter $\theta_0$, its estimate $\hat{\theta}
^{(r)}$, the estimated standard error of $\hat{\theta}^{(r)}$, $\hat{\sigma}
_{\hat{\theta}}^{(r)}$, and the critical value $\mathrm{cv}_{0.05} =
\Phi^{-1}\left( 0.975 \right)$ across $R = 5,000$ replications, where $%
\Phi\left( \cdot \right)$ is the cumulative distribution function of
standard normal distribution. }
\end{table}

\subsection{Experiments with higher $\mathrm{var}\left( \protect\beta_i
\right) $\label{subsec:mc_highvar}}

Following the data generating processes in Section \ref{subsec:dgp}, we
increase the variance of $\beta _{i}$ by considering the following two
parametrizations: 
\begin{equation}
\left( \pi ,\beta _{L},\beta _{H},\mathrm{E}\left( \beta _{i}\right) ,\func{%
var}\left( \beta _{i}\right) \right) =%
\begin{cases}
(0.3,0.5,6,4.35,6.3525), \\ 
(0.3,0.5,10,7.15,18.9525).%
\end{cases}
\label{eq:mc_para}
\end{equation}%
Table \ref{tab:mc_highvar} presents the results, which show that using
larger values of $\mathrm{var}\left( \beta _{i}\right) $ improves the small
sample performance of the GMM estimators.

\begin{table}[tbp]
\caption{Bias, RMSE and size of the GMM estimator for distributional
parameters of $\protect\beta$ }
\label{tab:mc_highvar}
\begin{center}
{\small \ 
\begin{tabular}{rrlllllllll}
\hline
\multicolumn{2}{r|}{DGP} & \multicolumn{3}{c|}{Baseline} & 
\multicolumn{3}{c|}{Categorical $x$} & \multicolumn{3}{c}{Categorical $u$}
\\ \hline
\multicolumn{2}{r|}{Sample size $n$} & \multicolumn{1}{c}{Bias} & 
\multicolumn{1}{c}{RMSE} & \multicolumn{1}{c|}{Size} & \multicolumn{1}{c}{
Bias} & \multicolumn{1}{c}{RMSE} & \multicolumn{1}{c|}{Size} & 
\multicolumn{1}{c}{Bias} & \multicolumn{1}{c}{RMSE} & \multicolumn{1}{c}{Size
} \\ \hline
\multicolumn{11}{c}{$\mathrm{var}\left( \beta_i \right) = 6.35$} \\ \hline
\multirow{7}{*}{\begin{turn}{90} $\pi=0.3$ \end{turn}} & \multicolumn{1}{r|}{
100} & 0.0755 & 0.3014 & \multicolumn{1}{l|}{0.1885} & 0.0628 & 0.2829 & 
\multicolumn{1}{l|}{0.1601} & 0.0760 & 0.2967 & 0.1795 \\ 
& \multicolumn{1}{r|}{1,000} & -0.0113 & 0.1058 & \multicolumn{1}{l|}{0.1485}
& -0.0002 & 0.0882 & \multicolumn{1}{l|}{0.1406} & -0.0092 & 0.1043 & 0.1509
\\ 
& \multicolumn{1}{r|}{2,000} & -0.0103 & 0.0646 & \multicolumn{1}{l|}{0.1025}
& -0.0016 & 0.0495 & \multicolumn{1}{l|}{0.1072} & -0.0077 & 0.0598 & 0.1104
\\ 
& \multicolumn{1}{r|}{5,000} & -0.0026 & 0.0276 & \multicolumn{1}{l|}{0.0718}
& -0.0009 & 0.0197 & \multicolumn{1}{l|}{0.0726} & -0.0021 & 0.0245 & 0.0742
\\ 
& \multicolumn{1}{r|}{10,000} & -0.0008 & 0.0095 & \multicolumn{1}{l|}{0.0576
} & -0.0005 & 0.0093 & \multicolumn{1}{l|}{0.0608} & -0.0010 & 0.0099 & 
0.0588 \\ 
& \multicolumn{1}{r|}{100,000} & -0.0002 & 0.0027 & \multicolumn{1}{l|}{
0.0490} & -0.0001 & 0.0026 & \multicolumn{1}{l|}{0.0518} & -0.0002 & 0.0028
& 0.0504 \\ \hline
\multirow{7}{*}{\begin{turn}{90} $\beta_L = 0.5$ \end{turn}} & 
\multicolumn{1}{r|}{100} & 2.7277 & 3.5109 & \multicolumn{1}{l|}{0.2385} & 
2.3640 & 3.2861 & \multicolumn{1}{l|}{0.2207} & 2.6810 & 3.4783 & 0.2292 \\ 
& \multicolumn{1}{r|}{1,000} & 0.2951 & 1.1688 & \multicolumn{1}{l|}{0.2743}
& 0.1539 & 0.9017 & \multicolumn{1}{l|}{0.2521} & 0.2473 & 1.1016 & 0.2725
\\ 
& \multicolumn{1}{r|}{2,000} & 0.0933 & 0.6394 & \multicolumn{1}{l|}{0.1916}
& 0.0460 & 0.5158 & \multicolumn{1}{l|}{0.1988} & 0.0698 & 0.5904 & 0.1951
\\ 
& \multicolumn{1}{r|}{5,000} & 0.0159 & 0.2570 & \multicolumn{1}{l|}{0.1236}
& -0.0005 & 0.1786 & \multicolumn{1}{l|}{0.1306} & 0.0066 & 0.2080 & 0.1225
\\ 
& \multicolumn{1}{r|}{10,000} & 0.0009 & 0.0607 & \multicolumn{1}{l|}{0.0884}
& -0.0005 & 0.0504 & \multicolumn{1}{l|}{0.0998} & -0.0014 & 0.0585 & 0.0830
\\ 
& \multicolumn{1}{r|}{100,000} & 0.0000 & 0.0130 & \multicolumn{1}{l|}{0.0572
} & 0.0005 & 0.0135 & \multicolumn{1}{l|}{0.0630} & -0.0003 & 0.0148 & 0.0622
\\ \hline
\multirow{7}{*}{\begin{turn}{90} $\beta_H = 6$ \end{turn}} & 
\multicolumn{1}{r|}{100} & 0.1286 & 1.1700 & \multicolumn{1}{l|}{0.0978} & 
0.0482 & 1.1467 & \multicolumn{1}{l|}{0.1057} & 0.1395 & 1.3662 & 0.0970 \\ 
& \multicolumn{1}{r|}{1,000} & 0.0031 & 0.2840 & \multicolumn{1}{l|}{0.1320}
& 0.0062 & 0.2695 & \multicolumn{1}{l|}{0.1200} & 0.0043 & 0.3197 & 0.1382
\\ 
& \multicolumn{1}{r|}{2,000} & -0.0108 & 0.1392 & \multicolumn{1}{l|}{0.0982}
& 0.0007 & 0.1552 & \multicolumn{1}{l|}{0.1094} & -0.0108 & 0.1519 & 0.1088
\\ 
& \multicolumn{1}{r|}{5,000} & -0.0041 & 0.0621 & \multicolumn{1}{l|}{0.0746}
& -0.0024 & 0.0608 & \multicolumn{1}{l|}{0.0736} & -0.0054 & 0.0652 & 0.0794
\\ 
& \multicolumn{1}{r|}{10,000} & -0.0018 & 0.0340 & \multicolumn{1}{l|}{0.0550
} & -0.0012 & 0.0347 & \multicolumn{1}{l|}{0.0678} & -0.0034 & 0.0386 & 
0.0642 \\ 
& \multicolumn{1}{r|}{100,000} & -0.0003 & 0.0109 & \multicolumn{1}{l|}{
0.0530} & 0.0001 & 0.0107 & \multicolumn{1}{l|}{0.0518} & -0.0006 & 0.0125 & 
0.0588 \\ \hline
\multicolumn{11}{c}{$\mathrm{var}\left( \beta_i \right) = 18.95$} \\ \hline
\multirow{7}{*}{\begin{turn}{90} $\pi=0.3$ \end{turn}} & \multicolumn{1}{r|}{
100} & 0.0575 & 0.2896 & \multicolumn{1}{l|}{0.1761} & 0.0530 & 0.2762 & 
\multicolumn{1}{l|}{0.1524} & 0.0554 & 0.2889 & 0.1646 \\ 
& \multicolumn{1}{r|}{1,000} & -0.0136 & 0.1070 & \multicolumn{1}{l|}{0.1217}
& -0.0025 & 0.0892 & \multicolumn{1}{l|}{0.1306} & -0.0110 & 0.1024 & 0.1369
\\ 
& \multicolumn{1}{r|}{2,000} & -0.0101 & 0.0650 & \multicolumn{1}{l|}{0.0850}
& -0.0032 & 0.0488 & \multicolumn{1}{l|}{0.0969} & -0.0077 & 0.0610 & 0.0957
\\ 
& \multicolumn{1}{r|}{5,000} & -0.0027 & 0.0291 & \multicolumn{1}{l|}{0.0668}
& -0.0010 & 0.0217 & \multicolumn{1}{l|}{0.0625} & -0.0023 & 0.0247 & 0.0713
\\ 
& \multicolumn{1}{r|}{10,000} & -0.0009 & 0.0122 & \multicolumn{1}{l|}{0.0549
} & -0.0005 & 0.0097 & \multicolumn{1}{l|}{0.0600} & -0.0009 & 0.0100 & 
0.0570 \\ 
& \multicolumn{1}{r|}{100,000} & -0.0002 & 0.0025 & \multicolumn{1}{l|}{
0.0480} & -0.0001 & 0.0024 & \multicolumn{1}{l|}{0.0514} & -0.0002 & 0.0025
& 0.0484 \\ \hline
\multirow{7}{*}{\begin{turn}{90} $\beta_L = 0.5$ \end{turn}} & 
\multicolumn{1}{r|}{100} & 4.5691 & 5.9597 & \multicolumn{1}{l|}{0.2001} & 
4.0139 & 5.6053 & \multicolumn{1}{l|}{0.1750} & 4.4575 & 5.8827 & 0.1991 \\ 
& \multicolumn{1}{r|}{1,000} & 0.5104 & 1.8908 & \multicolumn{1}{l|}{0.2327}
& 0.2907 & 1.5133 & \multicolumn{1}{l|}{0.2146} & 0.4062 & 1.7517 & 0.2522
\\ 
& \multicolumn{1}{r|}{2,000} & 0.1678 & 1.0260 & \multicolumn{1}{l|}{0.1683}
& 0.0929 & 0.8581 & \multicolumn{1}{l|}{0.1714} & 0.1178 & 0.9144 & 0.1736
\\ 
& \multicolumn{1}{r|}{5,000} & 0.0292 & 0.3901 & \multicolumn{1}{l|}{0.1069}
& 0.0073 & 0.3040 & \multicolumn{1}{l|}{0.1095} & 0.0186 & 0.3400 & 0.1036
\\ 
& \multicolumn{1}{r|}{10,000} & 0.0058 & 0.1638 & \multicolumn{1}{l|}{0.0719}
& 0.0014 & 0.0899 & \multicolumn{1}{l|}{0.0834} & 0.0000 & 0.0919 & 0.0740
\\ 
& \multicolumn{1}{r|}{100,000} & 0.0000 & 0.0171 & \multicolumn{1}{l|}{0.0572
} & 0.0006 & 0.0171 & \multicolumn{1}{l|}{0.0614} & -0.0004 & 0.0185 & 0.0576
\\ \hline
\multirow{7}{*}{\begin{turn}{90} $\beta_H = 10$ \end{turn}} & 
\multicolumn{1}{r|}{100} & 0.0520 & 1.5471 & \multicolumn{1}{l|}{0.0926} & 
-0.0530 & 1.4858 & \multicolumn{1}{l|}{0.0944} & 0.0460 & 1.6879 & 0.0888 \\ 
& \multicolumn{1}{r|}{1,000} & -0.0078 & 0.4047 & \multicolumn{1}{l|}{0.1185}
& -0.0108 & 0.4158 & \multicolumn{1}{l|}{0.1020} & -0.0100 & 0.4178 & 0.1195
\\ 
& \multicolumn{1}{r|}{2,000} & -0.0093 & 0.2058 & \multicolumn{1}{l|}{0.0936}
& -0.0005 & 0.2067 & \multicolumn{1}{l|}{0.0975} & -0.0129 & 0.2546 & 0.0933
\\ 
& \multicolumn{1}{r|}{5,000} & -0.0037 & 0.0944 & \multicolumn{1}{l|}{0.0727}
& -0.0034 & 0.0922 & \multicolumn{1}{l|}{0.0709} & -0.0052 & 0.0871 & 0.0709
\\ 
& \multicolumn{1}{r|}{10,000} & -0.0023 & 0.0512 & \multicolumn{1}{l|}{0.0555
} & -0.0010 & 0.0504 & \multicolumn{1}{l|}{0.0684} & -0.0034 & 0.0529 & 
0.0580 \\ 
& \multicolumn{1}{r|}{100,000} & -0.0005 & 0.0160 & \multicolumn{1}{l|}{
0.0522} & 0.0002 & 0.0154 & \multicolumn{1}{l|}{0.0526} & -0.0007 & 0.0171 & 
0.0560 \\ \hline
\end{tabular}
}
\end{center}
\par
{\footnotesize \textit{Notes:} The data generating process is %
\eqref{eq:mc_dgp}. Parametrization are described in \eqref{eq:mc_para}. $S=4$
is used. ``Baseline'', ``Categorical $x$'' and ``Categorical $u$'' refer to
DGP 1 to 3 as in Section \ref{subsec:dgp}. Generically, bias, RMSE and size
are calculated by $R^{-1}\sum_{r=1}^R \left( \hat{ \theta}^{(r)} - \theta_0
\right)$, $\sqrt{ R^{-1}\sum_{r= 1}^R \left( \hat{\theta}^{(r)} -\theta_0
\right)^2}$, and $R^{-1}\sum_{r=1}^R \mathbf{1}\left[ \left\vert \hat{\theta}%
^{(r)} - \theta_0 \right\vert / \hat{\sigma}_{\hat{\theta}}^{(r)} > \mathrm{%
cv}_{0.05} \right] $, respectively, for true parameter $\theta_0$, its
estimate $\hat{\theta} ^{(r)}$, the estimated standard error of $\hat{\theta}%
^{(r)}$, $\hat{\sigma} _{\hat{\theta}}^{(r)}$, and the critical value $%
\mathrm{cv}_{0.05} = \Phi^{-1}\left( 0.975 \right)$ across $R = 5,000$
replications, where $\Phi\left( \cdot \right)$ is the cumulative
distribution function of standard normal distribution. }
\end{table}

\subsection{Experiments with three categories ($K=3)$\label{subsec:mc_k3}}

\subsubsection{Data generating processes\label{subsec:dgp_supp}}

We generate $y_{i}$ as 
\begin{equation}
y_{i}=\alpha +x_{i}\beta _{i}+z_{i1}\gamma _{1}+z_{i2}\gamma _{2}+u_{i},%
\text{ for }i=1,2,...,n,  \label{eq:mc_dgp_3}
\end{equation}%
with $\beta _{i}$ distributed as in \eqref{eq:category_dist} with $K=3$, 
\begin{equation*}
\beta_i= 
\begin{cases}
\beta_L, & \text { w.p. } \pi_L \\ 
\beta_M, & \text { w.p. } \pi_M \\ 
\beta_L, & \text { w.p. } 1- \pi_L - \pi_M,%
\end{cases}%
\end{equation*}
where w.p. denotes ``with probability''. The parameters take values $\left(
\pi_L, \pi_M, \beta _{L}, \beta _{M}, \beta_{H} \right) = \left( 0.3, 0.3,
1, 2, 3 \right)$. Corresponding, the moments of $\beta_i$ are $\left( 
\mathrm{E}\left( \beta_i \right), \mathrm{E}\left( \beta_i^2 \right), 
\mathrm{E}\left( \beta_i^3 \right), \mathrm{E}\left( \beta_i^4 \right), 
\mathrm{E}\left( \beta_i^5 \right) \right) = \left( 2.1, 5.1, 13.5, 37.5,
107.1 \right) $. The remaining parameters are set as $\alpha =0.25$, and $%
\mathbf{\gamma }=\left( 1,1\right) ^{\prime }$.

We first generate $\tilde{x}_{i}\sim \text{IID}\chi ^{2}(2)$, and then set $%
x_{i}=(\tilde{x}_{i}-2)/2$ so that $x_{i}$ has $0$ mean and unit variance.
The additional regressors, $z_{ij}$, for $j=1,2$ with homogeneous slopes are
generated as 
\begin{equation*}
z_{i1}=x_{i}+v_{i1}\text{ and }z_{i2}=z_{i1}+v_{i2},
\end{equation*}
with $v_{ij}\sim \text{IID }N\left( 0,1\right) $, for $j=1,2$. The error
term, $u_{i}$, is generated as $u_{i}=\sigma _{i}\varepsilon _{i}$, where $%
\sigma _{i}^{2}$ are generated as $0.5(1+\text{IID}\chi ^{2}(1))$, and $%
\varepsilon _{i}\sim \text{IID}N(0,1)$.

\subsubsection{Results}

Table \ref{tab:mc_k3} reports the bias, RMSE and size of the GMM estimator
for distributional parameters and moments of $\beta _{i}$. The results are
based on $5,000$ replications and $S=6$. The results show that even larger
sample sizes are needed for the GMM estimators (both the moments of $\beta
_{i}$ and its distributional parameters) to achieve reasonable finite sample
performance, since higher order of moments are involved.

In additional to the results of jointly estimating distributional parameters
and moments of $\beta _{i}$ by GMM, Table \ref{tab:mc_k3_m1} reports the
results of GMM estimation of moments of $\beta _{i}$ up to order 3 using the
moment conditions as in the $K=2$ case where $S=4$ in the left panel, and
the results of OLS estimation of $\mathbf{\phi }$ in the right panel. These
results show that we are still able to obtain accurate estimation of lower
order moments of $\beta _{i}$ when the fourth and fifth moments of $\beta
_{i}$ are not used, confirming the lower information content of the higher
order moments for estimation of the lower order moments of $\beta _{i}$.

\begin{table}[tbp]
\caption{Bias, RMSE and size of the GMM estimator for distributional
parameters and moments of $\protect\beta$ with $K = 3$ }
\label{tab:mc_k3}
\begin{center}
{\small \ 
\begin{tabular}{r|r|rrr|r|rrr}
\hline
\multicolumn{1}{l|}{} &  & \multicolumn{3}{c|}{Distribution of $\beta_i$} & 
& \multicolumn{3}{c}{Moments of $\beta_i$} \\ \hline
\multicolumn{1}{l|}{Sample size $n$} & \multicolumn{1}{c|}{} & Bias & RMSE & 
Size & \multicolumn{1}{c|}{} & Bias & RMSE & Size \\ \hline
100 & \multirow{8}{*}{\begin{turn}{90} $\pi_L=0.3$ \end{turn}} & -0.0405 & 
0.1910 & 0.1319 & \multirow{8}{*}{\begin{turn}{90} $\mathrm{E}(\beta_i)=2.1$
\end{turn}} & 0.1484 & 0.7471 & 0.6451 \\ 
1,000 &  & -0.0417 & 0.1633 & 0.1915 &  & -0.0711 & 0.5415 & 0.6128 \\ 
2,000 &  & -0.0383 & 0.1474 & 0.2354 &  & -0.1112 & 0.4408 & 0.5264 \\ 
5,000 &  & -0.0299 & 0.1186 & 0.3098 &  & -0.0904 & 0.3712 & 0.4034 \\ 
10,000 &  & -0.0209 & 0.0949 & 0.3371 &  & -0.0523 & 0.2740 & 0.2910 \\ 
100,000 &  & -0.0074 & 0.0314 & 0.2295 &  & -0.0026 & 0.0400 & 0.0678 \\ 
200,000 &  & -0.0050 & 0.0208 & 0.1917 &  & -0.0004 & 0.0202 & 0.0568 \\ 
\hline
100 & \multirow{8}{*}{\begin{turn}{90} $\beta_M = 0.3$ \end{turn}} & 0.2166
& 0.2995 & 0.0492 & \multirow{8}{*}{\begin{turn}{90}
$\mathrm{E}(\beta_i^2)=5.1$ \end{turn}} & 0.2841 & 2.8452 & 0.7223 \\ 
1,000 &  & 0.1404 & 0.2378 & 0.1364 &  & -0.6374 & 1.9507 & 0.6456 \\ 
2,000 &  & 0.1035 & 0.2117 & 0.1901 &  & -0.7163 & 1.7408 & 0.5472 \\ 
5,000 &  & 0.0615 & 0.1645 & 0.2381 &  & -0.5478 & 1.4628 & 0.4472 \\ 
10,000 &  & 0.0364 & 0.1292 & 0.2477 &  & -0.3391 & 1.1394 & 0.3432 \\ 
100,000 &  & 0.0013 & 0.0322 & 0.1305 &  & -0.0209 & 0.2300 & 0.0932 \\ 
200,000 &  & 0.0006 & 0.0185 & 0.1033 &  & -0.0046 & 0.1128 & 0.0620 \\ 
\hline
100 & \multirow{8}{*}{\begin{turn}{90} $\beta_L = 1$ \end{turn}} & 0.6881 & 
1.1994 & 0.1110 & \multirow{8}{*}{\begin{turn}{90}
$\mathrm{E}(\beta_i^3)=13.5$ \end{turn}} & 0.4897 & 10.0757 & 0.7189 \\ 
1,000 &  & 0.2588 & 0.7438 & 0.1994 &  & -2.7735 & 7.0573 & 0.6718 \\ 
2,000 &  & 0.1096 & 0.5372 & 0.2607 &  & -2.9100 & 6.3988 & 0.5894 \\ 
5,000 &  & 0.0205 & 0.4184 & 0.3426 &  & -2.1889 & 5.4307 & 0.5078 \\ 
10,000 &  & 0.0070 & 0.2733 & 0.3360 &  & -1.3454 & 4.3382 & 0.4042 \\ 
100,000 &  & -0.0064 & 0.0556 & 0.2213 &  & -0.0942 & 1.0263 & 0.1132 \\ 
200,000 &  & -0.0047 & 0.0320 & 0.1775 &  & -0.0236 & 0.5035 & 0.0738 \\ 
\hline
100 & \multirow{8}{*}{\begin{turn}{90} $\beta_M = 2$ \end{turn}} & 0.1249 & 
0.7256 & 0.0642 & \multirow{8}{*}{\begin{turn}{90}
$\mathrm{E}(\beta_i^4)=37.5$ \end{turn}} & 0.9092 & 35.1538 & 0.7235 \\ 
1,000 &  & -0.1190 & 0.6298 & 0.1531 &  & -10.1071 & 24.1521 & 0.6944 \\ 
2,000 &  & -0.1935 & 0.5762 & 0.2303 &  & -10.7108 & 21.5751 & 0.6268 \\ 
5,000 &  & -0.1662 & 0.4777 & 0.3670 &  & -8.2675 & 18.7735 & 0.5464 \\ 
10,000 &  & -0.1261 & 0.3703 & 0.4414 &  & -5.5310 & 15.4382 & 0.4406 \\ 
100,000 &  & -0.0326 & 0.1175 & 0.2681 &  & -0.4433 & 3.5927 & 0.1240 \\ 
200,000 &  & -0.0193 & 0.0682 & 0.2203 &  & -0.1114 & 1.6644 & 0.0810 \\ 
\hline
100 & \multirow{8}{*}{\begin{turn}{90} $\beta_H = 3$ \end{turn}} & 0.8514 & 
3.1645 & 0.1064 & \multirow{8}{*}{\begin{turn}{90}
$\mathrm{E}(\beta_i^5)=107.1$ \end{turn}} & 2.4059 & 121.1286 & 0.6989 \\ 
1,000 &  & 1.6632 & 4.5208 & 0.3124 &  & -34.0298 & 77.5508 & 0.7012 \\ 
2,000 &  & 1.7929 & 4.6701 & 0.4000 &  & -35.4018 & 69.5876 & 0.6424 \\ 
5,000 &  & 1.3425 & 4.0152 & 0.4539 &  & -27.3828 & 60.4373 & 0.5638 \\ 
10,000 &  & 0.9637 & 3.3831 & 0.4333 &  & -18.1022 & 50.3990 & 0.4590 \\ 
100,000 &  & 0.0474 & 0.8321 & 0.2046 &  & -1.5330 & 11.7796 & 0.1314 \\ 
200,000 &  & 0.0033 & 0.3237 & 0.1573 &  & -0.4226 & 5.9529 & 0.0812 \\ 
\hline
\end{tabular}%
}
\end{center}
\par
{\footnotesize \textit{Notes:} The data generating process is %
\eqref{eq:mc_dgp_3}. Generically, bias, RMSE and size are calculated by $%
R^{-1}\sum_{r=1}^R \left( \hat{ \theta}^{(r)} - \theta_0 \right)$, $\sqrt{
R^{-1}\sum_{r= 1}^R \left( \hat{\theta}^{(r)} -\theta_0 \right)^2}$, and $%
R^{-1}\sum_{r=1}^R \mathbf{1}\left[ \left\vert \hat{\theta}^{(r)} - \theta_0
\right\vert / \hat{\sigma}_{\hat{\theta}}^{(r)} > \mathrm{cv}_{0.05} \right] 
$, respectively, for true parameter $\theta_0$, its estimate $\hat{\theta}
^{(r)}$, the estimated standard error of $\hat{\theta}^{(r)}$, $\hat{\sigma}
_{\hat{\theta}}^{(r)}$, and the critical value $\mathrm{cv}_{0.05} =
\Phi^{-1}\left( 0.975 \right)$ across $R = 5,000$ replications, where $%
\Phi\left( \cdot \right)$ is the cumulative distribution function of
standard normal distribution. }
\end{table}

\begin{table}[tbp]
\caption{Bias, RMSE and size of estimation of $\protect\phi$ and moments of $%
\protect\beta_i$ (using $S = 4$) with $K = 3$ }
\label{tab:mc_k3_m1}
\begin{center}
{\small \ 
\begin{tabular}{r|r|rrr|r|rrr}
\hline
&  & \multicolumn{3}{c|}{Moments of $\beta_i$ ($S=4$)} &  & 
\multicolumn{3}{c}{OLS Estimate $hat{\phi}i$} \\ \hline
$n$ &  & \multicolumn{1}{c}{Bias} & \multicolumn{1}{c}{RMSE} & 
\multicolumn{1}{c|}{Size} &  & \multicolumn{1}{c}{Bias} & \multicolumn{1}{c}{
RMSE} & \multicolumn{1}{c}{Size} \\ \hline
100 & \multirow{8}{*}{\begin{turn}{90} $\mathrm{E}(\beta_i)=2.1$ \end{turn}}
& 0.0025 & 0.2867 & 0.2088 & \multirow{8}{*}{\begin{turn}{90}
$\mathrm{E}(\beta_i)=2.1$ \end{turn}} & -0.0031 & 0.2768 & 0.1042 \\ 
1,000 &  & -0.0006 & 0.0821 & 0.1008 &  & -0.0008 & 0.0939 & 0.0588 \\ 
2,000 &  & 0.0004 & 0.0537 & 0.0734 &  & 0.0000 & 0.0653 & 0.0550 \\ 
5,000 &  & 0.0004 & 0.0323 & 0.0610 &  & -0.0008 & 0.0422 & 0.0506 \\ 
10,000 &  & 0.0007 & 0.0224 & 0.0572 &  & -0.0001 & 0.0299 & 0.0510 \\ 
100,000 &  & 0.0000 & 0.0069 & 0.0454 &  & -0.0001 & 0.0093 & 0.0462 \\ 
200,000 &  & 0.0000 & 0.0050 & 0.0550 &  & 0.0000 & 0.0067 & 0.0498 \\ \hline
100 & \multirow{8}{*}{\begin{turn}{90} $\mathrm{E}(\beta_i^2)=5.1$
\end{turn}} & -0.1195 & 1.8290 & 0.3948 & \multirow{8}{*}{\begin{turn}{90}
$\gamma_1 = 1$ \end{turn}} & -0.0020 & 0.1817 & 0.0604 \\ 
1,000 &  & -0.0455 & 0.5965 & 0.1602 &  & 0.0000 & 0.0581 & 0.0474 \\ 
2,000 &  & -0.0196 & 0.3454 & 0.0902 &  & 0.0001 & 0.0409 & 0.0474 \\ 
5,000 &  & -0.0073 & 0.1630 & 0.0608 &  & -0.0001 & 0.0259 & 0.0494 \\ 
10,000 &  & -0.0004 & 0.1028 & 0.0544 &  & -0.0004 & 0.0183 & 0.0518 \\ 
100,000 &  & 0.0001 & 0.0311 & 0.0488 &  & -0.0001 & 0.0058 & 0.0490 \\ 
200,000 &  & -0.0002 & 0.0217 & 0.0492 &  & -0.0001 & 0.0041 & 0.0490 \\ 
\hline
100 & \multirow{8}{*}{\begin{turn}{90} $\mathrm{E}(\beta_i^3)=13.5$
\end{turn}} & -0.7404 & 6.7772 & 0.4396 & \multirow{8}{*}{\begin{turn}{90}
$\gamma_2 = 1$ \end{turn}} & 0.0011 & 0.1296 & 0.0672 \\ 
1,000 &  & -0.3116 & 2.2732 & 0.1964 &  & 0.0000 & 0.0414 & 0.0570 \\ 
2,000 &  & -0.1433 & 1.3285 & 0.1110 &  & 0.0000 & 0.0291 & 0.0478 \\ 
5,000 &  & -0.0524 & 0.6468 & 0.0702 &  & -0.0001 & 0.0183 & 0.0506 \\ 
10,000 &  & -0.0117 & 0.4052 & 0.0568 &  & 0.0002 & 0.0130 & 0.0526 \\ 
100,000 &  & 0.0001 & 0.1236 & 0.0528 &  & 0.0001 & 0.0041 & 0.0494 \\ 
200,000 &  & -0.0009 & 0.0850 & 0.0462 &  & 0.0000 & 0.0029 & 0.0542 \\ 
\hline
\end{tabular}
}
\end{center}
\par
{\footnotesize \textit{Notes:} The data generating process is %
\eqref{eq:mc_dgp_3}. Generically, bias, RMSE and size are calculated by $%
R^{-1}\sum_{r=1}^R \left( \hat{ \theta}^{(r)} - \theta_0 \right)$, $\sqrt{
R^{-1}\sum_{r= 1}^R \left( \hat{\theta}^{(r)} -\theta_0 \right)^2}$, and $%
R^{-1}\sum_{r=1}^R \mathbf{1}\left[ \left\vert \hat{\theta}^{(r)} - \theta_0
\right\vert / \hat{\sigma}_{\hat{\theta}}^{(r)} > \mathrm{cv}_{0.05} \right] 
$, respectively, for true parameter $\theta_0$, its estimate $\hat{\theta}
^{(r)}$, the estimated standard error of $\hat{\theta}^{(r)}$, $\hat{\sigma}
_{\hat{\theta}}^{(r)}$, and the critical value $\mathrm{cv}_{0.05} =
\Phi^{-1}\left( 0.975 \right)$ across $R = 5,000$ replications, where $%
\Phi\left( \cdot \right)$ is the cumulative distribution function of
standard normal distribution. }
\end{table}

\subsection{Experiments with idiosyncratic heterogeneity\label%
{subsec:mc_heterogeneity}}

In addition to the existing results, the following Monte Carlo experiment is
designed to examine the finite sample performance of the estimator under
different degrees of idiosyncratic heterogeneity. Following DGP 1 in Section %
\ref{subsec:dgp}, we generate $\tilde{x}_{i}\sim \text{IID}\chi ^{2}(2)$,
and then set $x_{i}=(\tilde{x}_{i}-2)/2$. The additional regressors, $z_{ij}$%
, for $j=1,2$ with homogeneous slopes are generated as 
\begin{equation*}
z_{i1}=x_{i}+v_{i1}\text{ and }z_{i2}=z_{i1}+v_{i2},
\end{equation*}
with $v_{ij}\sim \text{IID }N\left( 0,1\right) $, for $j=1,2$. The error
term, $u_{i}$, is generated as 
\begin{equation*}
u_{i} = 
\begin{cases}
\sigma _{i}\varepsilon _{i} + e_i & \text{if } i=1,2,\cdots, \lfloor
n^\alpha \rfloor \\ 
\sigma _{i}\varepsilon _{i} & \text{if } i=\lfloor n^\alpha \rfloor + 1,
\cdots, n%
\end{cases}%
\end{equation*}
where $\sigma _{i}^{2}$ are generated as $0.5(1+\text{IID}\chi ^{2}(1))$, $%
\varepsilon _{i}\sim \text{IID}N(0,1)$, and $e_i$ is the idiosyncratic
heterogeneity that is generated from the standard normal distribution and
then set to be fixed across Monte Carlo replications. Then in this case we
have 
\begin{equation*}
\left|n^{-1} \sum_{i=1}^n \mathrm{E}\left(u_i^2\right) - 1\right| = \left|
n^{-1} \sum_{i=1}^{\lfloor n^\alpha \rfloor} e_i^2 \right| \leq
n^{-1}\sum_{i=1}^{\lfloor n^\alpha \rfloor} \left\vert e_i^2 \right\vert
\leq \left( \max_{1\leq i \leq \lfloor n^\alpha \rfloor} \left\vert e_i^2
\right\vert \right) n^{\alpha - 1}.
\end{equation*}
Similar arguments can be made for $r = 3$.

Following the same parametrization as in Section \ref%
{sec:Monte-Carlo-Simulation}, we consider the degree of heterogeneity $%
\alpha = 0.25$, $0.4$, and $0.5$. The estimation results are reported in
Table \ref{tab:mc_gmm_hetero}. The results are similar to that of the
Baseline DGP as reported in Table \ref{tab:mc_S4}, which suggests that the
GMM estimator is robust to limited degrees of idiosyncratic heterogeneity.

\begin{table}[tbp]
\caption{Bias, RMSE and size of the GMM estimator for distributional
parameters of $\protect\beta$}
\label{tab:mc_gmm_hetero}
\begin{center}
{\small \ 
\begin{tabular}{rrrrrrrrrrr}
\hline
\multicolumn{2}{r|}{$\alpha$} & \multicolumn{3}{c|}{0.25} & 
\multicolumn{3}{c|}{0.40} & \multicolumn{3}{c}{0.50} \\ \hline
\multicolumn{2}{r|}{Sample size $n$} & \multicolumn{1}{c}{Bias} & 
\multicolumn{1}{c}{RMSE} & \multicolumn{1}{c|}{Size} & \multicolumn{1}{c}{
Bias} & \multicolumn{1}{c}{RMSE} & \multicolumn{1}{c|}{Size} & 
\multicolumn{1}{c}{Bias} & \multicolumn{1}{c}{RMSE} & \multicolumn{1}{c}{Size
} \\ \hline
\multicolumn{11}{c}{\textit{high variance}: $\mathrm{var}\left( \beta_i
\right) = 0.25$} \\ \hline
\multirow{6}{*}{\begin{turn}{90} $\pi=0.5$ \end{turn}} & \multicolumn{1}{r|}{
100} & 0.0292 & 0.2201 & \multicolumn{1}{r|}{0.1957} & 0.0293 & 0.2177 & 
\multicolumn{1}{r|}{0.1859} & 0.0297 & 0.2160 & 0.1609 \\ 
& \multicolumn{1}{r|}{1,000} & 0.0020 & 0.1273 & \multicolumn{1}{r|}{0.1943}
& 0.0039 & 0.1293 & \multicolumn{1}{r|}{0.2047} & 0.0037 & 0.1356 & 0.2150
\\ 
& \multicolumn{1}{r|}{2,000} & 0.0014 & 0.0879 & \multicolumn{1}{r|}{0.1585}
& 0.0003 & 0.0812 & \multicolumn{1}{r|}{0.1421} & 0.0020 & 0.0851 & 0.1455
\\ 
& \multicolumn{1}{r|}{5,000} & 0.0002 & 0.0440 & \multicolumn{1}{r|}{0.0980}
& 0.0010 & 0.0457 & \multicolumn{1}{r|}{0.0982} & -0.0003 & 0.0445 & 0.0946
\\ 
& \multicolumn{1}{r|}{10,000} & -0.0007 & 0.0301 & \multicolumn{1}{r|}{0.0764
} & 0.0003 & 0.0304 & \multicolumn{1}{r|}{0.0824} & -0.0001 & 0.0311 & 0.0910
\\ 
& \multicolumn{1}{r|}{100,000} & 0.0000 & 0.0098 & \multicolumn{1}{r|}{0.0610
} & 0.0000 & 0.0097 & \multicolumn{1}{r|}{0.0536} & -0.0002 & 0.0096 & 0.0556
\\ \hline
\multirow{6}{*}{\begin{turn}{90} $\beta_L = 1$ \end{turn}} & 
\multicolumn{1}{r|}{100} & 0.2027 & 0.5686 & \multicolumn{1}{r|}{0.1807} & 
0.1993 & 0.5706 & \multicolumn{1}{r|}{0.1738} & 0.2007 & 0.5662 & 0.1712 \\ 
& \multicolumn{1}{r|}{1,000} & 0.0104 & 0.1711 & \multicolumn{1}{r|}{0.2115}
& 0.0136 & 0.1750 & \multicolumn{1}{r|}{0.2156} & 0.0079 & 0.1827 & 0.2132
\\ 
& \multicolumn{1}{r|}{2,000} & 0.0094 & 0.1121 & \multicolumn{1}{r|}{0.1741}
& 0.0069 & 0.1025 & \multicolumn{1}{r|}{0.1529} & 0.0087 & 0.1109 & 0.1593
\\ 
& \multicolumn{1}{r|}{5,000} & 0.0040 & 0.0543 & \multicolumn{1}{r|}{0.1090}
& 0.0052 & 0.0557 & \multicolumn{1}{r|}{0.1136} & 0.0050 & 0.0546 & 0.1112
\\ 
& \multicolumn{1}{r|}{10,000} & 0.0023 & 0.0365 & \multicolumn{1}{r|}{0.0856}
& 0.0024 & 0.0365 & \multicolumn{1}{r|}{0.0882} & 0.0025 & 0.0367 & 0.0922
\\ 
& \multicolumn{1}{r|}{100,000} & 0.0004 & 0.0116 & \multicolumn{1}{r|}{0.0602
} & 0.0005 & 0.0115 & \multicolumn{1}{r|}{0.0604} & 0.0004 & 0.0115 & 0.0584
\\ \hline
\multirow{6}{*}{\begin{turn}{90} $\beta_H = 2$ \end{turn}} & 
\multicolumn{1}{r|}{100} & -0.1947 & 0.5616 & \multicolumn{1}{r|}{0.1307} & 
-0.1983 & 0.5545 & \multicolumn{1}{r|}{0.1421} & -0.2094 & 0.5510 & 0.1358
\\ 
& \multicolumn{1}{r|}{1,000} & -0.0096 & 0.1720 & \multicolumn{1}{r|}{0.1682}
& -0.0078 & 0.1729 & \multicolumn{1}{r|}{0.1710} & -0.0066 & 0.1802 & 0.1751
\\ 
& \multicolumn{1}{r|}{2,000} & -0.0060 & 0.1142 & \multicolumn{1}{r|}{0.1445}
& -0.0068 & 0.1066 & \multicolumn{1}{r|}{0.1523} & -0.0070 & 0.1060 & 0.1405
\\ 
& \multicolumn{1}{r|}{5,000} & -0.0047 & 0.0530 & \multicolumn{1}{r|}{0.1130}
& -0.0037 & 0.0545 & \multicolumn{1}{r|}{0.1110} & -0.0054 & 0.0559 & 0.1088
\\ 
& \multicolumn{1}{r|}{10,000} & -0.0031 & 0.0360 & \multicolumn{1}{r|}{0.0922
} & -0.0023 & 0.0370 & \multicolumn{1}{r|}{0.0826} & -0.0024 & 0.0372 & 
0.0896 \\ 
& \multicolumn{1}{r|}{100,000} & -0.0004 & 0.0116 & \multicolumn{1}{r|}{
0.0592} & -0.0003 & 0.0115 & \multicolumn{1}{r|}{0.0546} & -0.0005 & 0.0114
& 0.0600 \\ \hline
\multicolumn{11}{c}{\textit{low variance}: $\mathrm{var}\left( \beta_i
\right) = 0.15$} \\ \hline
\multirow{6}{*}{\begin{turn}{90} $\pi=0.3$ \end{turn}} & \multicolumn{1}{r|}{
100} & 0.2132 & 0.2951 & \multicolumn{1}{r|}{0.1851} & 0.2133 & 0.2912 & 
\multicolumn{1}{r|}{0.1797} & 0.2132 & 0.2945 & 0.1716 \\ 
& \multicolumn{1}{r|}{1,000} & 0.0133 & 0.1591 & \multicolumn{1}{r|}{0.1894}
& 0.0125 & 0.1613 & \multicolumn{1}{r|}{0.1872} & 0.0163 & 0.1637 & 0.1840
\\ 
& \multicolumn{1}{r|}{2,000} & -0.0051 & 0.1103 & \multicolumn{1}{r|}{0.1619}
& -0.0055 & 0.1048 & \multicolumn{1}{r|}{0.1553} & -0.0027 & 0.1083 & 0.1559
\\ 
& \multicolumn{1}{r|}{5,000} & -0.0046 & 0.0599 & \multicolumn{1}{r|}{0.1198}
& -0.0029 & 0.0607 & \multicolumn{1}{r|}{0.1070} & -0.0046 & 0.0620 & 0.1208
\\ 
& \multicolumn{1}{r|}{10,000} & -0.0038 & 0.0418 & \multicolumn{1}{r|}{0.0900
} & -0.0023 & 0.0418 & \multicolumn{1}{r|}{0.0932} & -0.0022 & 0.0423 & 
0.0930 \\ 
& \multicolumn{1}{r|}{100,000} & -0.0003 & 0.0132 & \multicolumn{1}{r|}{
0.0622} & -0.0003 & 0.0130 & \multicolumn{1}{r|}{0.0576} & -0.0004 & 0.0127
& 0.0532 \\ \hline
\multirow{6}{*}{\begin{turn}{90} $\beta_L = 0.5$ \end{turn}} & 
\multicolumn{1}{r|}{100} & 0.3935 & 0.6293 & \multicolumn{1}{r|}{0.1959} & 
0.3900 & 0.6353 & \multicolumn{1}{r|}{0.1853} & 0.3917 & 0.6236 & 0.1811 \\ 
& \multicolumn{1}{r|}{1,000} & 0.0310 & 0.2598 & \multicolumn{1}{r|}{0.1590}
& 0.0357 & 0.2634 & \multicolumn{1}{r|}{0.1589} & 0.0298 & 0.2653 & 0.1609
\\ 
& \multicolumn{1}{r|}{2,000} & 0.0025 & 0.1590 & \multicolumn{1}{r|}{0.1539}
& 0.0004 & 0.1478 & \multicolumn{1}{r|}{0.1274} & 0.0025 & 0.1565 & 0.1459
\\ 
& \multicolumn{1}{r|}{5,000} & -0.0008 & 0.0849 & \multicolumn{1}{r|}{0.1100}
& 0.0018 & 0.0849 & \multicolumn{1}{r|}{0.1122} & 0.0003 & 0.0854 & 0.1078
\\ 
& \multicolumn{1}{r|}{10,000} & -0.0001 & 0.0586 & \multicolumn{1}{r|}{0.0922
} & 0.0004 & 0.0586 & \multicolumn{1}{r|}{0.0958} & 0.0012 & 0.0576 & 0.0918
\\ 
& \multicolumn{1}{r|}{100,000} & 0.0005 & 0.0183 & \multicolumn{1}{r|}{0.0596
} & 0.0002 & 0.0181 & \multicolumn{1}{r|}{0.0582} & 0.0003 & 0.0177 & 0.0558
\\ \hline
\multirow{6}{*}{\begin{turn}{90} $\beta_H = 1.345$ \end{turn}} & 
\multicolumn{1}{r|}{100} & -0.0463 & 0.4194 & \multicolumn{1}{r|}{0.1128} & 
-0.0509 & 0.4224 & \multicolumn{1}{r|}{0.1147} & -0.0489 & 0.4386 & 0.1239
\\ 
& \multicolumn{1}{r|}{1,000} & -0.0097 & 0.1428 & \multicolumn{1}{r|}{0.1498}
& -0.0106 & 0.1427 & \multicolumn{1}{r|}{0.1523} & -0.0094 & 0.1486 & 0.1467
\\ 
& \multicolumn{1}{r|}{2,000} & -0.0107 & 0.0920 & \multicolumn{1}{r|}{0.1443}
& -0.0106 & 0.0917 & \multicolumn{1}{r|}{0.1439} & -0.0093 & 0.0915 & 0.1389
\\ 
& \multicolumn{1}{r|}{5,000} & -0.0065 & 0.0492 & \multicolumn{1}{r|}{0.1166}
& -0.0056 & 0.0500 & \multicolumn{1}{r|}{0.1092} & -0.0063 & 0.0532 & 0.1134
\\ 
& \multicolumn{1}{r|}{10,000} & -0.0045 & 0.0345 & \multicolumn{1}{r|}{0.0910
} & -0.0037 & 0.0344 & \multicolumn{1}{r|}{0.0902} & -0.0035 & 0.0344 & 
0.0900 \\ 
& \multicolumn{1}{r|}{100,000} & -0.0006 & 0.0108 & \multicolumn{1}{r|}{
0.0602} & -0.0004 & 0.0107 & \multicolumn{1}{r|}{0.0572} & -0.0005 & 0.0105
& 0.0560 \\ \hline
\end{tabular}
}
\end{center}
\par
{\footnotesize \textit{Notes:} The data generating process is %
\eqref{eq:mc_dgp_3}. \textit{high variance} and \textit{low variance}
parametrization are described in \eqref{eq:mc_dgp_para}. $\alpha$ is the
degree of heterogeneity as in Remark \ref{heterogeniety of moments}.
Generically, bias, RMSE and size are calculated by $R^{-1}\sum_{r=1}^R
\left( \hat{ \theta}^{(r)} - \theta_0 \right)$, $\sqrt{ R^{-1}\sum_{r= 1}^R
\left( \hat{\theta}^{(r)} -\theta_0 \right)^2}$, and $R^{-1}\sum_{r=1}^R 
\mathbf{1}\left[ \left\vert \hat{\theta}^{(r)} - \theta_0 \right\vert / \hat{%
\sigma}_{\hat{\theta}}^{(r)} > \mathrm{cv}_{0.05} \right] $, respectively,
for true parameter $\theta_0$, its estimate $\hat{\theta} ^{(r)}$, the
estimated standard error of $\hat{\theta}^{(r)}$, $\hat{\sigma} _{\hat{\theta%
}}^{(r)}$, and the critical value $\mathrm{cv}_{0.05} = \Phi^{-1}\left(
0.975 \right)$ across $R = 5,000$ replications, where $\Phi\left( \cdot
\right)$ is the cumulative distribution function of standard normal
distribution.}
\end{table}

\section{Additional empirical results\label{sec:empirical_supp}}

In this section, we provide additional results for the empirical
application. In addition to the quadratic in experience in Section \ref%
{sec:Empirical-Application}, we further consider the following quartic in
experience specification, 
\begin{equation}
\log \text{wage}_{i}=\alpha +\beta _{i}\text{edu}_{i}+\rho _{1}\text{exper}%
_{i}+\rho _{2}\text{exper}_{i}^{2}+\rho _{3}\text{exper}_{i}^{3}+\rho _{4}%
\text{exper}_{i}^{4}+\tilde{\mathbf{z}}_{i}^{\prime }\tilde{\mathbf{\gamma }}%
+u_{i},  \label{eq:ccrm_spec_quar}
\end{equation}%
where 
\begin{equation*}
\beta _{i}=%
\begin{cases}
b_{L} & \text{w.p. }\pi , \\ 
b_{H} & \text{w.p. }1-\pi .%
\end{cases}%
\end{equation*}%
Table \ref{tab:beta_est_quartic} and \ref{tab:gamma_est_quartic} report the
estimates of the distributional parameters of $\beta _{i}$ and the estimates
of $\mathbf{\gamma }$ with the specification \eqref{eq:ccrm_spec_quar}.

The estimates of parameter of interests with specification %
\eqref{eq:ccrm_spec_quar} are almost the same as that with quadratic in
experience specification (6.3), reported in Table 5. The qualitative
analysis and conclusion discussed in Section 6 remain robust to adding third
and fourth order powers of exper$_{i}$ in the regressions.

\begin{table}[p!]
\caption{Estimates of the distribution of the return to education with
specification \eqref{eq:ccrm_spec_quar} across two periods, 1973 - 75 and
2001 - 03, by years of education and gender}
\label{tab:beta_est_quartic}
\begin{center}
\begin{tabular}{rrrrrrrrr}
\hline
& \multicolumn{2}{c}{High School or Less} &  & \multicolumn{2}{c}{
Postsecondary Edu.} &  & \multicolumn{2}{c}{All} \\ 
\cline{2-3}\cline{5-6}\cline{8-9}
& 1973 - 75 & 2001 - 03 &  & 1973 - 75 & 2001 - 03 &  & 1973 - 75 & 2001 - 03
\\ \hline
& \multicolumn{8}{r}{Both Male and Female} \\ \hline
$\pi$ & 0.4841 & 0.5081 &  & 0.4281 & 0.3576 &  & 0.4689 & 0.3559 \\ 
& \textit{\small (5274.3)} & \textit{\small (0.0267)} &  & \textit{\small %
(0.0495)} & \textit{\small (0.0089)} &  & \textit{\small (0.0534)} & \textit%
{\small (0.0046)} \\ 
$\beta_L$ & 0.0617 & 0.0392 &  & 0.0627 & 0.0859 &  & 0.0567 & 0.0658 \\ 
& \textit{\small (5.9252)} & \textit{\small (0.0013)} &  & \textit{\small %
(0.0035)} & \textit{\small (0.0009)} &  & \textit{\small (0.0022)} & \textit%
{\small (0.0004)} \\ 
$\beta_H$ & 0.0628 & 0.0928 &  & 0.1108 & 0.1397 &  & 0.0938 & 0.1270 \\ 
& \textit{\small (5.5919)} & \textit{\small (0.0019)} &  & \textit{\small %
(0.0031)} & \textit{\small (0.0007)} &  & \textit{\small (0.0023)} & \textit%
{\small (0.0004)} \\ 
$\beta_H / \beta_L$ & 1.0177 & 2.3645 &  & 1.7675 & 1.6267 &  & 1.6533 & 
1.9299 \\ 
& \textit{\small (7.1413)} & \textit{\small (0.0400)} &  & \textit{\small %
(0.0629)} & \textit{\small (0.0111)} &  & \textit{\small (0.0305)} & \textit%
{\small (0.0076)} \\ 
$\mathrm{E}\left(\beta_i\right)$ & 0.0623 & 0.0656 &  & 0.0902 & 0.1205 &  & 
0.0764 & 0.1053 \\ 
$\mathrm{var}\left(\beta_i\right)$ & 0.0005 & 0.0268 &  & 0.0238 & 0.0258 & 
& 0.0185 & 0.0293 \\ 
$n$ & 77,899 & 216,136 &  & 33,733 & 295,683 &  & 111,632 & 511,819 \\ \hline
& \multicolumn{8}{r}{Male} \\ \hline
$\pi$ & 0.4835 & 0.4968 &  & 0.4478 & 0.3007 &  & 0.4856 & 0.3550 \\ 
& n/a & \textit{\small (0.0394)} &  & \textit{\small (0.0676)} & \textit%
{\small (0.0095)} &  & \textit{\small (0.0936)} & \textit{\small (0.0052)}
\\ 
$\beta_L$ & 0.0648 & 0.0419 &  & 0.0520 & 0.0733 &  & 0.0553 & 0.0581 \\ 
& n/a & \textit{\small (0.0019)} &  & \textit{\small (0.0047)} & \textit%
{\small (0.0012)} &  & \textit{\small (0.0033)} & \textit{\small (0.0005)}
\\ 
$\beta_H$ & 0.0651 & 0.0927 &  & 0.0988 & 0.1321 &  & 0.0875 & 0.1220 \\ 
& n/a & \textit{\small (0.0026)} &  & \textit{\small (0.0041)} & \textit%
{\small (0.0008)} &  & \textit{\small (0.0034)} & \textit{\small (0.0005)}
\\ 
$\beta_H / \beta_L$ & 1.0048 & 2.2143 &  & 1.9002 & 1.8015 &  & 1.5816 & 
2.1003 \\ 
& n/a & \textit{\small (0.0495)} &  & \textit{\small (0.1124)} & \textit%
{\small (0.0210)} &  & \textit{\small (0.0456)} & \textit{\small (0.0124)}
\\ 
$\mathrm{E}\left(\beta_i\right)$ & 0.0649 & 0.0675 &  & 0.0778 & 0.1144 &  & 
0.0719 & 0.0993 \\ 
$\mathrm{var}\left(\beta_i\right)$ & 0.0002 & 0.0254 &  & 0.0233 & 0.0269 & 
& 0.0161 & 0.0306 \\ 
$n$ & 44,299 & 116,129 &  & 20,851 & 144,138 &  & 65,150 & 260,267 \\ \hline
& \multicolumn{8}{r}{Female} \\ \hline
$\pi$ & 0.5000 & 0.5210 &  & 0.4512 & 0.3849 &  & 0.4733 & 0.3773 \\ 
& \textit{\small (0.5611)} & \textit{\small (0.0281)} &  & \textit{\small %
(0.0739)} & \textit{\small (0.0167)} &  & \textit{\small (0.0870)} & \textit%
{\small (0.0083)} \\ 
$\beta_L$ & 0.0453 & 0.0352 &  & 0.0804 & 0.0956 &  & 0.0644 & 0.0762 \\ 
& \textit{\small (0.0143)} & \textit{\small (0.0016)} &  & \textit{\small %
(0.0050)} & \textit{\small (0.0013)} &  & \textit{\small (0.0034)} & \textit%
{\small (0.0006)} \\ 
$\beta_H$ & 0.0724 & 0.0969 &  & 0.1307 & 0.1449 &  & 0.1032 & 0.1338 \\ 
& \textit{\small (0.0169)} & \textit{\small (0.0025)} &  & \textit{\small %
(0.0052)} & \textit{\small (0.0011)} &  & \textit{\small (0.0040)} & \textit%
{\small (0.0007)} \\ 
$\beta_H / \beta_L$ & 1.5994 & 2.7540 &  & 1.6252 & 1.5154 &  & 1.6012 & 
1.7564 \\ 
& \textit{\small (0.1537)} & \textit{\small (0.0666)} &  & \textit{\small %
(0.0551)} & \textit{\small (0.0125)} &  & \textit{\small (0.0323)} & \textit%
{\small (0.0084)} \\ 
$\mathrm{E}\left(\beta_i\right)$ & 0.0588 & 0.0648 &  & 0.1080 & 0.1260 &  & 
0.0848 & 0.1121 \\ 
$\mathrm{var}\left(\beta_i\right)$ & 0.0136 & 0.0308 &  & 0.0250 & 0.0240 & 
& 0.0193 & 0.0279 \\ 
$n$ & 33,600 & 100,007 &  & 12,882 & 151,545 &  & 46,482 & 251,552 \\ \hline
\end{tabular}%
\end{center}
\par
{\footnotesize \textit{Notes:} This table reports the estimates of the
distribution of $\beta_i$ with the quartic in experience specification %
\eqref{eq:ccrm_spec_quar}, using $S = 4$ order moments of $\text{edu}_i$.
``Postsecondary Edu.'' stands for the sub-sample with years of education
higher than 12 and ``High School or Less'' stands for those with years of
education less than or equal to 12. $\mathrm{s.d.}\left(\beta_{i}\right)$
corresponds to the square root of estimated $\mathrm{var}\left(\beta_{i}
\right)$. $n$ is the sample size. ``n/a" is inserted when the estimates show
homogeneity of $\beta_i$ and $\pi$ is not identified and cannot be
estimated. }
\end{table}

\begin{table}[p!]
\caption{Estimates of $\mathbf{\protect\gamma}$ associated with control
variables $\mathbf{z}_i$ with specification \eqref{eq:ccrm_spec_quar} across
two periods, 1973 - 75 and 2001 - 03, by years of education and gender,
which complements Table \protect\ref{tab:beta_est_quartic}}
\label{tab:gamma_est_quartic}
\begin{center}
{\small \ 
\begin{tabular}{rrrrrrlrr}
\hline
& \multicolumn{2}{c}{High School or Less} & \multicolumn{1}{l}{} & 
\multicolumn{2}{c}{Postsecondary Edu.} &  & \multicolumn{2}{c}{All} \\ 
\cline{2-3}\cline{5-6}\cline{8-9}
& \multicolumn{1}{c}{1973 - 75} & \multicolumn{1}{c}{2001 - 03} & 
\multicolumn{1}{l}{} & \multicolumn{1}{c}{1973 - 75} & \multicolumn{1}{c}{
2001 - 03} &  & \multicolumn{1}{c}{1973 - 75} & \multicolumn{1}{c}{2001 - 03}
\\ \hline
& \multicolumn{8}{c}{\textit{Both male and female}} \\ \hline
\texttt{exper.} & 0.0769 & 0.0526 &  & 0.0817 & 0.0763 &  & 0.0757 & 0.0603
\\ 
& \textit{(0.0015)} & \textit{(0.0009)} &  & \textit{(0.0029)} & \textit{\
(0.0012)} &  & \textit{(0.0013)} & \textit{(0.0007)} \\ 
$\mathtt{exper.}^2$ & -0.0040 & -0.0020 &  & -0.0045 & -0.0039 &  & -0.0038
& -0.0024 \\ 
& \textit{(0.0001)} & \textit{(0.0001)} &  & \textit{(0.0003)} & \textit{\
(0.0001)} &  & \textit{(0.0001)} & \textit{(0.0001)} \\ 
$\mathtt{exper.}^3$ ($\times 10^5$) & 9.2470 & 3.4329 &  & 11.2100 & 8.9370
&  & 8.3625 & 3.6521 \\ 
& \textit{(0.4146)} & \textit{(0.2882)} &  & \textit{(1.2538)} & \textit{\
(0.4460)} &  & \textit{(0.3677)} & \textit{(0.2412)} \\ 
$\mathtt{exper.}^4$ ($\times 10^5$) & -0.0768 & -0.0236 &  & -0.1074 & 
-0.0777 &  & -0.0654 & -0.0169 \\ 
& \textit{(0.0043)} & \textit{(0.0031)} &  & \textit{(0.0158)} & \textit{\
(0.0054)} &  & \textit{(0.0039)} & \textit{(0.0027)} \\ 
\texttt{marriage} & 0.0819 & 0.0700 &  & 0.0728 & 0.0674 &  & 0.0799 & 0.0718
\\ 
& \textit{(0.0037)} & \textit{(0.0020)} &  & \textit{(0.0060)} & \textit{\
(0.0020)} &  & \textit{(0.0031)} & \textit{(0.0014)} \\ 
\texttt{nonwhite} & -0.1052 & -0.0808 &  & -0.0486 & -0.0613 &  & -0.0855 & 
-0.0719 \\ 
& \textit{(0.0046)} & \textit{(0.0024)} &  & \textit{(0.0088)} & \textit{\
(0.0025)} &  & \textit{(0.0041)} & \textit{(0.0018)} \\ 
\texttt{gender} & 0.4146 & 0.2272 &  & 0.2933 & 0.2008 &  & 0.3854 & 0.2150
\\ 
& \textit{(0.0029)} & \textit{(0.0017)} &  & \textit{(0.0049)} & \textit{\
(0.0018)} &  & \textit{(0.0025)} & \textit{(0.0013)} \\ 
$n$ & 77,899 & 216,136 &  & 33,733 & 295,683 &  & 111,632 & 511,819 \\ \hline
& \multicolumn{8}{c}{\textit{Male}} \\ \hline
\texttt{exper.} & 0.0823 & 0.0620 &  & 0.0859 & 0.0780 &  & 0.0825 & 0.0664
\\ 
& \textit{(0.0020)} & \textit{(0.0012)} &  & \textit{(0.0040)} & \textit{\
(0.0018)} &  & \textit{(0.0017)} & \textit{(0.0010)} \\ 
$\mathtt{exper.}^2$ ($\times 10^2$) & -0.0039 & -0.0024 &  & -0.0041 & 
-0.0036 &  & -0.0037 & -0.0025 \\ 
& \textit{(0.0002)} & \textit{(0.0001)} &  & \textit{(0.0004)} & \textit{\
(0.0002)} &  & \textit{(0.0001)} & \textit{(0.0001)} \\ 
$\mathtt{exper.}^3$ ($\times 10^5$) & 8.2014 & 4.3686 &  & 9.2747 & 7.3170 & 
& 7.4306 & 3.6749 \\ 
& \textit{(0.5321)} & \textit{(0.3864)} &  & \textit{(1.7422)} & \textit{\
(0.6709)} &  & \textit{(0.4700)} & \textit{(0.3241)} \\ 
$\mathtt{exper.}^4$ ($\times 10^5$) & -0.0650 & -0.0314 &  & -0.0880 & 
-0.0582 &  & -0.0552 & -0.0161 \\ 
& \textit{(0.0054)} & \textit{(0.0042)} &  & \textit{(0.0223)} & \textit{\
(0.0081)} &  & \textit{(0.0049)} & \textit{(0.0036)} \\ 
\texttt{marriage} & 0.1493 & 0.1052 &  & 0.1310 & 0.1234 &  & 0.1421 & 0.1192
\\ 
& \textit{(0.0056)} & \textit{(0.0029)} &  & \textit{(0.0088)} & \textit{\
(0.0031)} &  & \textit{(0.0048)} & \textit{(0.0021)} \\ 
\texttt{nonwhite} & -0.1362 & -0.1191 &  & -0.1214 & -0.1040 &  & -0.1309 & 
-0.1136 \\ 
& \textit{(0.0064)} & \textit{(0.0035)} &  & \textit{(0.0126)} & \textit{\
(0.0039)} &  & \textit{(0.0057)} & \textit{(0.0027)} \\ 
$n$ & 44,299 & 116,129 &  & 20,851 & 144,138 &  & 65,150 & 260,267 \\ \hline
& \multicolumn{8}{c}{\textit{Female}} \\ \hline
\texttt{exper.} & 0.0713 & 0.0455 &  & 0.0911 & 0.0782 &  & 0.0729 & 0.0568
\\ 
& \textit{(0.0022)} & \textit{(0.0013)} &  & \textit{(0.0040)} & \textit{\
(0.0016)} &  & \textit{(0.0019)} & \textit{(0.0011)} \\ 
$\mathtt{exper.}^2$ ($\times 10^2$) & -0.0044 & -0.0018 &  & -0.0067 & 
-0.0045 &  & -0.0045 & -0.0025 \\ 
& \textit{(0.0002)} & \textit{(0.0001)} &  & \textit{(0.0004)} & \textit{\
(0.0002)} &  & \textit{(0.0002)} & \textit{(0.0001)} \\ 
$\mathtt{exper.}^3$ ($\times 10^5$) & 11.0325 & 3.4767 &  & 19.6859 & 11.2858
&  & 11.3406 & 4.4944 \\ 
& \textit{(0.6649)} & \textit{(0.4360)} &  & \textit{(1.7412)} & \textit{\
(0.5915)} &  & \textit{(0.6095)} & \textit{(0.3682)} \\ 
$\mathtt{exper.}^4$ ($\times 10^5$) & -0.0974 & -0.0264 &  & -0.1979 & 
-0.1046 &  & -0.0969 & -0.0272 \\ 
& \textit{(0.0071)} & \textit{(0.0048)} &  & \textit{(0.0216)} & \textit{\
(0.0071)} &  & \textit{(0.0066)} & \textit{(0.0042)} \\ 
\texttt{marriage} & -0.0078 & 0.0278 &  & -0.0175 & 0.0168 &  & -0.0082 & 
0.0234 \\ 
& \textit{(0.0048)} & \textit{(0.0028)} &  & \textit{(0.0080)} & \textit{\
(0.0026)} &  & \textit{(0.0041)} & \textit{(0.0020)} \\ 
\texttt{nonwhite} & -0.0714 & -0.0479 &  & 0.0276 & -0.0291 &  & -0.0356 & 
-0.0375 \\ 
& \textit{(0.0065)} & \textit{(0.0033)} &  & \textit{(0.0117)} & \textit{\
(0.0033)} &  & \textit{(0.0057)} & \textit{(0.0024)} \\ 
$n$ & 33,600 & 100,007 &  & 12,882 & 151,545 &  & 46,482 & 251,552 \\ \hline
\end{tabular}
}
\end{center}
\par
{\footnotesize \textit{Notes:} This table reports the estimates of $\mathbf{%
\ \gamma}$ in \eqref{eq:ccrm_spec_quar}. ``Postsecondary Edu.'' stands for
the sub-sample with years of education higher than 12 and ``High School or
Less'' stands for those with years of education less than or equal to 12.
The standard error of estimates of coefficients associated with control
variables are estimated based on Theorem \ref{lem:gamma_est_consistency} and
reported in parentheses. $n$ is the sample size.}
\end{table}

\section{Computational algorithm\label{sec:computation}}

In this section, we describe the computational procedure used for estimation
of $\mathbf{\gamma }$, moments of $\beta _{i}$, and distributional
parameters of $\beta _{i}$.

\begin{enumerate}
\item {\ Denote $\mathbf{w}_{i}=\left( x_{i},\mathbf{z}_{i}^{\prime }\right)
^{\prime }$. Compute the OLS estimator 
\begin{equation*}
\left( \widehat{\mathrm{E}\left( \beta _{i}\right) }^{(0)},\widehat{\mathbf{%
\gamma }}^{\prime }\right) ^{\prime }=\left( \frac{1}{n}\sum_{i=1}^{n}%
\mathbf{w}_{i}\mathbf{w}_{i}^{\prime }\right) ^{-1}\left( \frac{1}{n}%
\sum_{i=1}^{n}\mathbf{w}_{i}^{\prime }y_{i}\right) ,
\end{equation*}%
and $\hat{\tilde{y}}_{i}=y_{i}-\mathbf{z}_{i}^{\prime }\widehat{\gamma }$.}

\item {\ For $r=2,3,\cdots ,2K-1$, compute the sample version of the moment
conditions \eqref{eq:mc_limit_r} and \eqref{eq:mc_limit_2r} in the main
paper by replacing $\rho _{r,s}$ by $n^{-1}\sum_{i=1}^{n}\hat{\tilde{y}}%
_{i}^{r}x_{i}^{s}$, and solving for $\widehat{\mathrm{E}\left( \beta
_{i}^{r}\right) }^{(0)}$ and $\widehat{\sigma _{r}}^{(0)}$, recursively. }

\item {\ Use the initial estimates $\left\{ \widehat{\mathrm{E}\left( \beta
_{i}^{r}\right) }^{(0)}\right\} _{r=1}^{2K-1}$ and $\left\{ \widehat{\sigma
_{r}}^{(0)}\right\} _{r=2}^{2K-1}$ to construct the weighting matrix $\hat{%
\mathbf{A}}_{n}$ in \eqref{eq:efficient_weight_mat} and compute the GMM
estimators $\left\{ \widehat{\mathrm{E}\left( \beta _{i}^{r}\right) }%
^{(1)}\right\} _{r=1}^{2K-1}$ and $\left\{ \widehat{\sigma _{r}}%
^{(1)}\right\} _{r=2}^{2K-1}$ to compute the moments of $\beta _{i}$ and $%
\sigma _{r}$. Iterate the GMM estimation one more time with $\left\{ 
\widehat{\mathrm{E}\left( \beta _{i}^{r}\right) }^{(1)}\right\}
_{r=1}^{2K-1} $ and $\left\{ \widehat{\sigma _{r}}^{(1)}\right\}
_{r=2}^{2K-1}$ as initial estimates to obtain $\left\{ \widehat{\mathrm{E}%
\left( \beta _{i}^{r}\right) }\right\} _{r=1}^{2K-1}$ and $\left\{ \widehat{%
\sigma _{r}}\right\} _{r=2}^{2K-1}$. }

\item {\ Solve 
\begin{equation*}
\min_{\pi _{k},b_{k}}\left\{ \sum_{j=1}^{r}\left( \sum_{k=1}^{K}\pi
_{k}b_{k}^{r}-\widehat{\mathrm{E}\left( \beta _{i}^{r}\right) }\right)
^{2}\right\}
\end{equation*}%
to get the initial estimates, $\widehat{\mathbf{\theta }}^{(0)}=\left( 
\widehat{\mathbf{\pi }}^{(0)\prime },\widehat{\mathbf{b}}^{(0)\prime
}\right) ^{\prime }$. }

\item {\ Using $\widehat{\mathbf{\theta }}^{(0)}=\left( \widehat{\mathbf{\pi 
}}^{(0)\prime },\widehat{\mathbf{b}}^{(0)\prime }\right) ^{\prime }$
construct the weighting matrix $\hat{\mathbf{A}}_{n}$ and compute the GMM
estimator as $\widehat{\mathbf{\theta }}^{(1)}=\left( \widehat{\mathbf{\pi }}%
^{(1)\prime },\widehat{\mathbf{b}}^{(1)\prime }\right) ^{\prime }$ for $%
\mathbf{\theta }$. Iterate the GMM estimation one more time with $\widehat{%
\mathbf{\theta }}^{(1)}=\left( \widehat{\mathbf{\pi }}^{(1)\prime },\widehat{%
\mathbf{b}}^{(1)\prime }\right) ^{\prime }$ as initial estimates to obtain $%
\widehat{\mathbf{\theta }}=\left( \widehat{\mathbf{\pi }}^{\prime },\widehat{%
\mathbf{b}}^{\prime }\right) ^{\prime }$. In the setup of the optimization
problem for the optimization solver, imposing the constraint $%
b_{1}<b_{2}<\cdots <b_{K}$ is important to improve the numerical
performance, particularly when }$n$ is not sufficiently large (less than $%
5,000$).
\end{enumerate}

\end{document}